\documentclass[final]{svjour3}

\smartqed
\usepackage{graphicx}
\usepackage{mathptmx}

\usepackage{latexsym}
\usepackage{xspace}
\usepackage{amsfonts}
\usepackage{amssymb}
\usepackage{wasysym}
\usepackage{color}
\usepackage{cite}
\usepackage{url}

\usepackage{prooftree}
\usepackage{abbrevs}
\usepackage{macros}
\usepackage{ammacros}

\journalname{Journal of Automated Reasoning}
\begin{document}
\title{\hybrid\thanks{Felty was supported in part by the Natural
  Sciences and Engineering Research Council of Canada Discovery program.
  Momigliano was
  supported by EPSRC grant GR/M98555 
and partially by the MRG project (IST-2001-33149),
  funded by the EC under the FET proactive initiative on Global
  Computing.}} \subtitle{A Definitional Two-Level Approach to Reasoning
  with Higher-Order Abstract Syntax}

\titlerunning{\hybrid: reasoning with HOAS}  
\author{Amy Felty \and Alberto Momigliano}

\institute{Amy Felty \at School of Information Technology and
Engineering, University of Ottawa, Ottawa, Ontario K1N 6N5, Canada \\
\email{afelty@site.uottawa.ca}
\and 
Alberto Momigliano \at Laboratory for the Foundations of
  Computer Science, School of Informatics, University of Edinburgh,
  Edinburgh EH9 3JZ, Scotland\\
\email{amomigl1@inf.ed.ac.uk}
}

\date{\today. Received:?? / Accepted:??}

\maketitle

\begin{abstract}
  Combining higher-order abstract syntax and (co)-induction in a
  logical framework is well known to be problematic.  We describe the
  theory and the practice of a tool called \hybrid, within \HOL and
  Coq, which aims to address many of these difficulties.  It allows
  object logics to be represented using higher-order abstract syntax,
  and reasoned about using tactical theorem proving and principles of
  (co)induction.  Moreover, it is definitional, which guarantees
  consistency within a classical type theory.  The idea is to have a
  de Bruijn representation of $\lambda$-terms providing a definitional
  layer that allows the user to represent object languages using
  higher-order abstract syntax, while offering tools for reasoning
  about them at the higher level.  In this paper we describe how
  to use \hybrid in a multi-level reasoning fashion, similar in spirit
  to other systems such as \emph{Twelf} and \emph{Abella}. By
  explicitly referencing provability in a middle layer called a
  specification logic, we solve the problem of reasoning by
  (co)induction in the presence of non-stratifiable hypothetical
  judgments, which allow very elegant and succinct specifications of
  object logic inference rules.  We first demonstrate the method on a
  simple example, formally proving type soundness (subject reduction)
  for a fragment of a pure functional language, using a minimal
  intuitionistic logic as the specification logic.  We then prove an
  analogous result for a continuation-machine presentation of the
  operational semantics of the same language, encoded this time in an
  ordered linear logic that serves as the specification layer.  This
  example demonstrates the ease with which we can incorporate new
  specification logics, and also illustrates a significantly more
  complex object logic whose encoding is elegantly expressed using
  features of the new specification logic.
\end{abstract}

\section{Introduction}
\label{sec:intro}

\emph{Logical frameworks} provide general languages in which it is
possible to represent a wide variety of logics, programming languages,
and other formal systems.  They are designed to capture uniformities
of the deductive systems of these object logics and to provide
support for implementing and reasoning about them.  One application of
particular interest of such frameworks is the specification of
programming languages and the formalization of their semantics in view
of formal reasoning about important properties of these languages,
such as their soundness.  Programming languages that enjoy such
properties provide a solid basis for building software systems that
avoid a variety of harmful  defects, leading to 
systems that are significantly more reliable and trustworthy.

The mechanism by which object-logics are represented in a logical
framework has a paramount importance on the success of a
formalization. A naive choice of representation can seriously endanger
a project almost from the start, making it almost impossible to move
beyond the very first step of the developments of a case study (see
\cite{Melham94}, which barely goes beyond encoding the syntax of the
$\pi$-calculus).

\renewcommand{\gr}{\ensuremath{\mathit expr}}

Higher-Order Abstract Syntax (HOAS) is a representation technique used
in some logical frameworks.  Using HOAS, whose idea dates back to
Church~\cite{Church40}, binding constructs in an object logic are
encoded within the function space provided by a meta-language based on
a $\lambda$-calculus.  For example, consider encoding a simple
functional programming language such as Mini-ML \cite{Clement86} in a
typed meta-language, where object-level programs are represented as
meta-level terms of type $\gr$.  We can introduce a constant
$\ikw{fun}$ of type $(\underline{\gr} \rightarrow \gr)\rightarrow \gr$
to represent functions of one argument.  Using such a representation
allows us to delegate to the meta-language $\alpha$-conversion and
capture-avoiding substitution.  Further, object logic substitution can
be rendered as meta-level $\beta$-conversion.  However, experiments
such as the one reported in~\cite{Momigliano02lfm} suggest that the
full benefits of HOAS can be enjoyed only when the latter is paired
with support for \emph{hypothetical} and \emph{parametric}
judgments~\cite{MartinLof85,Harper93jacm,Pfenning99handbook}.  Such
judgments are used, for example, in the well-known encoding of
inference rules assigning simple types to Mini-ML programs.  Both the
encoding of programs and the encoding of the typing predicate
typically contain \emph{negative} occurrences of the type or predicate
being defined (\eg, the underlined occurrence of $\gr$ in the type of
$\ikw{fun}$ above).  
This rules out any naive approach to view those set-theoretically as
least fixed points~\cite{GunterwhyMLnot,Paulson94cade}  or
type-theoretically as inductive types, which employ {strict
  positivity} \cite{PaulinMohring93} to enforce strong
normalization.  As much as HOAS sounds appealing, it raises the
question(s): how are we going to reason about such encodings, in
particular are there induction and case analysis principles available?

Among the many proposals---that we will survey in
Section~\ref{sec:rel}---one solution that has emerged in the last
decade stands out: \emph{specification} and (inductive)
\emph{meta-reasoning} should be handled within a single system but at different \emph{levels}. The first
example of such a meta-logic was $\foldn$~\cite{McDowell01}, soon to
be followed by its successor, $\foldna$~\cite{Tiu04phd}.\footnote{This
  is by no way the end of the story; on the contrary, the development
  of these ambient logics is very much a work in progress: Tiu
  \cite{Tiu07} introduced the system LG$^{\omega}$ to get rid of the
  local signatures required by Linc's $\nabla$ quantifier. Even more recently Gacek,
  Miller \& Nadathur presented the logic ${\cal G}$ to ease reasoning
  on open terms and implemented it in the \emph{Abella system}
  \cite{gacek08lics,Abella,AbellaSOS}. However, as this overdue report
  of our approach describes with an undeniable tardiness a system that
  was developed before the aforementioned new contributions,
  we will take the liberty to refer to $\foldna$ as the ``canonical''
  two-level system. We will discuss new developments in more depth in
  Section~\ref{ssec:2lr}.}  They are both based on intuitionistic logic augmented with introduction and elimination rules
for \emph{defined} atoms (partial inductive definitions,
PIDs~\cite{Halnass91}), in particular \emph{definitional reflection}
(\emph{defL}), which provides support for case analysis. While
$\foldn$ has only induction on natural numbers as the primitive form
of inductive reasoning, the latter generalizes that to standard forms
of induction and co-induction \cite{MomiglianoT03}; $\foldna$ also
introduces the so-called ``nabla'' quantifier
$\nabla$~\cite{miller05tocl} to deal with parametric judgments.  This
quantifier accounts for the dual properties of eigenvariables, namely
\emph{freshness} (when viewed as constants introduced by the
quantifier right rule) and \emph{instantiability} as a consequence of
the left rule and case analysis. Consistency and viability of proof search are ensured by
cut-elimination~\cite{mcdowell00tcs,Tiu04phd}. Inside the
meta-language, a \emph{specification logic} (SL) is developed that is
in turn used to specify and (inductively) reason about the
\emph{object logic/language} (OL) under study.
This partition avoids the issue of inductive meta-reasoning in
the presence of negative occurrences in OL judgments, since
hypothetical judgments are intensionally read in terms of object-level
provability. The price to pay is coping with this additional layer where we
explicitly reference the latter.  Were we to work with only a bare
proof-checker, this price could be indeed deemed too high; however, if
we could rely on some form of automation such as tactical theorem
proving, the picture would be significantly different.

The first author has proposed in~\cite{Felty02} that, rather than
implementing an interactive theorem prover for such meta-logics from
scratch, they can be simulated within a modern proof assistant.
(Coq~\cite{bertot/casteran:2004} in that case.)  The correspondence is
roughly as follows: the ambient logic of the proof assistant in place
of the basic (logical) inference rules of $\foldn$, introduction and
elimination (inversion) rules of inductive types (definitions) in
place of the \emph{defR} and \emph{defL} rules of PIDs.\footnote{The
  \emph{defL} rule for PIDs may use full higher-order unification,
  while inversion in an inductive proof assistant typically generates
  equations that may or may not be further simplified, especially at
  higher-order types.}  Both approaches introduce a minimal sequent calculus
\cite{JoMinLog} as a SL, and a Prolog-like set of clauses for the
OL\@.  Nevertheless, in a traditional inductive setting, this is not
quite enough, as reasoning by inversion crucially depends on
simplifying in the presence of constructors. When such constructors
are non-inductive, which is typically the case with variable-binding
operators, this presents a serious problem.  The approach used in that
work was axiomatic: encode the HOAS signature with  a set of constants
and add a set of axioms stating the freeness and extensionality
properties of the constants.  With the critical use of those axioms,
it was shown that it is possible to replicate, in the well-understood
and interactive setting of Coq, the style of proofs typical of
$\foldn$.  In particular, subject reduction for Mini-ML is formalized
in~\cite{Felty02} following this style very closely; this means that
the theorem is proved immediately without any ``technical'' lemmas
required by the choice of encoding technique or results that may be
trivial but are intrinsically foreign to the mathematics of the
problem. Moreover, HOAS proofs of subject reduction typically do not
require weakening or substitutions lemmas, as they are implicit in the
higher-order nature of the encoding. However, this approach did not
offer any formal justification to the axiomatic approach and it is
better seen as a proof-of-concept more than foundational work.

The \hybrid tool~\cite{Ambler02} was developed around the same time:
it implements a higher-order meta-language within
\HOL~\cite{Nipkow-Paulson-Wenzel:2002} that provides a form of HOAS
for the user to represent OLs.  The user level is separated from the
infrastructure, in which HOAS is implemented \emph{definitionally} via
a de Bruijn style encoding.  Lemmas stating properties such as
freeness and extensionality of constructors are \emph{proved} and no
additional axioms are required.

\begin{figure}    \setlength{\unitlength}{4144sp}  \begingroup\makeatletter\ifx\SetFigFont\undefined
    \def\x#1#2#3#4#5#6#7\relax{\def\x{#1#2#3#4#5#6}}  \expandafter\x\fmtname xxxxxx\relax \def\y{splain}  \ifx\x\y   \gdef\SetFigFont#1#2#3{    \ifnum #1<17\tiny\else \ifnum #1<20\small\else \ifnum
    #1<24\normalsize\else \ifnum #1<29\large\else \ifnum
    #1<34\Large\else \ifnum #1<41\LARGE\else \huge\fi\fi\fi\fi\fi\fi
    \csname #3\endcsname}  \else \gdef\SetFigFont#1#2#3{\begingroup \count@#1\relax \ifnum
    25<\count@\count@25\fi
    \def\x{\endgroup\@setsize\SetFigFont{#2pt}}    \expandafter\x \csname \romannumeral\the\count@
    pt\expandafter\endcsname \csname @\romannumeral\the\count@
    pt\endcsname \csname #3\endcsname}  \fi \fi\endgroup
  \begin{picture}(4692,2010)(34,-1198) \thinlines
    {\color[rgb]{0,0,0}\put(1456,269){\oval(210,210)[bl]}
      \put(1456,509){\oval(210,210)[tl]}
      \put(2821,269){\oval(210,210)[br]}
      \put(2821,509){\oval(210,210)[tr]} \put(1456,164){\line( 1,
        0){1365}} \put(1456,614){\line( 1, 0){1365}}
      \put(1351,269){\line( 0, 1){240}} \put(2926,269){\line( 0,
        1){240}} }    {\color[rgb]{0,0,0}\put(1006,-181){\oval(210,210)[bl]} \put(1006,
      59){\oval(210,210)[tl]} \put(3271,-181){\oval(210,210)[br]}
      \put(3271, 59){\oval(210,210)[tr]} \put(1006,-286){\line( 1,
        0){2265}} \put(1006,164){\line( 1, 0){2265}}
      \put(901,-181){\line( 0, 1){240}} \put(3376,-181){\line( 0,
        1){240}} }    {\color[rgb]{0,0,0}\put(556,-631){\oval(210,210)[bl]}
      \put(556,-391){\oval(210,210)[tl]}
      \put(3721,-631){\oval(210,210)[br]}
      \put(3721,-391){\oval(210,210)[tr]} \put(556,-736){\line( 1,
        0){3165}} \put(556,-286){\line( 1, 0){3165}}
      \put(451,-631){\line( 0, 1){240}} \put(3826,-631){\line( 0,
        1){240}} }    {\color[rgb]{0,0,0}\put(151,-1081){\oval(210,210)[bl]}
      \put(151,-841){\oval(210,210)[tl]}
      \put(4216,-1081){\oval(210,210)[br]}
      \put(4216,-841){\oval(210,210)[tr]} \put(151,-1186){\line( 1,
        0){4065}} \put(151,-736){\line( 1, 0){4065}} \put(
      46,-1081){\line( 0, 1){240}} \put(4321,-1081){\line( 0, 1){240}}
    }    \put(1936,-556){\makebox(0,0)[lb]{\smash{\SetFigFont{10}{12.0}{rm}{\color[rgb]{0,0,0}\hybrid}        }}}
    \put(1601,-1006){\makebox(0,0)[lb]{\smash{\SetFigFont{10}{12.0}{rm}{\color[rgb]{0,0,0}Isabelle/HOL, Coq}        }}}
    \put(3376,504){\makebox(0,0)[lb]{\smash{\SetFigFont{10}{12.0}{rm}{\color[rgb]{0,0,0}Syntax:
            $\llFun{x}{E\ x}, \llrec{E\ x}\dots$ }        }}}
    \put(3376,279){\makebox(0,0)[lb]{\smash{\SetFigFont{10}{12.0}{rm}{\color[rgb]{0,0,0}Semantics:
            typing $E\hastype t$,\dots}        }}}

    \put(3601,299){\makebox(0,0)[lb]{\smash{\SetFigFont{10}{12.0}{rm}{\color[rgb]{0,0,0}          } }}} 

\put(3626,
    10){\makebox(0,0)[lb]{\smash{\SetFigFont{10}{12.0}{rm}{\color[rgb]{0,0,0}Sequent
            calculus: $\slvdn{\Gamma}{n}{G}$}        }}}
    \put(4051,-151){\makebox(0,0)[lb]{\smash{\SetFigFont{10}{12.0}{rm}{\color[rgb]{0,0,0}}        }}}
    \put(4000,-421){\makebox(0,0)[lb]{\smash{\SetFigFont{10}{12.0}{rm}{\color[rgb]{0,0,0}Meta-language:
            quasi}        }}}
    \put(4000,-601){\makebox(0,0)[lb]{\smash{\SetFigFont{10}{12.0}{rm}{\color[rgb]{0,0,0}
            datatype for
            a $\lambda$-calculus}        }}}
    \put(4500,-871){\makebox(0,0)[lb]{\smash{\SetFigFont{10}{12.0}{rm}{\color[rgb]{0,0,0}Ambient
            logic:}        }}}
  \put(4500,-1071){\makebox(0,0)[lb]{\smash{\SetFigFont{10}{12.0}{rm}{\color[rgb]{0,0,0}
            tactics/simplifier}        }}}
    \put(4500,-1271){\makebox(0,0)[lb]{\smash{\SetFigFont{10}{12.0}{rm}{\color[rgb]{0,0,0}
            (co)induction}        }}}
    \put(1800,344){\makebox(0,0)[lb]{\smash{\SetFigFont{10}{12.0}{rm}{\color[rgb]{0,0,0}Object
            logic}        }}}
    \put(1800,-106){\makebox(0,0)[lb]{\smash{\SetFigFont{10}{12.0}{rm}{\color[rgb]{0,0,0}Specification
            logic}        }}}
  \end{picture}
  \caption{Architecture of the \hybrid system}
  \label{fig:arch}
\end{figure}
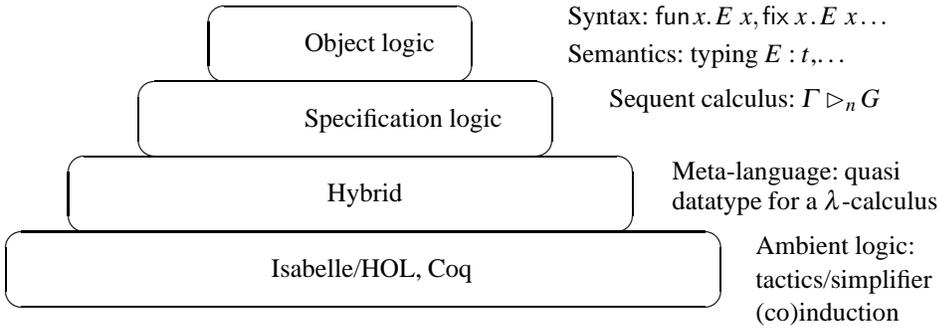

It was therefore natural to  combine the HOAS meta-language 
provided by \hybrid with Miller \& McDowell's two-level approach,
modified for inductive proof assistants.  We implement this combined
architecture in both \HOL and Coq, but we speculate that the approach
also works for other tactic-based inductive proof assistants,
such as PVS~\cite{cade92-pvs},  LEGO~\cite{LEGO} \etc.  We describe
mainly the \HOL version here,
though we compare it in some detail with the Coq
implementation.\footnote{We also compare it with a constructive version
  implemented in Coq~\cite{CapFel07}, which we describe in
  Section~\ref{ssec:h-v}.}  A graphical depiction of the architecture
is shown in Figure~\ref{fig:arch}.  We often refer to the \hybrid and
\HOL levels together as the meta-logic.  When we need to distinguish
the \HOL level on its own, we call it the meta-meta-logic.  When we
say \emph{two-level} reasoning, we are referring to the object and
specification levels, to emphasize that there are two separate
reasoning levels in addition to the meta-level.

Moreover, we suggest a further departure in design
(Section~\ref{ssec:variation}) from the original two-level
approach~\cite{McDowell01}: when possible, \ie, when the
structural properties of the meta-logic are coherent with the style of
encoding of the OL, we may reserve for the specification level only those
judgments that cannot be {adequately} encoded inductively and
leave the rest at the \HOL level.  We claim that this framework with
or without this variation has several advantages:

\begin{itemize}
\item The system is more trustworthy: freeness of constructors and,
  more importantly, extensionality properties at higher-order types are not
  assumed, but proved via the related properties of the
  infrastructure, as we show in Section~\ref{using} 
  (MC-Theorem~\ref{thm:clash}).
\item The mixing of meta-level and specification-level judgments makes proofs more
  easily mechanizable; more generally, there is a fruitful
  interaction between (co)-induction principles, meta-logic
  datatypes, classical reasoning, and hypothetical judgments, which
  lends itself to a good deal of automation.

\item We are not committed to a single monolithic SL,
  but we may adopt different ones (linear, relevant, bunched, \etc.)
  according to the properties of the OL we are encoding. The only
  requirement is consistency, to be established with a formalized
  cut-elimination argument.  We exemplify this methodology using
  non-commutative linear logic to reason about continuation machines
  (Section~\ref{sec:olli}).
\end{itemize}
 
Our architecture could also be seen as an approximation
of \emph{Twelf}~\cite{TwelfSP},
but it has a much lower mathematical overhead, simply consisting of a
small set of theories (modules) on top of a proof assistant. In a
sense, we could look at \hybrid\ as a way to ``represent'' Twelf's
meta-proofs in the well-understood setting of higher-order logic as
implemented in \HOL (or the calculus of (co)inductive constructions as
implemented in Coq). Note that by using a well-understood logic and
system, and working in a purely definitional way, we avoid the need to
justify \emph{consistency} by syntactic or semantic means.  For
example, we do not need to show a cut-elimination theorem for a new
logic as in~\cite{gacek08lics}, nor prove results such as strong
normalization of calculi of the $\cal M_{\omega}$ family \cite{S00} or
about the correctness of the totality checker behind
Twelf~\cite{SchurmannP03}. Hence our proofs are easier to trust, as
far as one trusts Isabelle/HOL and Coq.

Additionally, we can view our realization of the two-level approach as
a way of ``fast prototyping'' HOAS logical frameworks. We can quickly
implement and experiment with a potentially interesting SL; in
particular we can do meta-reasoning in the style of tactical theorem
proving in a way compatible with induction. For example, as we will
see in Section~\ref{sec:olli}, when experimenting with a different
logic, such as a sub-structural one, we do not need to develop all the
building blocks of a usable new framework, such as unification
algorithms, type inference or proof search, but we can rely on the ones
provided by the proof assistant.  The price to pay is, again, the
additional layer where we explicitly reference provability, requiring
a sort of meta-interpreter (the SL logic) to drive it.  This
indirectness can be alleviated, as we shall see, by defining
appropriate tactics, but this is intrinsic to the design choice of
relying on a general ambient logic (here \HOL or Coq,
in~\cite{McDowell01,Tiu04phd} some variation of $\foldna$). This
contrasts with the architecture proposed in~\cite{McCreight03}, where
the meta-meta-logic is itself sub-structural (linear in this case)
and, as such, explicitly tailored to the automation of a specific
framework.

We demonstrate the methodology by first formally verifying the
subject reduction property for the standard simply-typed call-by-value
$\lambda$-calculus, enriched with a recursion operator.  While this
property (and the calculus as well) has been criticized as too trivial to
be meaningful~\cite{poplmark2005}---and, to a degree, we  agree
with that---we feel that the familiarity of the set-up will ease the
understanding of the several layers of our architecture. Secondly we
tackle a more complex form of subject reduction, that of a
continuation machine, whose operational semantics is encoded
sub-structurally, namely in non-commutative linear logic.

\paragraph{Outline}

The paper is organized as follows: Section~\ref{sec:introh} recalls some
basic notions of \hybrid\ and its implementation in \HOL and Coq.
Section~\ref{using} shows how it can be used as a logical framework.  In
Section~\ref{sec:2lev} we introduce a two-level architecture and present the
first example SL and subject reduction proof,
while Section~\ref{sec:olli} introduces a sub-structural SL and uses it
for encoding continuation machines. We follow that up with an
extensive review and comparison of related work in Section~\ref{sec:rel},
and conclude in Section~\ref{sec:future}.
This paper is an archival documentation of \hybrid 0.1 (see
Section~\ref{ssec:h-v} for the terminology), extending previous joint
work with Simon Ambler and Roy
Crole~\cite{MomTP01,Ambler02,Momigliano02lfm,ACM03prim,Momigliano03fos},
Jeff Polakow~\cite{MomiglianoP03} and Venanzio
Capretta~\cite{CapFel07}.

 \begin{notation}[\HOL]  
   We use a pretty-printed version of \HOL concrete syntax. A type
   declaration has the form \mbox{$s\ \oftype \fsprems{t_1,\dots t_n}
     \fs t$}.  We stick to the usual logical symbols for \HOL
   connectives and quantifiers ($\neg$, $\land$, $\lor$, $\limp$,
   $\forall$, $\exists$).  Free variables (upper-case) are implicitly
   universally quantified (from the outside) as in logic programming.
   The sign $\idef$ (Isabelle meta-equality) is used for
   \emph{equality by definition}, and $\Forall$ for Isabelle universal
   meta-quantification.  A rule (a sequent) of the schematic form: $$
   \frac {H_1\dots H_n}{C}$$ is represented as $\prems{H_1;\dots; H_n}
   \Implies C$. A rule with discharged assumptions such as conjunction
   elimination is represented as $\prems{ P \land Q ; \prems{P;Q}
     \Implies R} \Implies R$.  The keyword \textbf{MC-Theorem (Lemma)}
   denotes a machine-checked theorem (lemma), while \emph{Inductive}
   introduces an inductive relation in \HOL, and
   \emph{datatype} introduces a new datatype.
   We freely use infix notations, without explicit declarations.  We
   have tried to use the same notation for mathematical and formalized
   judgments. The proof scripts underlying this paper are written in
   the so-called ``Isabelle old style'', \ie, they are exclusively in
   the tactical-style, \eg, sequences of
   commands. This was still fashionable and supported by \HOL 2005, as
   opposed to the now required ISAR \cite{ISAR} idioms of the new \HOL
   versions.  However, in the interest of time, intellectual honesty
   (and also consistency with Coq), we have decided to base the paper
   on the original code of the project, which had as a fundamental
   goal the \emph{automation} of two-level reasoning.  Naturally, some
   of the comments that we make about concrete features of the system,
   (as well as interactions with it) are by now relevant only to that
   version. When those happen to be obsolete, we will try to make this
   clear to the reader. We expect, however (and indeed we already are
   in the process, see Section~\ref{ssec:h-v}) to carry over this work
   to the current version of \HOL, possibly enhanced by the new
   features of the system.
 \end{notation}

\begin{notation}[Coq]
  We keep Coq's notation similar to \HOL's where possible.  We use the
  same syntax for type declarations, though of course the allowable
  types are different in the two languages.  We also use $\idef$ for
  equality by definition and $=$ for equality.  There is no
  distinction between a functional type arrow and logical implication
  in Coq, though we use both $\fs$ and $\Implies$ depending on the
  context.  In \HOL, there is a distinction between notation at the
  Isabelle meta-level and the HOL object-level, which we do not have
  in Coq. Whenever an \HOL formula has the form $\prems{H_1;\dots;
    H_n} \Implies C$, and we say that the Coq version is the same, we
  mean that the Coq version has the form $H_1 \Implies\cdots\Implies
  H_n \Implies C$, or equivalently $H_1 \fs\cdots\fs H_n \fs C$, where
  implication is right-associative as usual.
\end{notation}
Source files for the \HOL and Coq code can be found at
\url{hybrid.dsi.unimi.it/jar}~\cite{Hybrid}.

\section{Introducing \hybrid}
\label{sec:introh}

The description of the \hybrid layer of our architecture is taken fairly
directly from previous work, \viz \cite{Ambler02}.  Central to our
approach is the introduction of a binding operator that (1) allows a
direct expression of $\lambda$-abstraction, and (2) is \emph{defined}
in such a way that expanding its definition results in the conversion
of a term to its de Bruijn representation.  The basic idea is inspired
by the work of Gordon~\cite{Gor93}, and also appears in collaborative
work with Melham~\cite{Gordon96}.  Gordon introduces a
$\lambda$-calculus with constants where free and bound variables are
named by \emph{strings}; in particular, in a term of the form
$(\iapp{\iapp{\ikw{dLAM}}{\sv{v}}}{t})$, $\sv{v}$ is a string
representing a variable bound in $t$, and $\ikw{dLAM}$ is a function
of two arguments, which when applied, converts free occurrences of
$\sv{v}$ in $t$ to the appropriate de Bruijn indices and includes an
outer de Bruijn abstraction operator.  Not only does this approach
provide a good mechanism through which one may work with \emph{named}
bound variables under $\alpha$-renaming, but it can be used as a
meta-logic by building it into an \HOL type, say of \emph{proper
  terms}, from which other binding signatures can be defined, as
exemplified by Gillard's encoding of the object calculus~\cite{Gillard00}. As
in the logical framework tradition, every OL binding operator is
reduced to the $\lambda$-abstraction provided by the type of proper
terms.

Our approach
takes this a step further and exploits the built in \hoas\ which is
available in systems such as \HOL and Coq.
\hybrid's $\ikw{LAM}$ constructor is similar to Gordon's $\ikw{dLAM}$
except that $\ikw{LAM}$ is a \emph{binding} operator.  The syntax
$(\LAM{\ivar{\null}}{t})$ is actually notation for
$(\llambda{\ilam{\ivar{}}t})$, which makes explicit the use of {bound
variables in the meta-language} to represent {bound variables in the
OL}.  Thus the $\ivar{\null}$ in $(\LAM{\ivar{\null}}{t})$ is
a meta-variable (and not a string as in Gordon's approach).  

At the base level, we start with an inductive definition of de Bruijn
expressions, as Gordon does.
\begin{eqnarray*}
 \mathit{datatype} \ \expr  & = &       \Con{\con}
       \bnfalt      \Var{\var} 
       \bnfalt      \Bnd{\bnd} 
       \bnfalt      \App{\expr}{\expr} 
       \bnfalt      \Abs{\expr}
\end{eqnarray*}
In our setting, $\bnd$ and $\var$ are defined to be the natural
numbers, and $\con$ provides names for constants.  The latter type is
used to represent the constants of an OL, as each
OL introduces its own set of constants.

To illustrate the central ideas, we start with the $\lambda$-calculus
as an OL\@.  To avoid confusion with the meta-language
(\ie, $\lambda$-abstraction at the level of \HOL or Coq), we use
upper case letters for variables and a capital $\Lambda$ for
abstraction.  For example, consider the object-level term 
$T_0 = \Lambda V_1.(\Lambda V_2.V_1 V_2) V_1 V_3$.
The terms $T_G$ and $T_H$ below illustrate how this term is
represented using Gordon's approach and \hybrid, respectively.
$$\begin{array}{rcl}
\nt{G} & = & \ikw{dLAM}~\sv{v1}~
          (\ikw{dAPP}~
            (\ikw{dAPP}~
               (\ikw{dLAM}~\sv{v2}~
                  (\ikw{dAPP}~(\ikw{dVAR}~\sv{v1})\\
& &\qquad\qquad
                              (\ikw{dVAR}~\sv{v2})))~
               (\ikw{dVAR}~\sv{v1}))~
            (\ikw{dVAR}~\sv{v3})
          ) \\
\nt{H} & = & \ikw{LAM}~v_1.
          (
            ((\ikw{LAM}~v_2. (v_1\app v_2))\app v_1)\app
            \Var~3
          ) \\
\end{array}$$
In \hybrid we also choose to denote object-level free variables by terms
of the form $(\Var{i})$, though this is not essential.
In either case, the abstraction operator ($\ikw{dLAM}$ or $\ikw{LAM}$) is
defined, and expanding definitions in both $\nt{G}$ and $\nt{H}$
results in the same term, shown below using our de Bruijn notation.
$$\ikw{ABS}~(((\ikw{ABS}~(\underline{\ikw{BND}~1}\app\ikw{BND}~0))\app
\underline{\ikw{BND}~{0}})\app\Var~3)$$
In the above term all the variable occurrences bound by the first
$\ikw{ABS}$, which corresponds to the bound variable $V_1$ in the
object-level term, are underlined.  The $\ikw{lambda}$ operator is
central to this approach and its definition includes determining
correct indices.  We return to its definition in
Section~\ref{defh}.

\smallskip
In summary, \hybrid provides a form of HOAS where object-level:
\begin{itemize}
\item free variables correspond to \hybrid expressions of the
form $(\Var{i})$;
\item bound variables correspond to (bound) meta-variables;
\item abstractions $\olabs{\olvar{\null}}{\olexp}$ correspond to
  expressions $(\LAM{\ivar{\null}}{\hexp})$, defined as $
  (\llambda{\ilam{\ivar{\null}}{\hexp}})$;
\item  applications ${\olexp_1} \ {\olexp_2}$
  correspond to expressions $(\App{\hexp_1}{\hexp_2})$.
\end{itemize}

\subsection{Definition of \hybrid in \HOL}
\label{defh}

\hybrid consists of a small number of \HOL theories (actually two, for
a total of about 130 lines of definitions and 80 lemmas and theorems), which 
introduce the basic definition for de Bruijn expressions ($\expr$)
given above and provide operations and lemmas on them, building up to
those that hide the details of de Bruijn syntax and permit reasoning
on HOAS representations of OLs.  In this
section we outline the remaining definitions, and give some examples.
Note that our \HOL theories do not contain any axioms which
require external justification,\footnote{We will keep emphasizing this point: the
    package is a definitional extension of \HOL and could be brought
    back to HOL primitives, if one so wishes.} as in some
other approaches such as the Theory of Contexts~\cite{HonsellMS01}.

As mentioned, the operator $\llambda{} \oftype \fsprems{\expr \fs
\expr} \fs \expr$ is central to our approach, and we begin by
considering what is required to fill in its definition.  Clearly
$(\llambda{e})$ must expand to a term with \ikw{ABS} at the head.
Furthermore, we must define a function $f$ such that $(\llambda{e})$
is $(\Abs{(f~e)})$ where $f$ replaces occurrences of the bound
variable in $e$ with de Bruijn index $0$, taking care to increment the
index as it descends through inner abstractions.  In particular, we
will define a function \ikw{lbind} of two arguments such that
formally:
$$\llambda{e} \idef \Abs{(\lbind{0}{e})}$$
and $(\lbind{i}{e})$ replaces occurrences of the bound variable in $e$
with de Bruijn index $i$, where recursive calls on inner abstractions
will increase the index.  As an example, consider the function
$\ilam{\ivar{\null}}{\Abs{(\App{\Bnd{0}}{\ivar{\null}})}}$.  In this
case, application of \ikw{lbind} with argument index $0$ should result
in a level 1 expression:
\[
\lbind{0}{(\ilam{\ivar{\null}}{\Abs{(\App{\Bnd{0}}{\ivar{\null}})}})}
= \ldots = \Abs{(\App{\Bnd{0}}{\Bnd{1}})}
\]
and thus:
$$\llambda{(
  \ilam{\ivar{\null}}{\Abs{(\App{\Bnd{0}}{\ivar{\null}})}})} =
\Abs{(\Abs{(\App{\Bnd{0}}{\Bnd{1}})})}.$$

We define \ikw{lbind} as a total function operating on all functions
of type $(\expr \fs \expr)$, even \emph{exotic} ones
that do not encode $\lambda$-terms.  For example, we could have $e =
(\lambda x.\ikw{count}~x)$ where $(\ikw{count}~x)$ counts the total
number of variables and constants occurring in $x$.  Only functions
that behave \emph{uniformly} or \emph{parametrically} on their
arguments represent $\lambda$-terms.  We refer the reader to the
careful analysis of this phenomenon (in the context of Coq) given
in~\cite{DFHtlca95} and to Section~\ref{sec:rel} for more background.
We will return to this idea shortly and discuss how to rule out
non-uniform functions in our setting.  For now, we define \ikw{lbind}
so that it maps non-uniform subterms to a default value.  The subterms
we aim to rule out are those that do not satisfy the predicate
$\ordinary{} \oftype \fsprems{\expr\fs\expr} \fs \bool$, defined as
follows:\footnote{This definition is one of the points where the \HOL
  and Coq implementations of \hybrid diverge.  See
  Section~\ref{ssec:hcoq}.  }
 \[
 \begin{array}{rcl}
   \ordinary{e} & \idef
    & (\exists a.\;e=(\ilam{\ivar{}}{\Con{a}})\lor{} \\
   &&~e=(\ilam{\ivar{}}{\ivar{}})\lor{} \\
   &&~\exists n.\;e=(\ilam{\ivar{}}{\Var{n}})\lor{} \\
   &&~\exists j.\;e=(\ilam{\ivar{}}{\Bnd{j}})\lor{} \\
   &&~\exists f\ g.\;e=
        (\ilam{\ivar{}}{\App{\iapp{f}{\ivar{}}}{\iapp{g}{\ivar{}}}})\lor{} \\
   &&~\exists f.\;e=(\ilam{\ivar{}}{\Abs{(\iapp{f}{\ivar{}})}}))  \end{array}
 \]
\begin{sloppypar}
We do not define \ikw{lbind} directly, but instead define a relation
\mbox{$\lbnd{}{}{} \oftype \fsprems{\bnd,\expr\fs\expr,\expr} \fs \bool$} and
prove that this relation defines a function mapping the first two
arguments to the third.
\end{sloppypar}
\[
\begin{array}{rcl}
 \mathit{Inductive}\ \lbnd{}{}{} &\oftype& \fsprems{\bnd,\expr\fs\expr,\expr}
  \fs \bool\\ 
& \Implies &  \lbnd{i}{(\ilam{\ivar{\null}}{\Con{a}})}{(\Con{a})} \\
 & \Implies & \lbnd{i}{(\ilam{\ivar{\null}}{\ivar{\null}})}{(\Bnd{i})} \\
 & \Implies & \lbnd{i}{(\ilam{\ivar{\null}}{\Var{n}})}{(\Var{n})} \\
 & \Implies & \lbnd{i}{(\ilam{\ivar{\null}}{\Bnd{j}})}{(\Bnd{j})} \\
 \prems{\lbnd{i}{f}{s}; \lbnd{i}{g}{t}}
 & \Implies & 
 \lbnd{i}{(
 \ilam{\ivar{\null}}{\App{\iapp{f}{\ivar{\null}}}{\iapp{g}{\ivar{\null}}}}
 )}{(
 \App{s}{t}
 )}  \\
 \lbnd{(\Suc{i})}{f}{s}
 & \Implies & 
 \lbnd{i}{(\ilam{\ivar{\null}}{\Abs{(\iapp{f}{\ivar{\null}})}})}{(\Abs{s})}
\\  
\lnot (\ordinary{e})
 & \Implies & 
\lbnd{i}{e}{(\Bnd{0})} 
\end{array}
\]
In showing that this relation is a function, uniqueness is an easy
structural induction. Existence is proved using the following
abstraction induction principle.
\begin{goal}[abstraction\_induct]
 \[
 \begin{array}{l}
\prems{\Forall a.\; P~(\ilam{\ivar{}}{\Con{a}}) ;
       P~(\ilam{\ivar{}}{\ivar{}}) ;
       \Forall n.\; P~(\ilam{\ivar{}}{\Var{n}}) ;
       \Forall j.\; P~(\ilam{\ivar{}}{\Bnd{j}}) ; \\
      ~\Forall f\ g.\; \prems{P~f; P~g} \Implies
         P~(\ilam{\ivar{}}{\App{\iapp{f}{\ivar{}}}{\iapp{g}{\ivar{}}}});\\
      ~\Forall f.\; \prems{P~f} \Implies
         P~(\ilam{\ivar{}}{\Abs{(\iapp{f}{\ivar{}})}});\\
      ~\Forall f.\; \prems{\lnot\ordinary{f}} \Implies P~f} \Implies P~e
\end{array}
\]
\label{thm:abstractioninduct}
\end{goal}
\sloppy{
The proof of this induction principle is by measure induction 
($\Forall x.\, \prems{\forall y.\, \prems{f~y<f~x\limp P~y}\Implies P~x}
\Implies P~a$), where we instantiate $f$ with \ikw{rank} and
set $\iapp{\ikw{rank}}{e} \idef
\iapp{\ikw{size}}{(\iapp{e}{(\Var{0})})}$.
}

We now define $\lbind{}{}{} \oftype \fsprems{bnd,\expr\fs\expr} \fs
\expr$ as follows, thus completing the definition of \ikw{lambda}:
$$\lbind{i}{e}{} \idef \ihil{s}{\lbnd i e s}$$
where $\textit{THE}$ is Isabelle's notation for the
definite description operator $\iota$.
From these definitions, it is easy to prove a ``rewrite rule'' for
every de Bruijn constructor.  For example, the rule for
\ikw{ABS} is:
\begin{mclemma}[lbind\_ABS]
  \begin{eqnarray*}
  \lbind{i}{(\ilam{\ivar{\null}}{\Abs{(\iapp{e}{\ivar{\null}})}})}& = & 
  \Abs({\lbind{(\Suc{i})}{e}})    
  \end{eqnarray*}
\end{mclemma}
These rules are collected under the name \ikw{lbind\_simps}, and thus
can be used directly in simplification.

Ruling out non-uniform functions, which was mentioned before, will
turn out to be important for a variety of reasons.  For example, it is
necessary for proving that our encoding adequately represents the
$\lambda$-calculus.  To prove adequacy, we identify a subset of the
terms of type $\expr$ such that there is a bijection between this
subset and the $\lambda$-terms that we are encoding.  There are two
aspects we must consider in defining a predicate to identify this
subset.  First, recall that $(\Bnd{i})$ corresponds to a bound variable
in the \lcalc, and $(\Var{i})$ to a free variable; we refer to
\emph{bound} and \emph{free} \emph{indices} respectively. We call a
bound index $i$ \emph{dangling} if $i$ or less $\mathsf{ABS}$ labels
occur between the index $i$ and the root of the expression tree.  We
must rule out terms with dangling indices.  Second, in the presence of
the \ikw{LAM} constructor, we may have functions of type
$(\expr\fs\expr)$ that do not behave uniformly on their arguments.  We
must rule out such functions.  We define a predicate \ikw{proper},
which rules out dangling indices from terms of type $\expr$, and a
predicate \ikw{abstr}, which rules out dangling indices and exotic
terms in functions of type $(\expr\fs\expr)$.

To define \ikw{proper} we first define \ikw{level}.  Expression
$\hexp$ is said to be at \emph{level} $l\geq0$, if enclosing $\hexp$
inside $l$ $\mathsf{ABS}$ nodes ensures that the resulting expression
has no dangling indices.
 \[
 \begin{array}{rcl}
 \mathit{Inductive}\ \level{}{} &\oftype &\fsprems{\bnd,\expr} \fs \bool\\
   & \Implies & \level{i}{(\Con{a})} \\
   & \Implies & \level{i}{(\Var{n})} \\
   j < i  & \Implies & 
   \level{i}{(\Bnd{j})} \\
   \prems{\level{i}{s} ; \level{i}{t}} & \Implies &  
   \level{i}{(\App{s}{t})}\\ 
   \level{(\Suc{i})}{s} & \Implies & 
   \level{i}{(\Abs{s})}
 \end{array}
 \]
Then, $\proper{} \oftype \expr \fs \bool$ is defined simply as:
$$\proper{\hexp} \idef \level{0}{\hexp}.$$

To define \ikw{abstr}, we first define $\abst{}{} \oftype
\fsprems{\bnd,\expr\fs\expr} \fs\bool$  as follows:
 \[
 \begin{array}{rcl}
\mathit{Inductive}\ \abst{}{} &\oftype& \fsprems{\bnd,\expr\fs\expr} \fs
 \bool\\
   & \Implies & \abst{i}{(\ilam{\ivar{}}{\Con{a}})} \\
   & \Implies & \abst{i}{(\ilam{\ivar{}}{\ivar{}})} \\
   & \Implies & \abst{i}{(\ilam{\ivar{}}{\Var{n}})} \\
   j < i  & \Implies & 
   \abst{i}{(\ilam{\ivar{}}{\Bnd{j}})} \\
   \prems{\abst{i}{f} ; \abst{i}{g}} & \Implies &  
   \abst{i}{(\ilam{\ivar{}}{\App{\iapp{f}{\ivar{}}}{\iapp{g}{\ivar{}}}})}\\ 
   \abst{(\Suc{i})}{f} & \Implies & 
   \abst{i}{(\ilam{\ivar{}}{\Abs{(\iapp{f}{\ivar{}})}})}  \end{array}
 \]
Given  $\abstr{} \oftype \fsprems{\expr\fs\expr} \fs \bool$, 
we set:
$$\abstr{\hexp} \idef \abst{0}{\hexp}.$$
When an expression $\hexp$ of type $\expr\fs\expr$ satisfies this
predicate, we say it is an \emph{abstraction}.\footnote{This is akin
to the \ikw{valid} and \ikw{valid1} predicates present in weak
HOAS formalizations such as \cite{DFHtlca95} (discussed further in
Section~\ref{ssec:oth}), although this formalization has,
in our notation, the ``weaker'' type
$(\underline\var\fs\expr)\fs\bool$.}  In addition to being
important for adequacy, the notion of an abstraction is central to the
formulation of induction principles at the meta-level.\footnote{And
so much more for the purpose of this paper: it allows inversion on
inductive second-order predicates, simplification in presence of
higher-order functions, and, roughly said, it ensures the consistency of
those relations with the ambient logic.}

It's easy to prove the analogue of $\abst{}{}$ introduction rules in
terms of $ \abstr{}$, for example:
$$\abst{1}{f} \Implies
\abstr{(\ilam{\ivar{}}{\Abs{(\iapp{f}{\ivar{}})}})}$$ A simple, yet
important lemma is:
\begin{mclemma}[proper\_abst]
  \begin{eqnarray*}
    \proper t &\Implies & \abstr{(\ilam{\ivar{}}{t})}
  \end{eqnarray*}
\label{mclem:proper_abst}
\end{mclemma}
So any function is a legal abstraction if its body is a proper
expression. This strongly suggests that were we to turn the predicate
$\proper{}$ into a \emph{type} $\propert$, then any function with source type
$\propert\fs\propert$ would be de facto a legal
abstraction\footnote{This is indeed the case as we have shown in
  \cite{MMF07} and briefly comment on at the end of
  Section~\ref{sec:rel}.}.

It follows directly from the inductive definition of de Bruijn
expressions that the functions $\Con{}$, $\Var{}$, $\App{}$, and
$\Abs{}$ are injective, with disjoint images.  With the introduction
of \ikw{abstr}, we can now also prove the following fundamental
theorem:
\begin{goal}[abstr\_lam\_simp]
  \[
  \prems{\abstr e;\; \abstr f} \Implies (\LAM x {e\arg{x}} = \LAM y
  {f\arg{y}}) = (e = f)   \]
\label{thm:inj}
\end{goal}
which says that $\llambda{}$ is injective on the set of abstractions.
This follows directly from an analogous property of $\lbind{}{}$:
\begin{mclemma}[abst\_lbind\_simp\_lemma]
  \[
  \prems{\abst{i}{e};\; \abst{i}{f}} \Implies
  (\lbind{i}{e}=\lbind{i}{f}) = (e = f)
  \]
\end{mclemma}
This is proved by structural induction on the $\abst{}{}{}$ predicate
using simplification with  \ikw{lbind\_simps}.

Finally, it is possible to perform induction over the quasi-datatype
of proper terms.
\begin{goal}[proper\_VAR\_induct]
\label{thm:proper-induct}
 \[
 \begin{array}{l}
\prems{\proper{u}; \\
      ~\Forall a.\; P~(\Con{a}) ; \\
      ~\Forall n.\; P~(\Var{n}) ; \\
      ~\Forall s\ t.\; \prems{\proper{s}; \proper{t}; P~t} \Implies
         P~(\App{s}{t}) ; \\
      ~\Forall e.\; \prems{\abstr{e};\forall n.\;P~(e~(\Var{n}))} \Implies
         P~(\LAM{x}{e\arg{x}})} \Implies P~u
\end{array}
\]
\end{goal}
The proof is by induction on the size of $e$, and follows from the
following two lemmas.
\begin{mclemma}
$$\begin{array}{lll}
1. & 
   \level{(\Suc{i})}{e} \Implies \exists f.\;
   (\lbind{i}{f} = e) \land \abst{i}{f} 
   \qquad\qquad\qquad & \mbox{\textbf{\emph{(level\_lbind\_abst)}}} \\
2. & 
   \proper{(\Abs{e})} \Implies \exists f.\;
  (\LAM x {f\arg{x}} = \Abs{e}) \land \abstr{f}
   & \mbox{\textbf{\emph{(proper\_lambda\_abstr)}}}
\end{array}$$
\end{mclemma}
\begin{mclemma}[abstr\_size\_lbind]
$$
   \abstr{e} \Implies
   \iapp{\ikw{size}}{(\lbind{i}{e})} =
   \iapp{\ikw{size}}{(\iapp{e}{(\Var{n})})}
$$
\end{mclemma}
Note that MC-Theorem~\ref{thm:proper-induct} does not play
any active role in the two-level architecture, as induction will be
performed on the derivability of judgments.

\subsection{Remarks on \hybrid in Coq}
\label{ssec:hcoq}

In this section we comment briefly on the differences between the \HOL
and Coq implementations of \hybrid, which arise mainly from the
differences in the meta-languages.  \HOL implements a
polymorphic version of Church's higher-order (classical) logic plus
facilities for axiomatic classes and local reasoning in the form of
\emph{locales} \cite{Ballarin03}.  Coq
implements a constructive higher-order type theory,
but includes
libraries for reasoning classically, which we used in order to keep
the implementations as similar as possible.

Note that the definition of $\lbind{}{}$ uses \HOL's definite
description operator, which is not available in Coq.  The use of this
operator is the main reason for the differences in the two libraries.
In Coq, we instead use the description axiom available in Coq's
classical libraries:\footnote{In the Coq libraries, a dependent-type
version of this axiom is stated, from which the  version  here
follows directly.}
$$\begin{array}{l}
\forall A\ B \oftype \mathit{Type}.\;
\forall R \oftype \fsprems{A,B} \fs \mathit{Prop}.\; \\
~~(\forall x.\; \exists y.\;
  (R~x~y \land \forall y'.\; R~x~y' \Implies y=y')) \Implies
\exists f.\; \forall x.\;R~x~(f~x)
\end{array}$$
with $\ikw{lbnd}$ as relation $R$.
The Coq version of \hybrid is larger than the \HOL version,
mainly due to showing uniqueness for the $\lbnd{}{}{}$ relation.
We then eliminate the
existential quantifier in the description theorem to get a function
that serves as the Coq version of $\ikw{lbind}$.\footnote{Although this
elimination is not always justified, it is in our case since we define
the type $\expr$ to be a Coq ${Set}$.}

\newcommand{\Prop}{\textit{Prop}\xspace}
\newcommand{\Set}{\textit{Set}\xspace}

In more detail, if we consider the \HOL theory just described, the
operations and predicates $\ikw{ordinary}$, $\ikw{lbnd}$,
$\ikw{level}$, $\ikw{proper}$, $\ikw{abst}$, and $\ikw{abstr}$
are defined nearly the same as in the \HOL version.
For predicates such as $\ikw{level}$, we have a choice that we did not
have in \HOL\@.  In Coq, \Prop is the type of logical propositions,
whereas \Set is the type of datatypes.
 \Prop and \Set allow us to distinguish \emph{logical} aspects from
\emph{computational} ones \wrt our libraries.  The datatype $bool$ for
example, distinct from \Prop, is defined inductively in the Coq
standard library as a member of \Set.  One option in defining
$\ikw{level}$ is to define it as a function with target type $\bool$,
which evaluates via conversion to $true$ or $false$.  The other is to
define it as an inductive predicate (in \Prop), and then we will need
to provide proofs of $\ikw{level}$ subgoals instead of reducing them to
$true$.  We chose the latter option, using \Prop in the definition of
$\ikw{level}$ and all other predicates.
This allowed us to define inductive predicates in Coq that have the
same structure as the \HOL definitions, keeping the two versions as
close as possible.  For our purposes, however, the other option
should have worked equally well.

For predicates $\ikw{ordinary}$, $\ikw{lbnd}$, $\ikw{abst}$,
and $\ikw{abstr}$, which each have an argument of functional type,
there is one further difference in the Coq definitions.  Equality in
\HOL is extensional, while in Coq, it is not.  Thus, it was necessary
to define extensional equality on type $(\expr \fs \expr)$
\emph{explicitly} and use that equality whenever it is expressed on
this type, \viz
\begin{eqnarray*}
=_{ext} &\oftype &\fsprems{\expr\fs\expr,\expr\fs\expr} \fs Prop
\end{eqnarray*}
Formally, $(f =_{ext} g) \idef \forall x.(f x = g x)$.  For example,
this new equality appears in the definition of $\ikw{abst}$.  In the Coq
version, we first define an auxiliary predicate $\ikw{abst\_aux}$
defined exactly as $\ikw{abst}$ in \HOL, and then define $\ikw{abst}$
as:
$$\abst{i}{e} \idef \exists e'.\; e'=_{ext}e\land
\mathsf{abst\_aux}~i~e'.$$ 
The predicate $\ikw{abstr}$ has the same definition as in \HOL, via this
new version of $\ikw{abst}{}$.  The definition of $\ikw{lbnd}$ parallels
the one for $\ikw{abst}$, in this case using $\ikw{lbnd\_aux}$.  For
the $\ikw{ordinary}$ predicate, we obtain the Coq version from the \HOL
definition simply by replacing $=$ with $=_{ext}$.

The proof that $\ikw{lbnd}$ is a total relation is by induction on
$\ikw{rank}$ and
the induction case uses a proof by cases on whether or not a term of
type $(\expr \fs \expr)$ is ordinary.  Note that the $\ikw{ordinary}$
property is not decidable, and thus this proof requires classical
reasoning, which is a second reason for using Coq's classical
libraries.

Coq provides a module which helps to automate proofs using
user-defined equalities that are declared as \emph{setoids}.  A setoid
is a pair consisting of a type and an equivalence relation on that
type.  To use this module, we first show that $=_{ext}$ is reflexive,
symmetric, and transitive.  We then declare certain predicates as
morphisms.  A morphism is a predicate in which it is allowable to
replace an argument by one that is equivalent according to the user-defined
equality.  Such replacement is possible as long as the corresponding
compatibility lemma is proved.  For example, we declare
$\ikw{ordinary}$, $\ikw{lbnd}$, $\ikw{abst}$, and $\ikw{abstr}$ as
morphisms.  In particular, the lemma for $\ikw{lbnd}$ proves that if
$(\lbnd{i}{e}{t})$, then for all terms $e'$ that are extensionally
equal to $e$, we also have $(\lbnd{i}{e'}{t})$.  Setoid rewriting then
allows us to replace the second argument of $\ikw{lbnd}$ by
extensionally equal terms, and is especially useful in the proof that
every $e$ is related to a unique $t$ by $\ikw{lbnd}$.

As stated above, we obtain $\lbind{}{}$ by eliminating the existential
quantifier in the description theorem.  Once we have this function, we
can define $\llambda{}$ as in \HOL and prove the Coq version of the
\textit{abstr\_lam\_simp} theorem
\mbox{(MC-Theorem~\ref{thm:inj})}:  
 \[
  \abstr e\Implies \abstr f
     \Implies [(\LAM x {e\arg{x}} = \LAM y {f\arg{y}}) \longleftrightarrow
     (e =_{ext} f)]
\]
Note the use of logical equivalence ($\longleftrightarrow$) between
elements of \Prop.  Extensional equality is used between elements of
type $(\expr \fs \expr)$ and Coq equality is used between other terms
whose types are in \Set.  Similarly, extensional equality replaces
equality in other theorems involving expressions of type $(\expr \fs
\expr)$.  For example \textit{abstraction\_induct}
(MC-Theorem~\ref{thm:abstractioninduct}) is stated as follows:
 \[
 \begin{array}{l}
\prems{\forall e\ a.\; \prems{e =_{ext} (\ilam{\ivar{}}{\Con{a}})}
        \Implies P~e ; \\
      ~~~\forall e\; \prems{e =_{ext}(\ilam{\ivar{}}{\ivar{}})} \Implies
        P~e ; \\
      ~~~\forall e\ n.\; \prems{e =_{ext}(\ilam{\ivar{}}{\Var{n}})}
        \Implies P~e ; \\
      ~~~\forall e\ j.\; \prems{e =_{ext}(\ilam{\ivar{}}{\Bnd{j}})}
        \Implies P~e ; \\
      ~~~\forall e\ f\ g.\; 
         \prems{e =_{ext}
                (\ilam{\ivar{}}{\App{\iapp{f}{\ivar{}}}{\iapp{g}{\ivar{}}}});
                P~f; P~g} \Implies P~e;\\
      ~~~\forall e\ f.\;
         \prems{e =_{ext}(\ilam{\ivar{}}{\Abs{(\iapp{f}{\ivar{}})}});
                P~f} \Implies P~e;\\
      ~~~\forall e.\; \prems{\lnot\ordinary{e}} \Implies P~e} \Implies P~e
\end{array}
\]

\section{\hybrid as a Logical Framework}
\label{using}

In this section we show how to use \hybrid as a logical framework,
first by introducing our first OL (Section~\ref{ssec:coding}) and
discussing the adequacy of the encoding of its syntax
(Section~\ref{ssec:adeq}).  Representation and adequacy of syntax are
aspects of encoding OLs that are independent of the two-level
architecture.  We then show that some object-level judgments can be
represented directly as inductive definitions
(Section~\ref{ssec:oljudg}).  We also discuss the limitations of
encoding OL judgments in this way, motivating the need for the
two-level architecture of Section~\ref{sec:2lev}.

The system at this level provides:
\begin{itemize}
\item A suite of theorems: roughly three or four dozens propositions,
  most of which are only intermediate lemmas leading to the few that are
  relevant to our present purpose: namely, injectivity and
  distinctness properties of \hybrid constants.

\item Definitions \ikw{proper} and \ikw{abstr}, which are important for
\hybrid's adequate representation of OLs.

\item A very small number of automatic tactics: for example \ikw{proper\_tac}
  (resp.\ \ikw{abstr\_tac}) automatically recognizes whether a given
  term is indeed proper (resp.~an abstraction).

\end{itemize}

We report here the (slightly simplified) code for \ikw{abstr\_tac},
to give an idea of how lightweight such tactics are:

\begin{verbatim}
fun abstr_tac defs =
       simp_tac (simpset()
                 addsimps defs @ [abstr_def,lambda_def] @ lbind_simps)
        THEN'
        fast_tac(claset()
                addDs [abst_level_lbind]
                addIs abstSet.intrs
                addEs [abstr_abst, proper_abst]);
\end{verbatim}
First
the goal is simplified (\texttt{simp\_tac}) using the definition of
\ikw{abstr}, \ikw{lambda}, other user-provided lemmas (\texttt{defs}),
and more importantly the \ikw{lbind} ``rewrite rules''
(\ikw{lbind\_simps}). At this point, it is merely a question of
resolution with the introduction rules for \ikw{abst}
(\texttt{abstSet.intrs}) and a few key lemmas, such as
MC-Lemma~\ref{mclem:proper_abst}, possibly as elimination rules. In
\HOL 2005, a tactic, even a user 
defined one, could also be ``packaged'' into a \emph{solver}. In this
way, it can be combined with the other automatic tools, such as the
simplifier or user defined tactics, \viz  \ikw{2lprolog\_tac}.  (See
Section~\ref{ssec:tac}.)

\subsection{Coding the Syntax of an OL in \hybrid}
\label{ssec:coding}

The OL we consider here is a fragment of
a pure functional language known as Mini-ML\@.  As mentioned, we
concentrate on a $\lambda$-calculus augmented with a fixed point operator,
although this OL could be easily generalized as in~\cite{Pfenning01book}.
This fragment is sufficient to illustrate the main ideas without
cluttering the presentation with too many details.

The types and terms of the source language are given respectively by:
\[ \begin{array}{rrcl}
 \mbox{\textit{Types}} & \tau & \bnfas & \mathit{i} \bnfalt \tau\fsp \tau' \\
 \mbox{\textit{Terms}}& e & \bnfas & x\bnfalt\lFun x e \bnfalt 
e\ \at \   e'\bnfalt  \rec{ e}\end{array} \]
We begin by showing how to represent the syntax in HOAS format using
\hybrid. Since  types for this language have no
bindings, they are represented with a standard datatype, named
\textit{tp} and defined in the obvious way;
more interestingly, as far as terms are
concerned, we need constants for
abstraction,   application and fixed point, say $\mathit{cABS}$,
$\mathit{cAPP}$, and $\mathit{cFIX}$.
Recall that in the meta-language, application is denoted by infix
$\$$, and abstraction by $\mathsf{LAM}$.

The above grammar is coded in \hybrid verbatim, provided that we declare
these constants to belong to the enumerated datatype $\con$
\begin{eqnarray*}
  \mathit{datatype}   \ con & =  & \mathit{cABS}\bnfalt\mathit{cAPP}\bnfalt \mathit{cFIX}
\end{eqnarray*}
add the {type} abbreviation
\begin{eqnarray*}
    \uexp &\idef &con\ \expr
\end{eqnarray*}
and the following  \emph{definitions}:
\[
\begin{array}{rcl}
  \ikw{@} &\oftype& \fsprems{\uexp, \uexp} \fs  \uexp       \\ 
  \ikw{fun} & \oftype &  [\uexp \fs  \uexp] \fs  \uexp  \\
  \ikw{fix} & \oftype &  [\uexp \fs  \uexp] \fs  \uexp  \\
  \llApp{E_1}{E_2} & \idef &  \Con{cAPP} \app E_1 \app E_2\\
  \llFun{x}{E\ x} & \idef &  \Con{cABS} \app \LAM{x}{E\ x}\\
  \llrec{E\ x} & \idef &  \Con{cFIX} \app \LAM{x}{E\ x}
\end{array}
\]
where $\ikw{fun}$ (resp.~$\ikw{fix}$) is indeed an \HOL binder, \eg,
$(\llrec{E\ x})$ is a syntax translation for $(\ikw{fix} (\lambda x.\ E
\ x))$.  
For example, $(\llrec{\llFun{y}{\llApp x y}})$ abbreviates:
\[
(\Con{cFIX} \app (\LAM{x }{\Con{cABS} \app (\LAM{y }{(\Con{cAPP} \app x \app y)})}))
\]

Note again that the above are only \emph{definitions} and by
themselves would not inherit any of the properties of the constructors
of a datatype.  However, thanks to the thin infra-structural layer
that we have interposed between the $\lambda$-calculus natively
offered by Isabelle and the rich logical structure provided by the
axioms of \HOL, it is now possible to \emph{prove} the freeness
properties of those definitions as if they were the constructors of
what \HOL would ordinarily consider an ``impossible'' datatype as
discussed earlier. More formally:

\begin{goal}["Freeness" properties of constructors] \mbox{}
  Consider the constructors\footnote{By abuse of language, we call
    \emph{constructors} what are more precisely \HOL constant
    definitions.}  \ikw{fun, fix, @}:
  \begin{itemize}

  \item The constructors have distinct images. For example:
\[
    \begin{array}{lcr}
      \llFun x {E\ x}  &\not= &(\llApp{E_1}{E_2})
\qquad\qquad\qquad\qquad\qquad\qquad\qquad\qquad\qquad\qquad
\mbox{\textbf{\emph{(FA\_clash)}}}
    \end{array}
\]
  \item Every non {binding} constructor is \emph{injective}.
  \item Every \emph{binding} constructor is \emph{injective} on
    \emph{abstractions}.  For example:
  \begin{eqnarray*}
\prems{\abstr E;\; \abstr E'} &\Implies &(\llrec{E\ x} = \llrec{E'\ 
    x}) = (E = E')  \end{eqnarray*}
  \end{itemize}
\label{thm:clash}
\end{goal}

\begin{proof}
  By a call to \HOL's standard simplification, augmented with the left-to-right
  direction of the crucial property \textit{abstr\_lam\_simp}
  (MC-Theorem~\ref{thm:inj}).
\qed
\end{proof}

This result will hold for any signature containing at most
second-order constructors, provided they are encoded as we have
exhibited.
These ``quasi-freeness'' properties---meaning freeness conditionally
on whether the function in a binding construct is indeed an
abstraction---are added to \HOL's standard simplifier, so that they
will be automatically applied in all reasoning contexts that concern
the constructors. In particular, clash theorems are best encoded in the guise of
\emph{elimination} rules, already incorporating the ``ex falso
quodlibet'' theorem. For example, \textit{FA\_clash}
of MC-Theorem~\ref{thm:clash} is equivalent to:
\begin{eqnarray*}
     \prems{\llFun x {E\ x} = (\llApp{E_1}{E_2})} &\Implies& P
\end{eqnarray*}

\subsection{Adequacy of the Encoding}
\label{ssec:adeq}
It is a customary proof obligation (at least) \wrt higher-order
encoding to show that the syntax (and later the judgments) of an OL
such as Mini-ML are \emph{adequately} represented in the
framework. While this is quite well-understood in a framework such as
LF, the ``atypical'' nature of \hybrid requires a discussion and some
additional work.  We take for granted (as suggested in
\cite{Ambler02}, then painstakingly detailed in \cite{CroleHA})
that \hybrid provides an adequate
representation of the $\lambda$-calculus.
Yet,  it would not be
possible to provide a ``complete'' proof of the adequacy of \hybrid as a
theory running on a complex tool such as \HOL.  Here we take a more
narrow approach, by working with a convenient fiction, \ie, a
\emph{model} of \hybrid as a simply-typed $\lambda$-calculus presented as
a logical framework. This includes:
\begin{itemize}
\item a ``first-order'' $\lambda$-calculus (\ie, where \textit{bool} can
  only occur as the target of a legal arrow type) as our term language;
\item introduction and elimination rules for atoms generated by their
  inductive definition;
\item simplification on the \hybrid level and modulo other decidable theories
  such as linear arithmetic.
\end{itemize}
We can use this as our framework to represent OLs; further this model
is what we consider when we state meta-theoretical properties of OL
encodings and prove them adequate. 

We follow quite closely Pfenning's account in the \textit{Handbook of
Automated Reasoning}~\cite{Pfenning99handbook}.  By adequacy of a
language representation we mean 
 that there is an \emph{encoding} function $\encode{\Gamma}
{\cdot}$ from OL terms with free variables in $\Gamma$
to the canonical forms of the framework in an appropriate signature, as
well as its inverse $\decode{\Gamma} \cdot$ such that:
\begin{enumerate}
\item validity: for every mathematical object $t$ with free variables
  in $\Gamma$, $\encode{\Gamma}{t}$ is a canonical (and thus
  unique, modulo $\alpha$-conversion) representation in the
  framework. Note that we use $\Gamma$ both for the \hybrid and the
  OL's variables context;
\item completeness: for every canonical term $E$ over   $\Gamma$,
  $\decode{\Gamma}{E}$, results in a unique OL term $t$; furthermore
  $\encode\Gamma{\decode \Gamma E} = E$ and
  $\decode\Gamma{\encode \Gamma t} = t$.
\item compositionality: the bijection induced by
  $\encode{\cdot}{\cdot}$ and $\decode{\cdot}{\cdot}$ commutes with
  substitution; formally $\encode \Gamma {[t_1/x] {t_2}} = {[\encode
    \Gamma {t_1}/x]} \; \encode \Gamma {t_2}$ and $\decode \Gamma
  {[E_1/x] {E_2}} = {[\decode \Gamma {E_1}/x]} \; \decode \Gamma
  {E_2}$.
\end{enumerate}
Clearly the first requirement seems easier to satisfy, while the
second one tends to be more problematic.\footnote{Incidentally, some
  \emph{first-order} encodings, which are traditionally assumed not to
  be troublesome, may fail to satisfy the second requirement in the
  most spectacular way. Case in point are encodings typical of the
  Boyer-Moore theorem prover, \eg, case studies
  concerning the properties of the Java Virtual Machine
  \cite{acl2java}.  Since the framework's language consists of
  S-expressions, a decoding function does not really exist: in fact,
  it is only informally understood how to connect a list of pairs of
  S-exp to an informal function in, say, the operational semantics of
  the JVM, assuming that the code maintains the invariants of
  association lists.  Within \hybrid we can do much better, although
  we will fall somewhat short of LF's standards.}  In general, there
could be two main obstacles when representing an OL's signature with
some form of HOAS in a logical framework, both related to the
existence of ``undesirable'' canonical terms in the framework, \ie,
honest-to-goodness terms that are \emph{not} in the image of the
desired encoding:
\begin{enumerate}
\item If the framework is \emph{uni-typed}, we need predicates to
  express the well-formedness of the encoding of expressions of
  the OL\@. Such well-formedness properties must now be \emph{proved}, 
  differently from settings such as LF, where such properties are
  handled by type-checking.  In particular,
  \hybrid constants are not part of a datatype, so they do not enjoy
  the usual closure condition. Moreover there are proper \hybrid terms
  such as $\App {\LAM{x}{\arg{x}}} {(\Var 0)}$ that are not in the
  image of the encoding, but are still canonical forms of type
  $\expr$.

\item If the framework is strong enough, in particular if its type
  system supports at least a primitive recursive function space, 
 \emph{exotic} terms do arise, as discussed earlier, \ie, terms
  containing irreducible functions that are not parametric on their
  arguments, \eg, $\llrec {\llFun y {\ikw{if} \ x = y\ \ikw{then} \ x\
  \ikw{else}\ y}}$.
\end{enumerate}
As far as the second issue is concerned, we use \ikw{abstr}
annotations to get rid of such ``non-parametric'' functions.  
As mentioned by \cite{Pfenning99handbook} and is standard
practice in concrete approaches (\eg, the \texttt{vclosed} and
\texttt{term} predicate in the ``locally 
named/nameless'' representation of \cite{McKinna99,ACPPW07}), we 
introduce well-formedness predicates (as inductive
definitions in \HOL) to represent OL types.

To make clear the correspondence between the OL and its encoding, we
re-formulate the BNF grammar for Mini-ML terms as a well-formedness
judgment:
\begin{eqnarray*}
  \ian{}{\Gamma,x\vdash x}{}
\qquad\qquad
\ibn{\Gamma\vdash t_1}{\Gamma\vdash t_2}
    {\Gamma\vdash\lApp {t_1} t_2}{}
\qquad\qquad
\ian{\Gamma,x\vdash t}
    {\Gamma\vdash\lFun x t}{}
\qquad\qquad
\ian{\Gamma,x\vdash t}
    {\Gamma\vdash\rec t}{}
\end{eqnarray*}
Based on this formulation,
the definition of encoding of a Mini-ML term into \hybrid and its
decoding is unsurprising \cite{Pfenning01book}.  Notation-wise, we overload the
comma so that $\Gamma,x$ means $\Gamma\cup \{x\}$; we also use
$\Gamma$ for both the context of OL variables and of \hybrid variables
of type $\uexp$:
\[
\begin{array}{rcl@{\qquad}rcl}
  \encode{\Gamma,x}{x}& = & x&
\encode{\Gamma}{\lApp {t_1} t_2} & =& \llApp{\encode \Gamma {t_1}}{\encode
  \Gamma {t_2}} \\
\encode{\Gamma}{\lFun x t}& =&\llFun x {\encode{\Gamma,x}{t}} 
&\encode{\Gamma}{\rec t} &=&\llrec {\encode{\Gamma,x} t} 
\end{array}\]
\[
\begin{array}{rcl@{\qquad}rcl}
  \decode{\Gamma,x}{{x}}& = & x&
\decode{\Gamma}{\llApp {E_1} E_2} & =& \lApp{\decode \Gamma {E_1}}{\decode
  \Gamma {E_2}} \\
\decode{\Gamma}{\llFun x E}& =&\lFun x {\decode{\Gamma,x}{E}} 
&\decode{\Gamma}{\llrec E} &=&\rec {\decode{\Gamma,x}{E}} 
\end{array}\]

We then introduce an inductive predicate ${\isterm{\_}{\_}}$ of type
$\fsprems{\uexp\ \mathit{set},\uexp} \fs \bool$, which 
addresses at the same time the two aforementioned issues. It
identifies the subset of $\uexp$ that corresponds to the open terms of
Mini-ML over a set of (free) variables.
 \[
 \begin{array}{rcl}
 \mathit{Inductive}\ \isterm{\_}{\_} &\oftype &\fsprems{\uexp\ \mathit{set},\uexp} \fs \bool\\[1mm]
\prems{x\in\Gamma}   & \Implies & \isterm{\Gamma}{({x})} \\
   \prems{\isterm{\Gamma}{E_1} ;\; \isterm{\Gamma}{E_2}} & \Implies &  
   \isterm{\Gamma} {(\llApp{E_1}{E_2})}\\ 
    \prems{\forall x.\ \proper x \limp\isterm{\Gamma,x}{(E~x)};
      \; \ikw{abstr\ E }}& \Implies & 
    \isterm{\Gamma}{(\llFun{x}{E\ x})} \\
    \prems{\forall x.\ \proper x \limp\isterm{\Gamma,x}{(E~x)};
      \; \ikw{abstr\ E }}& \Implies & 
    \isterm{\Gamma}{(\llrec{E\ x})} 
 \end{array}
\]

We can now proceed to show the validity of the encoding in the sense that
$\Gamma\vdash t$ entails that $\isterm{\Gamma}{\encode \Gamma t}$ is
provable in \HOL. However, there is an additional issue: the obvious
inductive proof requires, in the binding case, the derivability of the
following fact:
\begin{equation}
\label{le:abstr-adeq}
\ikw{abstr} (\lambda x.\ \encode{\Gamma,x} t ) 
\end{equation}
A proof by induction on the structure of $t$ relies on
$$ \ikw{abstr} (\lambda x.\ \LAM{y}{\encode{\Gamma,x,y}{t}})$$
This holds once $\lambda x y.\, \encode{\Gamma,x,y}{t}$ is a
\emph{biabstraction},  namely:
$$
\ikw{biAbstr} (\lambda xy.\ {E\ x \ y})  \Implies \ikw{abstr} (\lambda
x.\ \LAM{\, y }{\ E\ x \ y}) 
$$ Biabstractions are the generalization of abstractions to
functions of type $(\expr\fs\expr\fs\expr)\fs\expr$.  The inductive
definition of this notion simply replays that of \ikw{abst} and we
skip it for the sake of space. We note however that the above theorem
follows by structural induction using only introduction
     and elimination rules for $\abst{}{} $.
 We therefore consider proven the above fact (\ref{le:abstr-adeq}).

 If $\Gamma=\{x_1,\ldots, x_n\}$, we write $(\proper \Gamma$) to denote
 the \HOL context $\prems{\proper{x_1};\ldots;\proper{x_n}}$.
\begin{lem}[Validity of Representation]
\label{le:exp-sound}
  If $\Gamma\vdash t$, then
  $(\proper\Gamma\Implies\isterm{\Gamma}{\encode \Gamma  t})$ is provable in
    \HOL.
\label{lem:valid}
\end{lem}
\begin{proof}
  By the standard induction on the derivation of $\Gamma\vdash t$,
  using fact (\ref{le:abstr-adeq}) in the binding cases.
\qed
\end{proof}
As far as the converse of Lemma~\ref{lem:valid} goes, we need an
additional consideration.  As opposed to intentionally weak frameworks
\cite{DeBruijn91lf}, \HOL has considerable expressive power; various
features of the underlying logic, such as classical reasoning and the
axiom of choice, can be used to construct proofs about an OL that do
not correspond to the informal constructive proofs we aim to
formalize.
We therefore need to  restrict ourselves to a second-order
intuitionistic logic. The issue
here is guaranteeing that \emph{inverting} on hypothetical judgments
respects the operational interpretation of the latter, \ie, the
deduction theorem, rather than viewing them as classical
tautologies. We call such a derivation \emph{minimal}. Since \HOL does
have  proof terms \cite{BerghoferN-TPHOLs00}, this notion is in
principle checkable.\footnote{Note that \HOL provides a basic
  intuitionistic prover \ikw{iprover}, and it could be connected to an
  external more efficient one via the sledgehammer protocol. }

\begin{lem}[Completeness of Representation]
\label{le:complete}
Let $\Gamma$ be the set $ \{x_1: \uexp,\dots,x_n:uexp\}$; if
$(\proper\Gamma\Implies\isterm\Gamma E)$ has a \emph{minimal} derivation
in \HOL, then $\decode\Gamma E$ is defined and yields a Mini-ML
expression $t$ such that $\Gamma\vdash t$ and $\encode\Gamma{\decode
  \Gamma E} = E$. Furthermore, $\decode\Gamma{\encode \Gamma t} = t$.
\end{lem}
\begin{proof}
  The main statement goes by induction on the minimal derivation of
  $\proper\Gamma\Implies\isterm\Gamma E$; we sketch one case: assume
  $\proper\Gamma\Implies\isterm{\Gamma}{(\llrec{E\ x})}$; by
  inversion, $\isterm{\Gamma,x}{(E~x)}$ holds for a parameter $x$
  under the assumption $\proper {(\Gamma, x)}$. By definition
  $\decode{\Gamma} {\llrec{(E\ x)}} = \rec{\decode{\Gamma,x} {E\ x}}$.
  By the I.H.\, the term $\decode{\Gamma,x} {E\ x}$ is defined and
  there is a $t$ s.t.\ $t=\decode{\Gamma,x} {E\ x}$ and
  $\Gamma,x\vdash t$.  By the BNF rule for $\mathbf{fix}$, $\Gamma
  \vdash \rec t$ and again by the I.H.\ and definition,
  $\encode{\Gamma}{\decode{\Gamma} {\llrec{E\ x}}}={\llrec{E\ x}}$.
  Finally, $\decode\Gamma{\encode \Gamma t} = t$ follows by a
  straightforward induction on $t$. \qed\end{proof}

\begin{lem}[Compositionality]
\label{le:comp}
\begin{enumerate}
\item $\encode \Gamma {[t_1/x] {t_2}} = {[\encode \Gamma {t_1}/x]} \
  \encode \Gamma {t_2}$, where $x$ may occur in $\Gamma$.
\item If $\decode\Gamma{E_1}$ and $\decode\Gamma{E_2}$ are defined,
  then $\decode \Gamma {[E_1/x] {E_2}} = {[\decode \Gamma {E_1}/x]} \
  \decode \Gamma {E_2}$.
\end{enumerate}
\end{lem}
\begin{proof}
  The first result may be proved by induction on $t_2$ as in Lemma 3.5
  of \cite{Pfenning01book}, since the encoding function is the same,
  or we can appeal to the compositionality property of \hybrid, proved
  as Theorem 4.3 of \cite{CroleHA}, by unfolding the \hybrid
  definition of the constructors.  The proof of the second part is a
  similar induction on $E_2$.  \qed\end{proof}
Note that completeness and compositionality do not depend on fact
(\ref{le:abstr-adeq}). 
\begin{figure}\renewcommand{\Lx}{\mathbf{fun\ x.\ }}
\[ \fboxsep=4mm\fbox{$ \begin{array}{c}
 \ianc{\eval{e_1}{\Lx e_1'}\qquad \eval{e_2}{v_2}
   \qquad\eval{[v_2/x]e_1'} {v}}{\eval{e_1\ \at \
     e_2}{v}}{\mathbf{ev\_app}}\vsk  
 \ianc{}{ \eval{\Lx  e}{\Lx  e} }{\mathbf{ev\_fun}} \qquad
\ianc{\eval{[\rec{e}/x]e}{v}}{\eval{\rec{e}}{v}}{\mathbf{ev\_fix}}\vsk
\dotfill\vsk
\ianc{\Gamma, x\oftp \tau \vd e \hastype \tau'}{\gvd \Lx e \hastype \tau \fsp
\tau'}{\mathbf{tp\_fun}} \qquad
 \ianc{\Gamma, x\oftp \tau \vd e \hastype \tau}{\gvd \rec{e} \hastype
   \tau }{\mathbf{tp\_fix}}\vsk 
 \ianc{\Gamma(x) = \tau}{\gvd x \hastype \tau}{\mathbf{tp\_var}}  \qquad
 \ibnc{\gvd e_1 \hastype \tau' \arrow \tau}{\gvd e_2 \hastype
   \tau'}{\gvd e_1\ \at \  e_2 \hastype \tau} 
{\mathbf{tp\_app}}
\end{array}$}  \]
\caption{Big step semantics and typing rules for a fragment of Mini-ML.
  \label{fig:dyn-st} }
 \end{figure}

\smallskip

\subsection{Encoding Object-Level Judgments}
\label{ssec:oljudg}
We now turn to the encoding of
object-level \emph{judgments}.  In this and the next section, we will
consider the standard judgments for big-step call-by-value operational
semantics ($\eval{e}{v}$) and type inference\ ($\Gamma \vd e \hastype
\tau$), depicted in Figure~\ref{fig:dyn-st}.  Evaluation can be directly
expressed as an \emph{inductive} relation (Figure~\ref{fig:hoas-eval1})
in full HOAS style.  Note that substitution is encoded via meta-level
$\beta$-conversion in clauses for \textbf{ev\_app} and \textbf{ev\_fix}.

\begin{figure}  \begin{center}
    \[
\begin{array}{rcl}
\hline \\
  \mathit{Inductive}\ \eval{ \_} { \_} &\Hoftype & \fsprems{\uexp,
    \uexp} \fs \bool  \\ 
  \prems{\eval{E_1}{\llFun{x}{E'\ x}};\; \eval{E_2}{V_2};\;
    \eval{(E'\ V_2)}{V};          \abstr{E'}  }
  & \Implies &  \eval{(\llApp{E_1} {E_2})}{V}  \\
  \prems{\isterm\emptyset {(\llFun{x}{E\ x})} ; \abstr{E}
  } 
  & \Implies &
  \eval{\llFun{x}{E\ x}}{\llFun{x}{E\ x}}\\ 
  \prems{\eval{E\ (\llrec{E\ x})}{V};\; \isterm\emptyset{(\llrec{E\ 
        x})} ; \abstr{E}
  }
  & \Implies & \eval{\llrec{E\ x}}{V}
 \\[.5ex]\hline
\end{array}
\]
    \caption{Encoding of big step evaluation in Mini-ML.}
    \label{fig:hoas-eval1}
  \end{center}
\end{figure}

This definition is an honest to goodness inductive relation that can
be used as any other one in an HOL-like setting: for example, queried
in the style of Prolog, as in $\exists t.\ \eval {\llrec{\llFun y
    {\llApp x y }}}{ t}$, by using only its introduction rules and
abstraction solving. Further this kind of relations can be reasoned
about using standard induction and case analysis.  In fact, the very fact that
evaluation is recognized by \HOL as {inductive} yields \emph{inversion
  principles} in the form of elimination rules.
This would correspond, in meta-logics such as $\foldna$, to
applications of \emph{definitional reflection}.  In \HOL (as well as
in Coq) case analysis is particularly well-supported as part of the
datatype/inductive package. Each predicate \texttt{p} has a general
inversion principle \texttt{p.elim}, which can be specialized to a
given instance $\mathtt{(p\ \vec t)}$ by an ML built-in function
\texttt{p.mk$\_$cases} that operates on the current simplification
set; specific to our architecture, note again the \emph{abstraction}
annotations as meta-logical premises in rules mentioning binding
constructs.
To take this into account, we call this ML function modulo the
quasi-freeness properties of \hybrid constructors so that it makes the
appropriate discrimination. For example the value of
\texttt{meval$\_$mk$\_$cases} ($\eval {\llFun x E}{V}$) is:
\begin{eqnarray*}
\mbox{\textbf{(meval\_fun\_E)}}\qquad
\prems{\eval {\llFun x E\ x}{V};\;\Forall F\,
  \prems{\isterm\emptyset{(\llFun x F\ x)};\;\abstr F\\
 \ikw{lambda\ 
    } E = \ikw{lambda\ } F;\; 
V = \llFun x F\ x} 
\Implies P}\Implies P
\end{eqnarray*}
Note also that the inversion principle has an explicit equation $
\llambda {E} = \llambda {F}$ (whereas definitional reflection employs
full higher-order unification) and such equations are solvable only under the
assumption that the body of a $\lambda$-term is well-behaved (\ie, is
an abstraction).

Finally, using such elimination rules, and more importantly the structural
induction principle provided by \HOL's inductive package, we can prove
standard meta-theorems, for instance uniqueness of evaluation.
\begin{goal}[eval\_unique]
$\eval E F \Implies  \forall G.\ \eval E G \limp  F = G$.
\end{goal}
\begin{proof}
  By induction on the structure of the derivation of $\eval E F$ and
  inversion on $\eval E G$.  \qed
\end{proof}

The mechanized proof does not appeal, as expected, to the
functionality of substitution, as the latter is inherited by the meta-logic,
contrary to first-order and ``weak'' HOAS encodings (see
Section~\ref{ssec:oth}).  Compare this also with
the standard paper and pencil proof, which usually ignores this
property.

We can also prove some ``hygiene'' results, showing
that the encoding of evaluation preserves properness and
well-formedness of terms:

\begin{mclemma}[eval\_proper, eval\_isterm]
$$\begin{array}{rll}
1. & \eval{E}{V} \Implies \proper{E} \land \proper{V}
     \qquad\qquad\qquad & \\
2. & \eval{E}{V} \Implies \isterm\emptyset{E} \land \isterm\emptyset{V}
     \qquad\qquad\qquad & \end{array}$$
\label{mclem:evalproper}
\end{mclemma}
Note the absence in Figure~\ref{fig:hoas-eval1} of any \ikw{proper}
assumptions at all: only  the \ikw{isterm} assumptions in the
application and fixed point
cases are needed.  We have included just enough assumptions to prove the above
results.
In general, this kind of result must be proven for each new OL,
but the proofs are simple and the reasoning steps follow a similar
pattern for all such proofs.

With respect to the adequacy of object-level judgments, we can
establish first the usual statements, for example soundness and
completeness of the representation; for the sake of clarity as well as
brevity in the statement and proof of the lemma we drop the infix
syntax in the \HOL definition of evaluation, and omit the obvious
definition of the encoding of said judgment:\newcommand{\eeval}[2]{(\ikw{eval} \ {#1} \ {#2})}
\begin{lem}[Soundness of the encoding of evaluation]
  \label{le:adeq-eval}
  Let $e$ and $v$ be \emph{closed} Mini-ML expressions such that $\eval e v$;
  then we can prove in \HOL $\eeval {\encode \emptyset e} {\encode
    \emptyset v}$.
\end{lem}
\begin{proof}
  By induction on the derivation of $\eval e v$. Consider the
  \textbf{ev\_fun} case: by definition of the encoding on expressions
  and its soundness (Lemma~\ref{le:exp-sound}) we have that $\isterm
  \emptyset {\encode{\emptyset}{\lFun x e}}$ is provable in \HOL; by
  definition and inversion $\isterm \emptyset {(\llFun x
    {\encode{x}{e}} )}$ and $\abstr (\lambda x.\ \encode{x}{e})$
  holds, hence by the introduction rules of the inductive definition
  of evaluation $\eeval {\llFun x {\encode{x}{e}}}{\llFun x
    {\encode{x}{e}}}$ is provable, that is, by definition, $\eeval
  {\encode{\emptyset} {\lFun x e}} {\encode{\emptyset} {\lFun x
      e}}$. The other two cases also use compositionality
  (Lemma~\ref{le:comp}) and the induction hypothesis.  \qed\end{proof}

\begin{lem}[Completeness of the encoding of evaluation]
  \label{le:complete-eval}
  If $\eeval E V$ has a minimal derivation in \HOL, then
  $\decode\emptyset E$ and $\decode\emptyset V$ are defined and yield
  Mini-ML expressions $e$ and $v$ such that $\eval{e}{v}$.
\end{lem}
\begin{proof}
It follows from MC-Lemma~\ref{mclem:evalproper} that
$\isterm\emptyset{E}$ and $\isterm\emptyset{V}$, and thus from
Lemma~\ref{le:complete} that $\decode\emptyset E$ and
$\decode\emptyset V$ are defined.  The proof of $\eval{e}{v}$ follows
directly by induction on the minimal derivation of $\eeval E V$, using
compositionality (Lemma~\ref{le:comp}).
\qed\end{proof}

Now that we have achieved this, does that mean that all the advantages
of HOAS are now available in a well-understood system such as \HOL\@?
The answer is, unfortunately, a qualified ``no''. Recall the three
``tenets'' of HOAS:
\begin{enumerate}
\item $\alpha$-renaming for free,  inherited from the ambient
  $\lambda$-calculus identifying meta-level and object-level bound variables;
\item object-level substitution as meta-level $\beta$-reduction;
\item object-level contexts as meta-level assumptions.
\end{enumerate}
As of now, we have achieved only the first two.  However,
while accomplishing in a consistent and relatively painless way the first
two points above is no little feat,\footnote{Compare this to other
  methods of obtaining $\alpha$-conversion by constructing equivalence
  classes~\cite{Vestergaard01rta,FordM03} in a proof assistant.} the
second one, in particular, being in every sense novel, no HOAS system
can really be worth its name without an accounting and exploiting
of reasoning in the presence of \emph{hypothetical and parametric
  judgments}.  
We consider the standard example of encoding type inference
(Figure~\ref{fig:dyn-st}) in a language such as Twelf.
Using
\HOL-like syntax (where we use $bool$ for Twelf's $type$, the
$\mathbf{tp\_app}$, $\mathbf{tp\_fun}$, and $\mathbf{tp\_fix}$ rules
would be represented as follows:
$$\begin{array}{rcl}
\_\hastype\_ & \Hoftype & \fsprems{\uexp, \htp} \fs \bool  \\
 \prems{E_1\hastype (T'\ofsp T);\; E_2\hastype T'}
  & \Implies & (\llApp {E_1} {E_2}) \hastype T \\
\prems {\forall x\;(\underline{x\hastype T}\limp (E\ x)\hastype T')}
  & \Implies & (\llFun{x} {E\ x})\hastype (T \ofsp T') \\
\prems {\forall x\;(\underline{x\hastype T}\limp (E\ x)\hastype T)}
  & \Implies & (\llrec{E\ x})\hastype T
\end{array}$$
Each typing judgment $x\hastype\tau$ in an object-level context
($\Gamma$ in Figure~\ref{fig:dyn-st}) is represented as a logical
assumption of the
form $x\hastype T$.  In the spirit of higher-order encoding, there is
no explicit representation of contexts and no need to encode the
$\mathbf{tp\_var}$ rule.  However, because of the underlined negative
recursive 
occurrences in the above formulas, there is simply no way to encode this directly 
in an inductive 
setting, short of betraying its higher-order nature by introducing
ad-hoc datatypes (in this case lists for environments) and, what's
worse, all the theory they require. The latter may be trivial on
paper, but it is time-consuming and has little to do with the
mathematics of the problem.\footnote{A compromise is the
  ``weak'' HOAS view mentioned earlier and discussed in
  Section~\ref{ssec:oth}.}

Moreover, at the level of the meta-theory, it is only the coupling of
items 2 and 3 above that makes HOAS encodings---and thus
proofs---so elegant and concise; while it is nice not to have to
encode substitution for every new signature, it is certainly much
nicer not to have to \emph{prove} the related substitution lemmas.
This is precisely what the pervasive use of hypothetical and parametric
judgments makes possible---one of the many lessons by Martin-L{\"o}f.

Even when hypothetical judgments are stratified and therefore
inductive, using \hybrid directly within \HOL (\ie, at a single level
as will become clear shortly) has been only successful in dealing with
predicates over \emph{closed} terms (such as simulation).  However, it
is necessary to resort to a more traditional encoding, \ie, via explicit
environments, when dealing with judgments involving \emph{open}
objects.  These issues became particularly clear in the case-study 
reported in~\cite{Momigliano02lfm}, where the \hybrid syntax allowed
the following elegant encoding of \emph{closed} applicative
(bi)simulation~\cite{Ong}:
\[
\begin{array}{l}
   \prems{
      \forall  T.\  \eval{R}{\llFun{x}{T\ x}} \limp ( \abstr{T} \limp\mbox{}  \\
                  \exists U.\  \eval{S}{\llFun{x} {U\ x}} \land
            \abstr{U}  \land\mbox{}  
            \forall  p.\  \lsim{(T\  p)}{(U\  p)})} \\
\qquad\qquad  \Implies    \lsim{R}{S} 
\end{array} 
\]
together with easy proofs of its basic properties (for example, being
a pre-order). Yet, dealing with \emph{open} (bi)simulation required
the duplication of analogous work in a much less elegant way. 

This does not mean that results of some interest cannot be proved
working at one level.  For example, the aforementioned paper
(painfully) succeeded in checking non-trivial results such as
a Howe-style proof of congruence of applicative
(bi)simulation~\cite{Howe96ic}.\footnote{However, it would take a
  significant investment in man-months to extend the result from the
  lazy $\lambda$-calculus  to
  more interesting calculi such as~\cite{Lassen06}.} Another
example~\cite{ACM03prim} is the quite intricate verification of
subject reduction of MIL-LITE \cite{MIL-LITE}, the intermediate
language of the MLj compiler \cite{MLJ}.

In those experiments, HOAS in \HOL seemed only a nice interlude, soon
to be overwhelmed by tedious and non-trivial (at least mechanically)
proofs of list-based properties of open judgments and by a number of
substitutions lemmas that we had hoped to have eliminated for good.
These are the kinds of issues we address with the two-level
architecture, discussed next.

\section{A Two-Level Architecture}
\label{sec:2lev}

The \emph{specification} level mentioned earlier (see
Figure~\ref{fig:arch}) is introduced to solve the problems discussed
in the previous section of reasoning in the presence of negative
occurrences of OL judgments and reasoning about open terms.  A
\emph{specification logic} (SL) is defined inductively, and used to
encode OL judgments.  Since hypothetical judgments are encapsulated
within the SL, they are not required to be inductive themselves.  In
addition, SL contexts can encode assumptions about OL variables, which
allows reasoning about open terms of the OL\@.  We introduce our first
example SL in Section~\ref{ssec:nseq}.  Then, in
Section~\ref{ssec:prog}, we continue the discussion of the sample OL
introduced in Section~\ref{using}, this time illustrating the encoding
of judgments at the SL level.
In Section~\ref{ssec:tac}, we discuss proof automation and in
Section~\ref{ssec:variation}
we present a variant of the proof in Section~\ref{ssec:prog} that
illustrates the flexibility of the system.

\subsection{Encoding the Specification Logic}
\label{ssec:nseq}
We introduce our first SL, namely a fragment of second-order
hereditary Harrop formulas~\cite{Miller91apal}. This is sufficient
for the
encoding of our first case-study: subject reduction for the sub-language of
Mini-ML that we have introduced before (Figure~\ref{fig:dyn-st}).  The
SL language is defined as follows, where $\tau$ is a ground type and
$A$ is an atomic formula.
\[ \begin{array}{rrcl}
 \mbox{\textit{Clauses}} & D & \bnfas & \top\bnfalt A \bnfalt
 {D_1}\land{D_2}\bnfalt
 {G}\to {A}\bnfalt\forall^{\tau} x\ldot{D}\bnfalt
 \forall^{\tau\arrow\tau} x\ldot{D}\\
\mbox{\textit{Goals}} & G & \bnfas & \top\bnfalt A \bnfalt
 {G_1}\land{G_2}\bnfalt
 {A}\to {G}\bnfalt\forall^{\tau} x\ldot G\\
\mbox{\textit{Context}} & \Gamma & \bnfas & \emptyset \bnfalt A,\Gamma
\end{array} \]
The $\tau$ in the grammar for goals is instantiated with $\expr$ in
this case. Thus, quantification is over a ground type whose exact
elements depend on the instantiation of $\con$, which, as discussed in
Section~\ref{ssec:prog}, is defined at the OL level.  Quantification in
clauses includes second-order variables. We will use it, for instance,
to encode variables $E$ of 
type $\expr\fs\expr$ that appear in terms such as $\llFun{x}{E\ x}$.
Quantification in clauses may also be over first-order variables of
type $\expr$, as well as over variables of other ground types such as
\textit{tp}.  In this logic, we view contexts as \emph{sets}, where we
overload the comma to denote adjoining an element to a set.  Not only
does this representation make mechanical proofs of the standard
proof-theoretic properties easier compared to using lists, but it is
also appropriate for a sequent calculus that enjoys contraction and
exchange, and designed so that weakening is an admissible
property. This approach will also better motivate the use of lists in
sub-structural logics in the next section. Further, in our setting,
contexts are particularly simple, namely sets of \emph{atoms}, since
only atoms are legal antecedents in implications in goals.

The syntax of \emph{goal} formulas can be directly rendered with an \HOL 
  datatype:
  \[ \begin{array}{rrcl} \mbox{\textit{datatype}}& \oo & = &
    \sltt\bnfalt\slAt{\atm}\bnfalt \oo\ \slAnd{}{} \ \oo\bnfalt \atm\
    \slImp{}{} \ \oo\bnfalt \mathsf{all}\ (\expr\fs\oo)
  \end{array} \] 
We write $\atm$ to represent the type of atoms;
  $\slAt{ \_ }$ coerces atoms into propositions. The definition of
  $\atm$ is left as an implicit parameter at this stage, because various
  instantiations will yield the signature of different OLs,
  specifically predicates used to encode their judgments.

This language is so simple that its sequent calculus is analogous to a
logic programming interpreter.  All clauses allowed by the
above grammar can be normalized to (a set of) clauses of the form:
\[ \begin{array}{rrcl}
 \mbox{\textit{Clauses}} & D & \bnfas &
    \forall^{\sigma_1} x_1\ldots\forall^{\sigma_n} x_n\;(G\Imp A)
\end{array} \]
where $n\geq 0$, and
for $i=1,\ldots,n$, $\sigma_i$ is either a ground type, or has
the form $\tau_1\arrow\tau_2$ where $\tau_1$ and $\tau_2$ are ground
types.  In analogy with logic programming, when writing clauses,
outermost universal quantifiers will be omitted, as those variables
are implicitly quantified by the meta-logic; implication will be
written in the reverse direction, \ie, we write simply $A \If
G$,\footnote{This is also why we can dispose of the mutual definition
  of clauses and goals and avoid using a mutually inductive datatype,
  which, in the absence of some form of subtyping, would make the
  encoding redundant.} or when we need to be explicit about the
quantified variables, we write $\forall\Sigma(A \If G)$ where
$\Sigma=\{x_1,\ldots,x_n$\}.  This notation yields a more proof-search
oriented notion of clauses.  In fact, we can write inference rules so
that the only left rule is similar to Prolog's backchaining. Sequents
have the form $\iseqp{\Sigma;\G}{G}$, where $\Sigma$ is the current
signature of eigenvariables and we distinguish clauses belonging to a
static database, written $\Pi$, from atoms introduced via the right
implication rule, written $\Gamma$.  The rules for this logic are
given in Figure~\ref{fig:minlogbc}.  In the $\ir{bc}$ rule,
$[\Pi]$ is the set of all possible instances of clauses in $\Pi$
obtained by instantiating outermost universal quantifiers with all
closed terms of appropriate types.

\begin{figure}  \centering
  \[ {\fboxsep=7mm\fbox{$\begin{array}{l}
 \ian{}{\iseqp{\Sigma;(\G,A)}{A}}{\ir{init}} 
\qquad\qquad\qquad
 \ian{\iseqp{\Sigma;(\G, A)}{G}}
     {\iseqp{\Sigma;\G}{A\Imp G}}
     {\Imp_R} \\[1em]
\qquad\qquad \ibn{\iseqp{\Sigma;\G}{G_1}}{\iseqp{\Sigma;\G}{G_2}} 
     {\iseqp{\Sigma;\G}{G_1\with G_2}}{\with_R}\\[1em]
  \ian{}{\iseqp{\Sigma;\G}{\top}}{\top_R}
\qquad\qquad\qquad
  \ian{\iseqp{(\Sigma, a : \tau);\G}{G[a/x]}}{\iseqp{\Sigma;\G}{\forall^\tau x  \ldot G}} {\forall_R} \\[1em]
\qquad\qquad\qquad
  \ibn{\iseqp{\Sigma;\G}{G}}{ A \If G \in [\Pi]}
      {\iseqp{\Sigma;\G}{A}}
      {\ir{bc}}
\end{array}$}} \]
  \caption{A minimal sequent calculus with backchaining}
  \label{fig:minlogbc}
\end{figure}

This inference system is equivalent to the standard presentation of
minimal logic \cite{JoMinLog}, where the right rules are the same and
the left rules (given below) for conjunction, implication and universal
quantification replace the \textbf{bc} rule.

\[{\fboxsep=8mm\fbox{
$\begin{array}{c}
  \ian{\iseqp{\Sigma;\G,D_1,D_2}{G}}
       {\iseqp{\Sigma;\G,D_1\land\ D_2}{G}}
       {\ir{\land_L}} 
\qquad\qquad
 \ian{\iseqp{\Sigma;\G,D[t/x]}{G}}
  {\iseqp{\Sigma;\G,\forall x\ldot D}{G}}
       {\ir{\forall_L}}
\\[1em]
\ibn{\Sigma;\iseqp{\G}{G}}{\Sigma;\iseqp {\G,B}{A}}
       {\Sigma;\iseqp{\G, (G \Imp B)}{A}}
       {\ir{\Imp_L}}
\end{array}$}}
\]

\noindent
In fact, the \ir{bc} rule is derivable by eliminating the universal
quantifiers until the head of a clause matches the atom on the
right and then applying $\ir{\Imp_L}$. The reader should remember that
we are working in an ambient logic modulo some equational theory (in
the case of Isabelle $=_{\alpha\beta\eta}$) and that both atomic rules
($\ir{init}$ and $\ir{bc}$)
are applicable in the case when an atom on the right appears as an
assumption and unifies with the head of a definite clause in the
program $\Pi$.  Thus, we can inherit the completeness of \emph{uniform
  provability}~\cite{Miller91apal} \wrt an ordinary sequent calculus,
which holds for a much more expressive conservative extension of our
SL, namely higher-order Harrop formulas.

We encode this SL in Figure~\ref{fig:nseq}.
\begin{figure}\begin{eqnarray*}
\hline\\
  \mathit{Inductive} \  \slvdn{ \_}{ \_ }{ \_}& \Hoftype&
  \fsprems{\atm\ \mathit{set} ,\nat, \oo}\fs\bool \\ 
  &\Implies &\slvdn{\Gamma}{n}{\sltt}\\
  \prems{\slvdn{\Gamma}{n}{G_1}; \,\slvdn{\Gamma}{n}{G_2}  } & \Implies &
  \slvdn{\Gamma}{n+1}{(\slAnd{G_1}{G_2})}\\ 
  \prems{\forall x.\,\proper{x}\limp \slvdn{\Gamma}{n}{(G\ x)}  } & \Implies &
  \slvdn{\Gamma}{n+1}{(\slAll{x}{G\ x})}\\ 
  \prems{ \slvdn{A,\Gamma}{n}{G}  } & \Implies &
  \slvdn{\Gamma}{n+1}{(\slImp{A}{G})}\\ 
  \prems{A\in \ \Gamma  } & \Implies & \slvdn{\Gamma}{n}{\slAt{A}}\\  
  \prems{\sbc{A}{G};\, \slvdn{\Gamma}{n}{G}} & \Implies &
  \slvdn{\Gamma}{n+1}{\slAt{A}}
 \\[.5ex]\hline
\end{eqnarray*}
\caption{Encoding of a minimal specification logic}
\label{fig:nseq}
\end{figure}
We use the symbol $\slvd{}{}$ for the sequent arrow, 
in this case decorated with natural numbers that represent the
\emph{height} of a proof; this measure allows us to reason by complete
induction.\footnote{Proven in the \HOL's library in the form $(\Forall
  n.\ \forall m<n.\ P~m\Implies P~n) \Implies P~x$.}  For convenience
we write $\slvd{\Gamma}{G}$ if there exists an $n$ such that
$\slvdn{\Gamma}{n}{G}$, and furthermore we simply write $\slvde{G}$
when $\slvd{\emptyset}{G}$.  The first four clauses of the definition
directly encode the introduction ($R$) rules of the figure.
In the encoding of the $\forall_R$ rule, when we introduce new
eigenvariables of type $\expr$, we need to assume that they are
$\ikw{proper}$.  This assumption might be required for proving
subgoals of the form $(\abstr{E})$ for subterms $E ::\expr\fs\expr$
that appear in the goal as arguments to binding constructors; see
MC-Lemma~\ref{mclem:proper_abst} (\textit{proper\_abst)}.

We remark that  the only dependence on \hybrid in this layer 
is on the definition of $\ikw{proper}$.  This will also be true of the
SL we consider in Section~\ref{sec:olli}.  Although we do not discuss
it here, we could use SLs with (different) kinds of quantifiers that
could not be implemented via a datatype but only with \hybrid
constants; for example universal quantification in higher-order logic.
In this case, the specification layer would have a much greater
dependence on \hybrid.  On the other hand, if we take the alternative
solution to proper terms mentioned earlier
(when discussing MC-Lemma~\ref{mclem:proper_abst}) and replace $\expr$
with a type $\propert$ containing exactly the terms that satisfy
$\ikw{proper}$, and consider only the SLs presented in this paper,
then these SLs can be parameterized by the type of terms used in
quantification, and can be instantiated with types other than
$\propert$.

In the last two rules in Figure~\ref{fig:nseq}, atoms are provable either by
assumption or via \emph{backchaining} over a set of Prolog-like rules,
which encode the properties of the OL in question as
an inductive definition of the
predicate \ikw{prog} of type $\fsprems{\atm, \oo}\fs \bool$, which
will be instantiated in
Section~\ref{ssec:prog}. The sequent calculus is parametric in
those clauses and so are its meta-theoretical
properties.
Because \ikw{prog} is static it will be mentioned explicitly only
in adequacy proofs.
The notation $\sbc{A}{G}$ in Figure~\ref{fig:nseq} represents an
instance of one of the clauses of the inductive definition of
\ikw{prog}.

As a matter of fact our encoding of the judgment
$\slvdn{\Gamma}{n}{G}$ can be seen as a simple extension of the
so-called ``vanilla'' Prolog meta-interpreter, often known as
$\mathtt{demo}$ \cite{Hill&Gallagher1998}; similarly, the \textbf{bc}
rule would correspond to the following clause, using the predicate
$\mathtt{prog}$ in place of Prolog's built-in $\mathtt{clause}$:
$$
\mathtt{demo(Gamma,s(N),A) :-\ prog(A,G),\ demo(Gamma,N,G).}
$$

Existential quantification could be added to the grammar of goals, as follows:
\begin{eqnarray*}
\prems{\exists x.\, \slvdn{\Gamma}{n}{(G\ x)}  } & \Implies &
\slvdn{\Gamma}{n+1}{(\mathsf{ex} \ x.\ G\ x)}
\end{eqnarray*}
but this yields no real increase in expressivity, as existentials in
the body of goals can be safely transformed to outermost universal
quantifiers, while (continuing the logic programming analogy) the
above rule simply delegates the witness choice to the ambient logic
unification algorithm.

As before, the fact that provability is {inductive} yields inversion
principles as elimination rules.  For example the inversion theorem
  that analyzes the shape of a derivation ending in an atom from the
\emph{empty} context is obtained simply with a call to the standard
$\mathtt{mk\_cases}$ function, namely $\mathtt{mk\_cases}
"\slvdn{}{j}{\slAt{A}}"$ is:
\begin{eqnarray*}
  \prems{\slvdn{}{j}{\slAt{A}};\;\Forall G\ i.\, \prems{A\If G;\;
      \slvdn{}{i}{G};\; j = Suc\ i} \Implies P} \Implies P
\end{eqnarray*}

The adequacy of the encoding of the SL can be established adapting the
analogous proof in~\cite{McDowell01}.  To do so, we overload the decoding
function in several ways.  First, we need to decode terms of
types other than $\expr$.  For example, decoding terms of type
$\expr\fs\expr$ is required for most OLs.  For Mini-ML, we also need
to decode terms of type \textit{tp}.  The decoding is extended in the
obvious way.  For example, for decoding second-order terms, we define
$\decode{\Sigma}{\ilam{x}{E\ x}} = \ilam{x}{\decode{\Sigma,x}{E\ x}}$.
Second, to decode both goals and clauses, we extend $\Sigma$ to allow
both first- and second-order variables.  We can then extend the
decoding so that if $G$ is a term of type $\oo$ with free variables in
$\Sigma$, then $\decode{\Sigma}{G}$ is its translation to a formula of
minimal logic, and if $\Gamma$ is a set of terms of type $\atm\
\mathit{set}$, then $\decode{\Sigma}{\Gamma}$ is its translation to a
set of atomic formulas of minimal logic.  In addition, we restrict the
form of the definition of \ikw{prog} so that every clause of the
inductive definition is a closed formula of the form:
$$\Forall \Sigma \Bigl(\prems{\abstr E_1;\,\ldots;\,\abstr E_n}
\Implies (A \If G)\Bigr)$$ where $\Sigma$ is a set of variables
including at least $E_1,\ldots,E_n$, each of type $\expr\fs\expr$,
with $n\geq 0$.  To obtain a theory in minimal logic that corresponds
to the definition of \ikw{prog}, we decode each clause to a formula of
minimal logic of the form $\forall
\Sigma(\decode{\Sigma}{G}\to\decode{\Sigma}{A})$. For SL adequacy, we
also need to introduce two conditions, which become additional proof
obligations when establishing OL adequacy.  They are:
\begin{enumerate}
\item It is only ever possible to instantiate universal quantifiers in
\ikw{prog} clauses with terms for which the decoding is defined.
\item For every term $E ::\expr\fs\expr$ used to instantiate universal
quantifiers in \ikw{prog} clauses, $(\abstr{E})$ holds.
\end{enumerate}
The latter will follow from the former and the fact that for
\emph{all} terms $E ::\expr\fs\expr$ for which the decoding is
defined, $(\abstr{E})$ holds.

\begin{lem}[Soundness and completeness of the encoding of the
    specification logic]
Let \ikw{prog} be an inductive definition of the restricted form
described above, and let $\Pi$ be the corresponding theory in minimal
logic.   
Let $G$ be a formula of type $\oo$ and let $\Gamma$ be a set of atoms.
Let $\Sigma$ be a set of variables of type $\expr$ that contains
all the free variables in $\Gamma$ and $G$.
Then the sequent $\proper\Sigma\Implies\slvd {\Gamma}{G}$ has a
minimal derivation in \HOL (satisfying conditions 1 and 2 above) if
and only if there is a derivation of
$\iseqp{\Sigma;\decode{\Sigma}{\Gamma}}{\decode{\Sigma}{G}}$ according
to the rules of Figure~\ref{fig:minlogbc}.
\label{lem:sladeq}
\end{lem}
\begin{proof}
The proof of the forward direction follows directly by induction on
the minimal derivation of $\proper\Sigma\Implies\slvd {\Gamma}{G}$.
Compositionality (Lemma~\ref{le:comp}) is needed for the case when
$\slvd {\Gamma}{G}$ is proved by the last clause of
Figure~\ref{fig:nseq}.
The proof of the backward direction is by
direct induction on the derivation of
$\iseqp{\Sigma;\decode{\Sigma}{\Gamma}}{\decode{\Sigma}{G}}$.
Compositionality (Lemma~\ref{le:comp}) and conditions 1 and 2 are
needed for the $\ir{bc}$ case.
\end{proof}

\begin{goal}[Structural Rules]
  The following rules are admissible:
\begin{enumerate}
\item Height weakening: $ \prems{\slvdn{\Gamma}{n}{G};\; n < m}
  \Implies \slvdn{\Gamma}{m}{G} $.\footnote{This lemma turns out to be
  fairly useful, as it permits manipulation as appropriate of the height
  of two sub-derivations, such as in the $\land_R$ rule.}
\item Context weakening: $ \prems{\slvdn{\Gamma}{n}{G};\; \Gamma\subseteq
    \Gamma'} \Implies \slvdn{\Gamma'}{n}{G} $.
\item Atomic cut: $ \prems{\slvd{A,\Gamma}{G};\;
    \slvd{\Gamma}{\slAt{A}}} \Implies \slvd{\Gamma}{G}$.
\end{enumerate}
\label{mcthm:struct}
 \end{goal}
 
\begin{proof}\mbox{}
   \begin{enumerate}
   \item The proof, by structural induction on sequents, consists of
     a one-line call to  an automatic tactic using
     the elimination rule for successor (from the \HOL library) and
     the introduction rules for the sequent calculus.  
   \item By a similar fully automated induction on the structure of
     the sequent derivation,  combining 
      resolution on the sequent introduction rules with simplification
   in order to discharge some easy
     set-theoretic subgoals.
   \item Atomic cut is a corollary of the following lemma:
     \begin{eqnarray*}
       \prems{\slvdn{A,\Gamma}{i}{G};\,
         \slvdn{\Gamma}{j}{\slAt{A}}}& \Implies& \slvdn{\Gamma}{i+j}{G} 
     \end{eqnarray*}
easily proved by complete induction on the height of the derivation
of $\slvdn{A,\Gamma}{i}{G}$.  
The whole proof consists of two dozen instructions, with very little
ingenuity required from the human collaborator.

   \end{enumerate}
\qed\end{proof}

\subsection{The Object Logic}
\label{ssec:prog}

Recall the rules for call-by-value operational semantics
($\eval{e}{v}$) and type inference\ ($\Gamma \vd e \hastype \tau$)
given in Figure~\ref{fig:dyn-st}.  The  subject reduction for
this source language is stated as usual.
\begin{thm}[Subject Reduction]
\label{thm:infsubred}  
If $\eval{e}{v}$ and $\vd e\hastype\tau$, then $\vd
v\hastype\tau$.
\end{thm}
\begin{proof}
  By structural induction on evaluation and inversion on typing, using
  weakening and a substitution lemma in the \textbf{ev\_app} and
  \textbf{ev\_fix} cases.
\qed\end{proof}

We now return to the encoding of the OL, this time using the SL to
encode judgments.  The encoding of OL syntax is unchanged.  (See
Section~\ref{using}.)  Recall that it involved introducing a specific
type for $\con$.  Here, we will also instantiate type $\atm$ and
predicate \ikw{prog}.  In this section and the next, we now also
make full use of the definitions and theorems in both \hybrid and the
SL layers.

Type $\atm$ is instantiated as expected, defining the atomic formulas
of the OL\@.
\[ \begin{array}{rrcl}
 \mbox{\textit{datatype}}& \atm & = &
 \istermtwo{uexp} \bnfalt \eval{\uexp}{\uexp} \bnfalt \uexp\hastype\htp
\end{array} \] 
The clauses for the OL deductive systems are given as rules of
the  inductive definition \ikw{prog}  in Figure~\ref{fig:prog1} (recall the notation $ \_ \If \_$).
\begin{figure}[bt]
 \begin{eqnarray*}
\hline\\
 \mathit{Inductive}\ \_ \If \_ &\Hoftype&  \fsprems{\atm, \oo}\fs \bool\\
& \Implies & {\istermtwo{E_1\ \oat\ E_2}}\If
  \slAnd{\slAt{\istermtwo{E_1}}}{\slAt{\istermtwo{E_2}}} \\
\prems{\abstr E} & \Implies & \sbc{\istermtwo{\llFun{x}{E\ x}}}
 {\slAll{x}{(\istermtwo{x})} \ \ikw{imp}\ {\slAt{\istermtwo{(E\  x)}}}}\\
\prems{\abstr E } & \Implies & \sbc{\istermtwo{\llrec{E\ x}}}
 {\slAll{x}{(\istermtwo{x})} \ \ikw{imp}\ {\slAt{\istermtwo{(E\  x)}}}}
\vsk
\prems{\abstr E_1'} & \Implies & {\eval{E_1\ \oat\ E_2}{V}}\If\\
& &     \slAnd{\slAt{\eval{E_1}{\llFun{x}{E'_1\
 x}}}}{\slAt{\eval{E_2}{V_2}}}\ \ikw{and}\  \slAt{\eval{(E'_1\ V_2)}
 {V}}\\ 
\prems{\abstr E} & \Implies & \sbc{\eval{\llFun{x}{E\ x}}{\llFun{x}{E\ x}}}
{\slAt{\istermtwo{(\llFun{x}{E\ x})}}}\\
 \prems{\abstr E } & \Implies & \sbc{\eval{\llrec{E\ x}}{V}}{
              \slAt{\eval{E\ (\llrec{E\ x})}{V}} \ \ikw{and}\
              \slAt{{\istermtwo{(\llrec{E\ x})}}}}
\vsk
&\Implies & (E_1\ \oat\ E_2) \hastype T\If
 \slAt{E_1\hastype (T'\fsp T)} \ \ikw{and}\ \slAt{E_2\hastype T'}\\
\prems {\abstr E } &\Implies & (\llFun{x} {E \ x}) \hastype ( T \fsp
 T') \If\slAll {x}{(x\hastype T)} \ \ikw{imp} \ {\slAt{(E\ x)\hastype T'}}\\
\prems{\abstr E} &\Implies & (\llrec{E \ x}) \hastype T\If 
 \slAll{x}{(x \hastype T)} \ \ikw{imp}\ {\slAt{(E\  x)\hastype T}}
 \\[.5ex]\hline
\end{eqnarray*}
\caption{OL clauses: encoding of well-formedness, evaluation and
  typing.}
  \label{fig:prog1} 
\end{figure}
Recall that the encoding of evaluation in Figure~\ref{fig:hoas-eval1}
and the encoding of the \ikw{isterm} predicate for adequacy purposes
both used inductive definitions.  Here we define them both at the SL
level along with the OL level typing judgment.
Note that no explicit variable context is needed for this version of
\ikw{isterm}.
They are handled implicitly by the contexts of atomic
assumptions of the SL, resulting in a more direct encoding.
As before, in the evaluation clauses, there are no \ikw{proper}
assumptions and two \ikw{isterm} assumptions.  Neither kind of
assumption appears in the clauses for the typing rules.  None is
required to prove the analogue of MC-Lemma~\ref{mclem:evalproper} for
both evaluation and typing.
\begin{mclemma}[eval\_proper, eval\_isterm, hastype\_proper, hastype\_isterm]
$$\begin{array}{rl}
1. & \slvde{\slAt{\eval{E}{V}}} \Implies \proper{E}\land \proper{V}
\\
2. & \slvde{\slAt{\eval{E}{V}}} \Implies
     \slvde{\slAt{\istermtwo{E}}}\land\slvde{\slAt{\istermtwo{V}}}\\
3. & \slvde{\slAt{E\hastype T}} \Implies \proper{E}\\
4. & \slvde{\slAt{E\hastype T}} \Implies \slvde{\slAt{\istermtwo{E}}}
\end{array}$$
\label{mclem:twoleveladeq}
\end{mclemma}
\begin{proof}
  All the proofs are by standard induction on the given derivation,
  except the last one, whose statement needs to be generalized as
  follows:
  \begin{eqnarray*}
    \prems{\forall E,T.\; (E\hastype T)\in\Gamma \limp (\istermtwo E)
      \in \Gamma';\  
      \slvdn \Gamma i {\slAt{E\hastype T}} } 
    & \Implies &  \slvdn {\Gamma'} i {\slAt{\istermtwo E}}
  \end{eqnarray*}
  \qed\end{proof} With the new version of \ikw{isterm}, we restate the
Validity and Completeness of Representation lemmas
(Lemmas~\ref{lem:valid} and~\ref{le:complete}). Let $\Gamma$ be the
set $ \{x_1: \uexp,\dots,x_n:uexp\}$ and let $\overline{\Gamma}$ be
the set of atoms $\{\istermtwo{x_1},\ldots,\istermtwo{x_n}\}$.

\begin{lem}[Two-level Validity of Representation]
If\ $\Gamma\vdash e$, then the following is provable in \HOL:
\begin{eqnarray*}
\proper \Gamma &\Implies&
 \slvd{\overline{\Gamma}}{\slAt{\istermtwo{\encode{\Gamma}{e}}}}  
\end{eqnarray*}
\end{lem}

\begin{lem}[Two-level Completeness of Representation]
If there is a minimal derivation in \HOL of
$\proper \Gamma \Implies
\slvd{\overline{\Gamma}}{\slAt{\istermtwo{E}}}$,
then $\decode\Gamma E$ is defined and yields a Mini-ML expression $t$
such that $\Gamma\vdash t$ and $\encode\Gamma{\decode \Gamma E} =
E$. Furthermore, $\decode\Gamma{\encode \Gamma t} = t$.
\label{lem:comp2level}
\end{lem}
We will skip the statement and proof of two-level adequacy of
the other OL judgments, hoping that the reader will spot the
similarity with the above two lemmas.
Note that, although we do not state it formally, condition 1 of
Lemma~\ref{lem:sladeq} follows from completeness lemmas such as
Lemma~\ref{lem:comp2level}.  The \ikw{isterm} and \ikw{abstr}
assumptions added to the clauses of Figure~\ref{fig:prog1} are exactly
the ones needed to establish this fact for this OL\@.

We remark again that the combination of \hybrid with the use of an
SL allows us to simulate definitional
reflection \cite{Halnass91} via the built-in elimination rules of the \ikw{prog}
inductive definition \emph{without} the use of additional axioms.  For
example the inversion principle of the function typing rule is:
\beq \prems{(\llFun x (E\ x) \hastype \tau) \If G ; \Forall F\ T_1\ 
  T_2.\, \prems{\abstr F; G = \slAll {x}{\slImp{(x \hastype
        T_1)}{\slAt{(F\ x) \hastype T_2}})};\\
    \llambda {E} = \llambda {F}; \tau = (T_1\Imp T_2) } \Implies P }
\Implies P
\eneq

Before turning to the proof of Theorem~\ref{thm:infsubred}, we first
illustrate the use of this encoding with the following simple OL
typing judgment.
\begin{mclemma}
\label{le:sl-ex}
$\exists T.\slvde{\slAt{\llFun{x}{\llFun{y}{x\ \oat\ y}}\hastype T}}$
\label{lem:OLsimple}
\end{mclemma}
\begin{proof}
  This goal is equivalent to: $\exists T.\exists
  n.\slvdn{\emptyset}{n}{\slAt{\llFun{x}{\llFun{y}{x\ \oat\
          y}}\hastype T}}$.  It can be proved fully automatically by a
  simple tactic described below.  Here, we describe the main steps in
  detail to acquaint the reader with the OL/SL dichotomy and in
  particular to show how the two levels interact.
We use the instantiations for $T$ and $n$ that would be
generated by the tactic and show:
$$\slvdn{\emptyset}{8}{\slAt{\llFun{x}{\llFun{y}{x\ \oat\ y}}\hastype 
(\ogr\fsp\ogr)\fsp(\ogr\fsp\ogr)}}.$$
We apply the last rule of the SL in Figure~\ref{fig:nseq},
instantiating the first premise with the OL clause from
Figure~\ref{fig:prog1} encoding the $\mathbf{tp\_fun}$ rule for typing
abstractions, leaving two premises to be proved:
$$\begin{array}{l}
(\llFun{x}{\llFun{y}{x\ \oat\ y}}) \hastype (\ogr\fsp\ogr)\fsp(\ogr\fsp\ogr)
 \If\slAll {x}{(x\hastype \ogr\fsp\ogr)} \ \ikw{imp} \ 
 {\slAt{\llFun{y}{x\ \oat\ y}\hastype
     \ogr\fsp\ogr}}; \\
\slvdn{\emptyset}{7}{\slAll {x}{(x\hastype \ogr\fsp\ogr)} \ \ikw{imp} \ 
 \slAt{\llFun{y}{x\ \oat\ y}\hastype
     \ogr\fsp\ogr}}.
\end{array}$$
The first now matches directly the clause in Figure~\ref{fig:prog1}
for $\mathbf{tp\_fun}$, resulting in the proof obligation
$(\abstr{\lambda x.\llFun{y}{x\ \oat\ y}})$ which is handled
automatically by \ikw{abstr\_tac} discussed in Section~\ref{using}.
To prove the second, we apply further rules of the SL to obtain the
goal:
$$\prems{\proper{x}} \Implies 
\slvdn{\{x\hastype \ogr\fsp\ogr\}}{5}{\slAt{\llFun{y}{x\ \oat\ y}\hastype
     \ogr\fsp\ogr}}.$$
We now have a subgoal of the same ``shape'' as the original theorem.
Repeating the same steps, we obtain:
$$\prems{\proper{x};\proper{y}} \Implies 
\slvdn{\{x\hastype \ogr\fsp\ogr,y\hastype \ogr\}}{2}
{\slAt{x\ \oat\ y\hastype\ogr}}.$$
Along the way, the proof obligation $(\abstr{\lambda y.x\ \oat\ y})$ is
proved by \ikw{abstr\_tac}.  The assumption $(\proper{x})$ is needed
to complete this proof.  At this point, we again apply the SL
backchain rule using the OL clause for $\mathbf{tp\_app}$, obtaining
two subgoals, the first of which is again directly provable from the
OL definition.  The second:
$$\prems{\proper{x};\proper{y}} \Implies 
\slvdn{\{x\hastype \ogr\fsp\ogr,y\hastype \ogr\}}{1}
{\slAt{x\hastype \ogr\fsp \ogr} \ \ikw{and}\ \slAt{y\hastype \ogr}}.
$$
is completed by applying the rules in Figure~\ref{fig:nseq} encoding
the $\with_R$ and $\ir{init}$ rules of the SL\@.
\qed\end{proof}

The code for the
\ikw{2lprolog\_tac} tactic
automating this proof and others involving OL goals using the
SL\@ is a simple modification of the standard $\mathtt{fast\_tac}$ tactic:
\begin{verbatim}
fun 2lprolog_tac defs i =
        fast_tac(HOL_cs addIs seq.intrs @ prog.intrs
                (simpset() addSolver (abstr_solver defs))) i;
\end{verbatim}
It is based on logic programming style depth-first search
(although we could switch to breadth-first or iterative deepening)
using a small set of initial axioms for the core of higher-order logic
(\verb+HOL_cs+), the rules of the SL (\verb+seq.intrs+) and of the OL
(\verb+prog.intrs+).  Additionally, it also employs simplification
augmented with \ikw{abstr\_tac} as discussed in Section~\ref{using}.

Now we have all the elements in place for a formal HOAS
proof of Theorem~\ref{thm:infsubred}.
Note that while a substitution lemma for typing plays a central
role in the informal subject reduction proof, here, in the HOAS tradition, it
will be subsumed by the use of the cut rule on the hypothetical encoding of the
typing of an abstraction.
\begin{goal}[OL\_subject\_reduction]
  $$\forall n.\, \slvdn{}{n}{\slAt{\eval{E}{V}}} \Implies(\forall T.\,
  \slvde{\slAt{E\hastype T}} \to  \slvde{\slAt{V\hastype T}})$$
\label{mcthm:olsr}
\end{goal}
\begin{proof}
  The proof is by complete induction on the
  height of the derivation of evaluation.  It follows closely the
  proofs in~\cite{Felty02,McDowell01}, although those theorems are for
  the lazy $\lambda$-calculus, 
  while here we consider eager evaluation.  Applying meta-level
  introduction rules and induction on $n$, we obtain the sequent:
$$\prems{IH;\,\,\slvdn{}{n}{\slAt{\eval{E}{V}}},\,\,
  \slvde{\slAt{E\hastype T}}} \Implies \slvde{\slAt{V\hastype T}}$$
where $IH$ is the induction hypothesis:

$$\forall m < n.\ E,\ V.\ \slvdn{}{m}{\slAt{\eval{E}{V}}} \longrightarrow
  (\forall T.\ 
\slvde{\slAt{E\hastype T}} \longrightarrow \slvde{\slAt{V\hastype
    T}}).$$
Since the right side of the SL sequent in the middle hypothesis is an
atom and the left side is
empty, any proof of this sequent must end with the last rule of the SL
in Figure~\ref{fig:nseq}, which implements the \textbf{bc} rule.
Also, since the right side is an evaluation judgment, backchaining
must occur on one of the middle three clauses of the OL in
Figure~\ref{fig:prog1}, thus breaking the proof into three cases.  In
the formal proof, we obtain these three cases by applying standard
inversion tactics:
\begin{eqnarray}\bprems IH[i+1/n];\,\,\abstr{E'_1};\,\,
 (\slvdn{}{i}
  {\slAnd{\slAt{\eval{E_1}{\llFun{x}{E'_1\ x}}}}
  {\slAnd{\slAt{\eval{E_2}{V_2}}}{\slAt{\eval{(E'_1\
	  V_2)}{V}}}}});\nonumber\\
\slvde{\slAt{(E_1\ \oat\ E_2)\hastype T}} \eprems
\Implies  \slvde{\slAt{V \hastype T}} \label{eq:appcase}\\[10pt]
\prems{IH[i+1/n];\,\,\abstr{E};\,\,
\slvdn{}{i}{\slAt{\istermtwo{(\llFun{x}{E\ x})}}};\,\,
\slvde{\slAt{(\llFun{x}{E\ x})\hastype T}}} \nonumber \\
\Implies \slvde{\slAt{(\llFun{x}{E\ x})\hastype T}}\nonumber \\[10pt]
\bprems IH[i+1/n];\,\,\abstr{E};\,\,
  (\slvdn{}{i}
   {\slAnd{\slAt{\eval{E\ (\llrec{E\ x})}{V}}}
          {\slAt{\istermtwo{(\llrec{E\ x})}}}});~~  \nonumber \\
  \slvde{\slAt{(\llrec{E \ x})\hastype T}} \eprems
\Implies  \slvde{\slAt{V \hastype T}} \nonumber
\end{eqnarray}
where $IH[i+1/n]$ denotes $IH$ with the single occurrence of $n$
replaced by $i+1$.  The theorems mentioned earlier about injectivity
and distinctness of the constructors $\ikw{fun}$, $\ikw{\oat}$,
and $\ikw{fix}$ are used by the inversion tactics.  In contrast, in
the proof in~\cite{Felty02}, because these constructors were not
defined inductively, specialized inversion theorems were proved from
axioms stating the necessary injectivity and distinctness properties,
and then applied by hand.  The second subgoal above is directly
provable.  We illustrate the first one further.
Applying inversion to both the third and fourth hypotheses
of the first subgoal, the subgoal reduces it to:
$$\begin{array}{rcl}
\bprems IH[i+3/n];\,\,\abstr{E'_1};\,\,
 \slvdn{}{i+1}{\slAt{\eval{E_1}{\llFun{x}{E'_1\ x}}}};\,\,
 \slvdn{}{i+1}{\slAt{\eval{E_2}{V_2}}};\,\,\\
 \slvdn{}{i}{\slAt{\eval{(E'_1\ V_2)}{V}}};
\slvde{\slAt{E_1\hastype T'\fsp T}};\,\,
\slvde{\slAt{E_2\hastype T'}} \eprems\\
 \Implies  \slvde{\slAt{V \hastype T}}.
\end{array}$$
It is now possible to apply the induction hypothesis to the typing and
evaluation judgments for $E_1$ and $E_2$ to obtain:
$$\begin{array}{rcl}
\bprems IH[i+3/n];\,\,\abstr{E'_1};\,\,
 \slvdn{}{i+1}{\slAt{\eval{E_1}{\llFun{x}{E'_1\ x}}}};\,\,
 \slvdn{}{i}{\slAt{\eval{E_2}{V_2}}};\,\,
 \slvdn{}{i}{\slAt{\eval{(E'_1\ V_2)}{V}}};\ldots;~~ \\
~~\slvde{\slAt{\llFun{x}{E'_1\ x}\hastype T'\fsp T}};\,\,
\slvde{\slAt{V_2\hastype T'}}\eprems\\
 \Implies  \slvde{\slAt{V \hastype T}}.
\end{array}$$
We can now apply inversion to the hypothesis with the arrow typing
judgment involving both
the $\ikw{fun}$ constructor of the OL and the $\mathsf{all}$
constructor of the SL\@.  Inversion at the OL level gives:
$$\begin{array}{rcl}
\bprems IH[i+3/n];\,\,\abstr{E'_1};\,\,
 \slvdn{}{i+1}{\slAt{\eval{E_1}{\llFun{x}{E'_1\ x}}}};\,\,
 \slvdn{}{i}{\slAt{\eval{E_2}{V_2}}};\,\,
 \slvdn{}{i}{\slAt{\eval{(E'_1\ V_2)}{V}}};\ldots;~~\\
~~\slvde{\slAt{V_2\hastype T'}};\,\,
\abstr{E};\,\,
\llambda{E}=\llambda{E'_1};\,\,
\slvde{\slAll{x}(\slImp{x\hastype T'}{\slAt{(E\ x)\hastype T}})}
\eprems\\
 \Implies  \slvde{\slAt{V \hastype T}}.
\end{array}$$
The application of the inversion principle \texttt{prog.mkH$\_$cases}
similar to the one from Section~\ref{using} is evident here.
MC-Theorem~\ref{thm:inj} can be applied to conclude that
$E=E'_1$.  Applying inversion at the SL level gives:
$$\begin{array}{rcl}
\bprems IH[i+3/n];\,\,\abstr{E};\,\,
 \slvdn{}{i+1}{\slAt{\eval{E_1}{\llFun{x}{E\ x}}}};\,\,
 \slvdn{}{i}{\slAt{\eval{E_2}{V_2}}};\,\,
 \slvdn{}{i}{\slAt{\eval{(E\ V_2)}{V}}};\ldots;~~\\
~~\slvde{\slAt{V_2\hastype T'}};\,\,
\forall x.\, (\proper{x}\limp
\slvde{\slImp{x\hastype T'}{\slAt{(E\ x)\hastype T}}})
\eprems\\
\Implies  \slvde{\slAt{V \hastype T}}.
\end{array}$$
Inversion cannot be applied directly under the universal
quantification and implication of the last premise, so we prove the
following inversion lemma, which is also useful for the $\mathbf{fix}$
case of this proof.
\begin{eqnarray}
  \label{eq:ol-prop}
\prems{\forall x.\, \proper{x}\limp
 \slvdn{\Gamma}{i}{(\slImp{x\hastype T_1}{\slAt{(E\ x)\hastype T_2}})}}
&\Implies&
\exists j.\, i=j+1\land\\\nonumber
&&
\forall x.\, \proper{x}\limp\\\nonumber
&&
(\slvdn{x\hastype T_1,\Gamma}{j}{\slAt{(E\ x)\hastype T_2}})
\end{eqnarray}
From this lemma, and the fact that $(\proper{V_2})$ holds by
MC-Lemma~\ref{mclem:evalproper}, we obtain:
$$\begin{array}{rcl}
\bprems IH[i+3/n];\,\,\abstr{E};\,\,
 \slvdn{}{i+1}{\slAt{\eval{E_1}{\llFun{x}{E\ x}}}};\,\,
 \slvdn{}{i}{\slAt{\eval{E_2}{V_2}}};\,\,
 \slvdn{}{i}{\slAt{\eval{(E\ V_2)}{V}}};\ldots;~~\\
~~\slvde{\slAt{V_2\hastype T'}};\,\,
 (\slvdn{(V_2\hastype T')}{j}{\slAt{(E\ V_2)\hastype T}})
\eprems\\
 \Implies  \slvde{\slAt{V \hastype T}}.
\end{array}$$
Applying the cut rule of MC-Theorem~\ref{mcthm:struct} allows
us to conclude $\slvde{\slAt{(E\ V_2)\hastype T}}$.  We can then
complete the proof by applying the induction hypothesis a third time
using this fact and $\slvdn{}{i}{\slAt{\eval{(E\ V_2)}{V}}}$.
\qed\end{proof}

A key point in this section, perhaps worth repeating, is that the
clauses for typing are not inductive and would be rejected in an
inductive-based proof assistant, or at best, asserted with no
guarantee of consistency. Here, instead, the typing rules are
encapsulated into the OL level (the \ikw{prog} predicate) and executed
via the SL, so that OL contexts are implicitly represented as SL
contexts. Therefore, we are able to reproduce full HOAS proofs, at the
price of a small degree of indirectness---the need for an interpreter
(the SL) for the \ikw{prog} clauses (the OL).  One may argue that
this seems at first sight a high price to pay, since we lose the possibility of attacking the
given problem directly within the base calculus and its
tools. However, very simple tactics, including a few \emph{safe} additions to
\HOL's default simplifier and rule set\footnote{In Isabelle a rule is
  considered \emph{safe} roughly if it does not involve backtracking
  on instantiation of unknowns.} make the use of the SL in OL proofs
hardly noticeable, as we explain next.

\subsection{Tactical support}
\label{ssec:tac}
We chose to develop \hybrid as a package, rather than a stand-alone
system mainly to exploit all the reasoning capabilities that a mature
proof assistant can provide: decision procedures, rewrite rules,
counter-model checking, extensive libraries, and support for
interactive theorem proving.  Contrast this with a system such as
Twelf, where proofs are manually coded and post-hoc checked for
correctness. Moreover, in Twelf as well as in Abella, any domain
specific knowledge has to be coded as logic programming theories and
all the relevant theorems proven about them.\footnote{Twelf does have
  \emph{constraint} domains such as the rationals, but those are
  currently incompatible with totality checking, making meta-proofs
  very hard to trust.}  At the same time, our aim is to try to retain
some of the conciseness of a language such as LF, which for us means
hiding most of the administrative reasoning concerning variable
binding and contexts.  Because of the ``hybrid'' nature of our
approach, this cannot be completely achieved, but some simple-minded
tactics go a long way toward mechanizing most of boilerplate
scripting.  We have already explained how to use specific tactics to
recognize \emph{proper} terms and \emph{abstractions}. Now, we can
concentrate on assisting two-level reasoning, which would otherwise be
encumbered by the indirection in accessing OL specifications via the
SL\@.  Luckily, Twelf-like reasoning\footnote{In Abella this is even
  more apparent.} consists, at a high-level, of three basic steps:
inversion, which subsumes instantiation of (meta-level) eigenvariables
as well as (case) analysis on the shape of a given judgment,
backchaining (filling, in Twelf's terminology) and recursion. This
corresponds to highly stereotyped proof scripts that we have
abstracted into:
\begin{enumerate}
\item an \emph{inversion} tactic \ikw{defL\_tac}, which goes
  through the SL inverting on the \textbf{bc} rule and applies as an
  elimination rule one of the OL clauses.  This is complemented by the
  eager application of other \emph{safe} elimination rules (\viz
  invertible SL rules such as conjunction elimination). This
  contributes to keeping the SL overhead to a minimum;
\item a dual \emph{backchaining} tactic \ikw{defR\_tac}, that calls
  \textbf{bc} and the applicable \ikw{prog} rule. The latter is the
  basic single step into the tactic \ikw{2lprolog\_tac}, which
  performs automatic depth first search (or other searches supported
  by Isabelle) on Prolog-like goals;
\item a \emph{complete induction} tactic, to be fired when given the
  appropriate derivation height by the user and yielding as additional
  premise the result of the application of the IH.
\end{enumerate}

\subsection{A Variation}
\label{ssec:variation}
As mentioned, the main reason to explicitly encode a separate notion
of provability is the intrinsic incompatibility of induction with
non-stratifiable hypothetical judgments. On the other hand, as
remarked in \cite{Momigliano03fos}, our definition of OL evaluation,
though it exploits \hybrid's HOAS to implement OL substitution, makes
no use of hypothetical judgments.  In fact, our encoding in
Figure~\ref{fig:hoas-eval1} showed that it is perfectly acceptable to
define evaluation of the OL at the meta-level.
Now, we can give a modified version of this definition using the new
\ikw{isterm} defined at the SL level.  The new definition is given in
Figure~\ref{fig:hoas-eval2}.
\begin{figure}[t]
  \begin{center}
    \[
\begin{array}{rcl}
\hline\\
 \mathit{Inductive}\ \eval{ \_} { \_} &\Hoftype & \fsprems{\uexp, \uexp} \fs \bool  \\
 \prems{\eval{E_1}{\llFun{x}{E'\ x}};\; \eval{E_2}{V_2};\;
          \eval{(E'\ V_2)}{V};\;          \abstr{E'}}
 & \Implies &  \eval{(\llApp{E_1} {E_2})}{V}  \\
   \prems{\slvd{}{\slAt{\istermtwo{(\llFun{x}{E\ x})}}};\; \abstr{E}}  & \Implies &
   \eval{\llFun{x}{E\ x}}{\llFun{x}{E\ x}}\\ 
   \prems{\eval{E\ (\llrec{E\ x})}{V};\; \slvd{}{\istermtwo{(\llrec{E\ x})}};\; \abstr{E}}
   & \Implies & \eval{\llrec{E\ x}}{V}
 \\[.5ex]\hline
\end{array}
\]
    \caption{Alternate HOAS encoding of big step evaluation}
    \label{fig:hoas-eval2}
  \end{center}
\end{figure}
Moreover, it is easy to show (formally) that the encoding in
Figure~\ref{fig:hoas-eval2} is equivalent to the one in
Figure~\ref{fig:prog1}:
\begin{goal}
  $\eval{E}{V}$ if and only if $\slvdn{}{n}{\slAt{\eval{E}{V}}}$.
\end{goal}
\begin{proof}
  Left-to right holds by straightforward structural induction on
  evaluation using introduction rules over sequents and \ikw{prog}
  clauses. The converse is a slightly more delicate complete induction
  on the height of the derivation, requiring some manual  instantiations.
\qed\end{proof}

The same remark applies also to hypothetical and parametric judgments, provided
they are stratified (see the previously cited definition of
applicative bisimulation).  This suggests that we can, in this case,
take a different approach from McDowell \& Miller's
architecture~\cite{McDowell01} and
opt to delegate to the OL level \emph{only} those judgments, such as
typing, that would not be inductive at the meta-level. This has the
benefit of limiting the indirectness of using an explicit SL\@.
Moreover, it has the further advantage of replacing complete induction
with structural induction, which is better behaved from a proof-search
point of view.  Complete induction, in fact, places an additional
burden on the user by requiring him/her to provide the correct
instantiation for the height of the derivation in question, so that
the inductive hypothesis can be fired. While this is not an
intellectual issue, it often limits the possibility of a complete, \ie,
without user intervention, mechanization of a proof via the automatic
tools provided by the proof assistant.

As it turns out, this approach is again reminiscent of a fairly old
idea from the theory of logic programming, namely the
\emph{amalgamation} of object and meta-language as initially suggested
in \cite{metalp82}, where clauses can be written interspersing
ordinary Prolog predicates with calls to a specific meta-interpreter of
the $\mathtt{demo}$ sort. This clearly also pertains to goals, \ie, in
our setting, theorems: subject reduction at the meta-level
(\ie, amalgamated subject reduction) has the form:
\begin{goal}[meta\_subject\_reduction]
$$\eval{E}{V}\Implies\forall T.\, (\slvde{\slAt{E\hastype T}}) \limp
  (\slvde{\slAt{V\hastype T}})$$
\end{goal}
\begin{proof}
The proof is similar but slightly simpler than the proof of
MC-Theorem~\ref{mcthm:olsr}.  Instead of complete induction,
we proceed by structural induction on the evaluation judgment, which
breaks the proof into three cases.  We again consider the
application case:
$$\begin{array}{rcl}
\bprems IH_1;\,\,IH_2;\,\,IH_3;\,\,\abstr{E'_1};\,\,
  (\eval{E_1}{\llFun{x}{E'_1\ x}});\,\,
  (\eval{E_2}{V_2});\,\,\\
  (\eval{(E'_1\ V_2)}{V});\,\,
\slvde{\slAt{(E_1\ \oat\ E_2)\hastype T}} \eprems
 \Implies \slvde{\slAt{V \hastype T}} 
\end{array}$$
where $IH_1$, $IH_2$, and $IH_3$ are the following three induction
hypotheses:
$$\begin{array}{ll}
IH_1:\quad & \forall T. \slvde{\slAt{E_1\hastype T}}
\Implies \slvde{\slAt{(\llFun{x}{E_1'\ x})\hastype T}} \\
IH_2:\quad & \forall T. \slvde{\slAt{E_2\hastype T}}
\Implies \slvde{\slAt{V_2\hastype T}} \\
IH_3:\quad & \forall T. \slvde{\slAt{(E'_1\ V_2)\hastype T}}
\Implies \slvde{\slAt{V\hastype T}}
\end{array}$$
This subgoal corresponds to subgoal~(\ref{eq:appcase}) in the proof of
MC-Theorem~\ref{mcthm:olsr}, with several differences.  For
instance, subgoal~(\ref{eq:appcase}) was obtained by an application of
complete induction followed by inversion on the OL and SL, while the
above subgoal is a direct result of applying structural induction.
Also, although both subgoals have three evaluation premises,
in~(\ref{eq:appcase}) they are inside conjunction at the SL level.
Finally, the general induction hypothesis $IH$ on natural numbers
in~(\ref{eq:appcase}) is replaced by three induction hypotheses here,
generated from the premises of the meta-level definition of the
evaluation rule for application.  The remaining steps of the proof
of this case are essentially the same as the steps for
MC-Theorem~\ref{mcthm:olsr}.  Inversion on the typing
judgment is used exactly as before since in both proofs, typing is
expressed via the SL\@.  Also, the three induction hypotheses in this
proof are used to reach the same conclusions as were obtained using
the single induction hypothesis three times in the previous proof.
\qed\end{proof}

Now that we have seen some proofs of properties of OLs, we
can ask what the minimal set of theorems and tactics is that the
two-level architecture needs from \hybrid.  The answer is: very little.
Essentially all we need is the quasi-freeness properties of the
\hybrid type, which are inherited from the OL:
\begin{itemize}
\item clash rules to rule out impossible cases in elimination rules;
\item injectivity facts, all going back to \emph{abstr\_lam\_simp} to simplify
  equations of the form $\llambda{E}=\llambda{F}$ for
  second-order functions $E$ and $F$;
\item an abstraction solver.\footnote{Again, this is not needed anymore
  in a newer version of \HOL and of our package~\cite{MMF07}.} 
\end{itemize}

The reader may find in~\cite{Momigliano03fos} other examples, such as
the verification of properties of compilation, of encoding OLs
using inductive predicates (types) at the meta-level for all
stratifiable object-level judgments.  However, this style of reasoning
is viable only when there is a substantial coincidence between the
meta-logical properties of the SL and the ambient (meta-) logic.  Were such
properties to clash with an encoding that could benefit from being
driven by a more exotic logic, then \emph{all} OL predicates will have
to be embedded as \ikw{prog} clauses. This, it may be argued, is a
relatively small price to pay for the possibility of adopting an SL
that better fits the logical peculiarities of 
interesting OLs, as we investigate next.

\section{Ordered Linear Logic as a Specification Logic}
\label{sec:olli}
\renewcommand{\coln}{\mid} 
\renewcommand{\sbc}[3]{#1 \If #2 \coln #3}
\renewcommand{\elist}{[ \ ]} 

In this section we aim to show the flexibility of the two-level
architecture by changing SL in order to have a better match with the
encoding on hand; the case-study we consider here is the operational
semantics of a \emph{continuation}-based abstract machine, where evaluation is
sequentialized: an instruction is executed in the context of a
{continuation} describing the rest of the computation and eventually
returning an {answer}.  We will adopt an \emph{ordered logical
  framework} (OLF)~\cite{Polakow01phd}. The general methodology
consists of refining a logical framework in a \emph{conservative} way,
so as to capture different object-level phenomena at the right level
of abstraction. Conservativity here guarantees that if a new feature
(such as order) is not required, it does not interfere with the
original system.

Although frameworks based on {intuitionistic} logic have been fairly fruitful, it so happens that the \emph{structural} properties of the
framework, namely weakening, contraction and exchange, are inherited
by the object-level encodings.  We have argued that one of
the keys to the success of an encoding lies in the ability of
specifying judgments ``in-a-context'' exploiting the context of the SL
itself; however those properties may not always be appropriate for
every domain we want to investigate.  Another  case in point is the
meta-theory of languages with imperative features, where the notion of
(updatable) \emph{state} is paramount.
It has been frequently observed that an elegant representation of the
store may rely on a \emph{volatile} notion of context.  \emph{Linear
  logic} is then the natural choice, since it offers a notion of
context where each assumption must be used exactly once; a declarative
encoding of store update can be obtained via linear operations that,
by accessing the context, consume the old assumption and insert the
new one. This is one of the motivations for proposing frameworks based
on linear logics (see~\cite{miller04llcs} for an overview) such as
Lolli~\cite{Lolli}, Forum~\cite{Forum}, and LLF~\cite{CervesatoP02}, a
\emph{conservative} extension of LF with multiplicative implication,
additive conjunction, and unit.  Yet, at the time of writing this
article, work on the \emph{automation} of reasoning in such frameworks
is still in its infancy~\cite{McCreight03} and may take other
directions, such as hybrid logics \cite{Reed08}.  The
literature offers only a few formalized meta-theoretical
investigations with linear logic as a framework, an impressive one
being the elegant encoding of type preservation of Mini-ML with
references (MLR) in LLF~\cite{CervesatoP02}. However, none of them
comes with anything like a formal certification of correctness that
would make people \emph{believe} they are in the presence of a proof.
Encoding in LLF lacks an analogue of Twelf's totality
checker. Moreover this effort may be reserved to LLF's extension, the
Concurrent Logical Framework~\cite{clf}.  A \foldn proof of a similar
result is claimed in \cite{McDowell01}, but not only the proof is not
available, but it has been implemented with Eriksson's Pi, a proof
checker~\cite{Eriksson94} for the theory of partial inductive
definitions, another software system that seems not to be available
anymore.

This alone would more than justify the use of a fragment of linear
logic as an SL on top of \hybrid, whose foundation, we have argued, is
not under discussion. 
However, we want to go beyond the logic of
state, towards a logic  of \emph{order}.
In fact, a continuation-based abstract machine follows an order, \viz
a stack-like discipline; were we able to also \emph{internalize} this
notion, we would be able to simplify the presentation, and hence, the
verification of properties of the continuation itself, taking an
additional step on the declarative ladder.  Our contribution here to
the semantics of continuation machines is, somewhat paradoxically, to
dispose of the notion of continuation itself via internalization in an
ordered context, in analogy with how the notion of state is realized
in the linear context.  In particular, the ordered context is used to
encode directly the stack of continuations to be evaluated, rather
than building an explicit stack-like structure to represent a
continuation.  While this is theoretically non-problematic, it
introduces entities that are foreign to the mathematics of the problem
and which bring their own numerous, albeit trivial, proof
obligations.\footnote{This is not meant to say that intuitionistic
  meta-logic, (full) HOAS and list-based techniques cannot cope with
  mutable data: in fact, significant case studies have been tackled: for
  example, Crary and Sarkar's proof of soundness for foundational
  certified code in typed assembly language for the x86
  architecture~\cite{CraryS03} as well the more recent attempt by Lee
  \etal \cite{lee+:towards} to verify an internal language for full SML.  }
 Further and more importantly, machine states can be  mapped not into OL
data, but OL \emph{provability}.

Ordered (formerly known as non-commutative) linear logic
\cite{Polakow99mfps} combines reasoning with unrestricted, linear and
ordered hypotheses.  Unrestricted (\ie, intuitionistic) hypotheses may
be used arbitrarily often, or not at all regardless of the order in
which they were assumed.  Linear hypotheses must be used exactly once,
also without regard to the order of their assumption.  Ordered
hypotheses must be used exactly once, subject to the order in which
they are assumed.

This additional expressive power allows the logic to handle directly
the notion of \emph{stack}.  Stacks of course are ubiquitous in
computer science and in particular when dealing with abstract and
virtual machines.  OLF has been previously applied to the meta-theory
of programming languages, but only in paper and pencil proofs: Polakow
and Pfenning~\cite{Polakow00lfm} have used OLF to formally show that
terms resulting from a CPS translation obey "stackability" and
linearity properties~\cite{Danvy99hoots}.  Polakow and
Yi~\cite{PolakowYi01flops} later extended these techniques to
languages with exceptions.  Remarkably, the formalization in OLF
provides a simple proof of what is usually demonstrated via more
complex means, \ie, an argument by logical relations.
Polakow~\cite{Polakow01phd} has also investigated proof-search and
defined a first-order logic programming language with ordered
hypotheses, called \emph{Olli}, based on the paradigm of abstract
logic programming and \emph{uniform} proofs, from which we draw
inspiration for our ordered SL, \ie, a second-order minimal ordered
linear sequent calculus.

We exemplify this approach by implementing a fragment of Polakow's
ordered logic as an SL and test it with a proof of type preservation
of a continuation machine for Mini-ML, as we sketched in
\cite{MomiglianoP03}.  For the sake of presentation we shall deal with
a call-by-name operational semantics.
It would not have been be unreasonable to use MLR as a test case,
where all the three different contexts would play a part. However,
linearity has already been thoroughly studied, while we wish to
analyze ordered assumptions in isolation, and for that aim, a basic
continuation machine will suffice (but see \cite{MartinCase} for a
thorough investigation of the full case).  Further, although the SL
implementation handles all of second-order Olli and in particular
proves cut-elimination for the whole calculus, we will omit references
to the (unordered) linear context and linear implication,
as well as to the ordered left implication, since they do not play any
role in this case-study.

\subsection{Encoding the Specification Logic}
\label{ssec:oseq}

\begin{figure}[t]
\[ \fboxsep=7mm\fbox{$\begin{array}{c}
 \ian{}{\seq{\G}{A}{A}}{\ir{init}_{\O}}
 \qquad\qquad
 \ian{}{\seq{\G, A }{\cdot}{A}}{\ir{init}_\G}
 \\[1em]
 \ian{\seq{(A,\G)}{\O}{G}}
     {\seq{\G}{\O}{A\Imp G}}
     {\Imp_R} 
\qquad\qquad
 \ian{\seq{\G}{{(\O, A)} }{G}}
      {\seq{\G}{\O}{A \roimp G}}
      {\roimp_R} 
\\[1em]
 \ibn{\seq{\G}{\O}{G_1}}{\seq{\G}{\O}{G_2}} 
    {\seq{\G}{\O}{G_1\with G_2}}{\with_R}\\[1em]
 \ian{}{\seq{\G}{\O}{\top}}{\top_R}
\qquad\qquad
 \ian{\seq{\G}{\O}{G[a/x]}}{\seq{\G}{\O}{\forall x\ldot
     G}}{\forall_R^a}  \\[1em] 
\ian{
  \begin{array}{l}
(\sbc{A}{[G_1,\ldots,G_m]}{  [G'_1,\ldots,G'_n]}) \in [\Pi]\\
 \seq{\G}{\cdot}{G_1} \:\ldots\: \seq{\G}{\cdot}{G_m}\\
  {\seq{\G}{{\O}_1}{G'_1} \:\ldots\: \seq{\G}{{\O}_n}{G'_n}}
  \end{array}}
      {\seq{\G}{{\O}_n\ldots{\O}_1}{A}}
      {\ir{bc}}
\end{array}$} \]
\caption{Sequent rules for \ollim\label{fig:oseq}}
\end{figure}

We call our specification logic \ollim, as it corresponds to the
aforementioned fragment of Olli, where '$\roimp$' denotes
right-ordered implication. We follow \cite{MomiglianoP03} again in
representing the syntax as:
\[
\begin{array}{rccl}
\mbox{Goals} & G & \bnfas & A \bnfalt
                            A \to G \bnfalt A \roimp G \bnfalt 
                            G_1 \with G_2 \bnfalt \top \bnfalt 
                            \forall^\tau x\ldot G
\\
\mbox{Clauses} & P & \bnfas & \forall(\sbc{A}{[G_1,\ldots,G_m]}{[G'_1,\ldots,G'_n]})
\end{array}
\]
The body of a clause $\forall(\sbc{A}{[G_1,\ldots,G_m]}{
  [G'_1,\ldots,G'_n]})$ consists of two \emph{lists}, the first one of
intuitionistic goals, the other of ordered ones. It represents the
``logical compilation'' of the formula $\forall(G_m\to\ldots\to G_1 \to
G'_n \roimp\ldots\roimp G'_1 \roimp A)$.  We choose this
compilation to emphasize that
 if one views the calculus as a non-deterministic logic programming
interpreter, the latter would solve subgoals from innermost to
outermost.  Note also that this notion of clause makes additive conjunction
useless, although we allow it in goals for a matter of style and
consistency with the previous sections.

Our sequents have the form:
$$\seq{\G}{\O}{G}$$
where ${\Pi}$ contains the program clauses, which are unrestricted
(\ie, they can be used an arbitrary number of times), $\G$ contains
unrestricted atoms, ${\O}$ contains ordered atoms and $G$ is the
formula to be derived. Contexts are {lists} of hypotheses, where
we overload the comma to denote adjoining an element to a list at both
ends. To simplify matters further, we leave eigenvariable signatures
implicit.
One may think of the two
contexts as one big context where the ordered hypotheses are in a fixed
relative order, while the intuitionistic ones may float, copy or
delete themselves.  The calculus is depicted in Figure~\ref{fig:oseq}.
Again in this fragment of the logic, implications have only atomic
antecedents.  There are obviously two implication introduction rules,
where in rule $\roimp_R$ the antecedent $A$ is appended to the
\emph{right} of $\O$, while in the other rule we have $(A,\G)$, but it
could have been the other way around, since here the order does not
matter. Then, we have all the other usual right sequent rules to break down the
goal and they all behave additively. Note how the $\top_R$ rule can be
used in discharging any unused ordered assumptions.  For atomic goals
there are two initial sequent rules, for the leaves of the derivation:
$\ir{init}_\O$ enforces linearity requiring $\O$ to be a singleton
list, while $\ir{init}_\G$ demands that all ordered assumptions
have been consumed. Additionally, there is a single backchaining rule
that simultaneously chooses a program formula to focus uponand derives all the ensuing subgoals; rule $(\ir{bc})$ is applied
provided there is an \emph{instance} $\sbc{A}{[G_1 \ldots G_m]}{[G'_1
  \ldots G'_n]}$ of a clause in the program $\Pi$.
Note that the rule assumes that every program clause
must be placed to the \emph{left} of the ordered context.  This assumption is
valid for our fragment of the logic because it only contains right
ordered implications ($\roimp$) and the ordered context is restricted
to atomic formulas. Furthermore, the ordering of the ${\O}_i$ in the
conclusion of the rule is forced by our compilation of the program
clauses. We leave to the keen reader the task to
connect  formally our backchain rule to the focused uniform proof system of
\opcit \cite{Polakow01phd}.

We encode this logical language extending the datatype from
Section~\ref{ssec:nseq} with right implication, where again outermost
universal quantifiers will be left implicit in clauses.
\[ \begin{array}{rrcl}
 \mbox{\textit{datatype}}& \oo & = &  \cdots\bnfalt \atm\roimp  \oo
\end{array} \]

Our encoding of the \ollim sequent calculus 
uses three mutually inductive definitions,
motivated by the compilation of the body of clauses into additive and
multiplicative lists:\footnote{Note that $\G$ could have easily been a
set, as in Section~\ref{sec:2lev}.}
\[\begin{array}{l}
  \slvdn{\Gamma\coln\Omega}{n }{G} \Hoftype
  \fsprems{\ilist{\atm},\ilist{\atm}, \nat,\oo}\fs\bool\\
 \qquad\mbox{goal $G$ has an ordered linear derivation from $\G$ and
    $\O$ of height $n$}\\
  \islvdn{\Gamma}{n}{Gs} \Hoftype
  \fsprems{\ilist{\atm}, \nat,\ilist{\oo}}\fs\bool\\
\qquad\mbox{list
    of goals $Gs$ is additively provable from $\G$ \etc}\\  
  \oslvdn{\Gamma}{\Omega}{n}{Gs}\Hoftype
  \fsprems{\ilist{\atm},\ilist{\atm}, \nat,\ilist{\oo}}\fs\bool \\
\qquad\mbox{list of goals $Gs$ is multiplicatively consumable given
    $\G$ and $\O$ \etc}
 \end{array}\]
 The rendering of the first judgment is completely
 unsurprising,\footnote{As a further simplification, the encoding of
   the $\forall_R$ rule will \emph{not} introduce the \ikw{proper}
   assumption, but the reader should keep in mind the fact that
   morally every eigenvariable is indeed proper.} except, perhaps, for
 the backchain rule, which calls the list predicates  required to
 recur
  on the body of a clause:
$$
\prems{(\sbc{A}{O_L}{I_L}) \; ;\; \oslvdn{\Gamma}{\Omega}{n}{O_L}\; ;\;
  \islvdn{\Gamma}{n}{I_L}} \Implies \slvdn{\Gamma\coln\Omega\ }{n+1 }{\slAt{A}}
$$ 
The notation $\sbc{A}{O_L}{I_L}$ corresponds to the inductive
definition  of a set \ikw{prog} this time of type $\fsprems{\atm,
  \ilist{\oo}, \ilist{\oo}}\fs \bool$, see Figure~\ref{fig:tph}.
Backchaining uses the two list judgments to encode, as we
anticipated,  execution of the
(compiled) body of the focused clause. Intuitionistic list
provability is just an additive recursion through the list of
intuitionistic subgoals:
\[
\begin{array}{rcl}
  & \Implies & \islvdn{\Gamma}{n}{\elist}\\
  \prems{ \slvdn{\Gamma\coln\elist}{n}{G}\; ; \;
  \islvdn{\Gamma}{n}{Gs}}
&\Implies & \islvdn{\Gamma}{n+1}{(G, Gs)}
 \end{array}
\]
Ordered list consumption involves an analogous recursion, but it behaves
multiplicatively \wrt the ordered context. Reading the rule bottom
up, the current ordered context $\Omega$ is non-deterministically
split into two ordered parts, one for the head $\Omega_G$ and one
$\Omega_R$ for the rest of the list of  subgoals.
\[
\begin{array}{l}
\qquad\qquad   \Implies  \oslvdn{\Gamma}{\elist}{n}{\elist}\\[0.3em]
\prems{ \osplit{\Omega}{\Omega_R}{\Omega_G}\; ;\;
 \slvdn{\Gamma\coln\Omega_G}{n}{G}\; ; \;
  \oslvdn{\Gamma}{\Omega_R}{n}{Gs}}\\
\qquad\qquad  \Implies \oslvdn{\Gamma}{\Omega}{n+1}{(G,  Gs)}
\end{array}
\]
Therefore the judgment relies on the inductive definition of a
predicate for order-preserving splitting of a context. This corresponds to
the usual logic programming predicate $\mathit{append}(\O_R,\O_G,\O)$
called with mode $\mathit{append}(-,-,+)$.
\[
\begin{array}{rcl}
  & \Implies &  \osplit{\Omega}{\elist}{\Omega}\\
  \osplit{\Omega_1}{\Omega_2}{\Omega_3}  & \Implies &
\osplit{(A,  \Omega_1)}{(A,  \Omega_2)}{\Omega_3}
\end{array}
\]
The rest of the sequent rules are encoded similarly to the previous SL
(Figure~\ref{fig:nseq}) and the details are here omitted (and left to
the web appendix of the paper, see \url{hybrid.dsi.unimi.it/jar}).
Again we define $\slvd{\Gamma\coln\Omega}{G}$ iff there exists an $n$
such that $\slvdn{\Gamma\coln\Omega}{n}{G}$ and simply $\slvde{G}$ iff
$\slvd{\elist\coln\elist}{G}$.  Similarly for the other judgments.

\begin{goal}[Structural Rules]
\label{mcg:str}
  The following rules are admissible:
\begin{itemize}
\item Weakening for numerical bounds:
  \begin{enumerate}
  \item $ \prems{\slvdn{\Gamma\coln\Omega}{n}{G};\; n < m} \Implies
    \slvdn{\Gamma\coln\Omega}{m}{G} $
  \item $\prems{\oslvdn{\Gamma}{\Omega}{n}{Gs};\; n < m} \Implies
    \oslvdn{\Gamma}{\Omega}{m}{Gs}$
  \item $\prems{\islvdn{\Gamma}{n}{Gs};\; n < m} \Implies
    \islvdn{\Gamma}{m}{Gs}$.
   \end{enumerate}   
 \item Context weakening, where $(set\ \Gamma)$ denotes the set
   underlying the context $\G$.
   \begin{enumerate}
   \item $ \prems{\slvd{\Gamma\coln\Omega}{G};\; set\
       \Gamma\subseteq set\ \Gamma'} \Implies
     \slvd{\Gamma'\coln\Omega}{G} $
   \item $ \prems{\oslvd{\Gamma}{\Omega}{Gs};\; set\
       \Gamma\subseteq set\ \Gamma' } \Implies \oslvd{\Gamma'}{
       \Omega}{Gs} $
   \item $ \prems{\islvd{\Gamma}{Gs};\; set\
       \Gamma\subseteq set\ \Gamma'}
     \Implies \islvd{\Gamma'}{Gs} $.
   \end{enumerate}
   \item Intuitionistic atomic cut:
     \begin{enumerate}
     \item $ \prems{\slvdn{\Gamma\coln\Omega} i {G} ;\; set\ \Gamma
         = set\ (A, \Gamma');\; \slvdn{\Gamma'\coln\elist} j {\slAt A}}
       \Implies \slvdn {\Gamma'\coln\Omega} {i + j} G $.

     \item $ \prems{\oslvdn{\Gamma}{\Omega} i {Gs} ;\; set\ \Gamma
         = set\ (A, \Gamma');\; \slvdn{\Gamma'\coln\elist} j {\slAt A}}
       \Implies \oslvdn {\Gamma'}{\Omega} {i + j} Gs $.

  \item $ \prems{\islvdn{\Gamma} i {Gs} ;\; set\ \Gamma
         = set\ (A, \Gamma');\; \slvdn{\Gamma'\coln\elist} j {\slAt A}}
       \Implies \islvdn {\Gamma'} {i + j} Gs $.

     \end{enumerate}
     \end{itemize}
\end{goal}
\begin{proof}
  All the proofs are by mutual structural induction on the three
  sequents judgments. For the two forms of weakening, all it takes is a
  call to \HOL's classical reasoner. Cut requires a little care in the
  implicational cases, but nevertheless it does not involve more then
  two dozens instructions.
\qed\end{proof}

Although the sequent calculus in~\cite{Polakow01phd} enjoys other
forms of cut-elimination, the following:
\begin{mccorollary}[seq\_cut]
  \begin{eqnarray*}
\prems{\slvd{A, \Gamma\coln\Omega}{G};\;
  \slvd{\Gamma}{\slAt{A}}} & \Implies & \slvd{\Gamma\coln\Omega}{G}  
\end{eqnarray*}
\end{mccorollary}
is enough for the sake of the type preservation proof
(MC-Theorem~\ref{thm:sr}). Further, admissibility of
contraction and exchange for the intuitionistic context is a
consequence of context weakening.

\subsection{A Continuation Machine and its Operational Semantics}
\label{sec:case}

 We avail ourselves of  the continuation machine for Mini-ML formulated
 in~\cite{Pfenning01book} (Chapters 6.5 and 6.6), which we refer to
 for motivation and additional details. We use the same language and
 we repeat it here for convenience: \[
\begin{array}{r@{\qquad}ccl}
 \mbox{\textit{Types}} & \tau & \bnfas & \mathit{i} \bnfalt \tau\fsp \tau' \\
  \mbox{Expressions} & e & \bnfas & x\bnfalt \ir{fun}\,x\ldot e \bnfalt
  e_1\ \at\ e_2
  \bnfalt \rec e
\end{array}
\]

The main judgment $s\hra s'$ (Figure~\ref{fig:Kopsem}) describes how
the state of the machine evolves into a successor state $s'$
in a \emph{small-step} style. The machine selects an {expression} to be
executed and a \emph{continuation} $K$, which contains all the
information required to carry on the execution.  To achieve this we
use the notion of \emph{instruction}, \eg, an intermediate command that
links an expression to its value. The continuation is either empty
(\ir{init}) or it has the form of a stack ($ K;\ilam{x}{i}$), each
item of which (but the top) is a \emph{function} from values to instructions.
Instruction (\ir{ev}  e) starts the first step of the computation, while
(\ir{return} v) tells the current continuation to apply to the top element
on the continuation stack the newly found value.  Other instructions
sequentialize the evaluation of subexpressions of constructs with more
than one argument; in our language, in the case of application, the
second argument is postponed until the first is evaluated completely.
This yields the following categories for the syntax of the machine:
\[
\begin{array}{r@{\qquad}ccl}
\mbox{Instructions} & i & \bnfas & \ir{ev}\,e \bnfalt \ir{return}\,v \bnfalt
\ir{app}_1 \,v_1\,e_2 \\[1ex]
\mbox{Continuations} & K & \bnfas & \ir{init} \bnfalt K;\ilam{x}{i} \\[1ex]
\mbox{Machine States} & s & \bnfas & K\:\diamond\: i \bnfalt \ir{answer}\, v
\end{array}
\]

\begin{figure}[t]
  \centering
  \[\fboxsep=3mm\fbox{$\begin{array}{rl}
\ir{st\_init} \;\;::\;\; & \ir{init}\:\diamond\:\ir{return}\,v \;\hra\;\ir{answer}\,v
\\[1em] 
\ir{st\_return} \;\;::\;\; & K;\ilam{x}{i}\:\diamond\:\ir{return}\,v 
\;\hra\; K\:\diamond\:i[v/x]
\\[1em]
\ir{st\_fun} \;\;::\;\; & K\:\diamond\:\ir{ev}\,(\ir{fun}\,x\ldot e) \;\hra\; 
K\:\diamond\: \ir{return}\,(\ir{fun}\, x\ldot e)
\\[1em]
\ir{st\_fix} \;\;::\;\; & K\:\diamond\:\ir{ev}\,(\ir{fix}\,x\ldot e) \;\hra\; 
K\:\diamond\: \ir{ev}\,(e[\ir{fix}\, x\ldot e/x])
\\[1em]
\ir{st\_app} \;\;::\;\; & K\:\diamond\:\ir{ev}\,(e_1\ \at\ e_2) \;\hra\; 
K;\ilam{x_1}{\ir{app}_1\,x_1\,e_2} \:\diamond\: \ir{ev}\,e_1
\\[1em]
\ir{st\_app_1} \;\;::\;\; & K\:\diamond\: \ir{app}_1\,(\ir{fun}\,x\ldot e)\,e_2 \;\hra\; 
K \:\diamond\: \ir{ev}\,e[e_2/x]\\[1em]
\end{array}
$}\]
  \caption{Transition rules for machine states}
  \label{fig:Kopsem}
\end{figure}

\begin{figure}[t]
  \centering
  \[\fboxsep=3mm\fbox{$ \begin{array}{c}
      \ianc{}{s \hra^* s}{\ir {stop}}\qquad\qquad
      \ibnc{s_1 \hra s_2}{s_2 \hra^* s_3}{s_1 \hra^* s_3}{\ir {step}}\\[1em]
      \ianc{\ir{init}\:\diamond\:\ir{ev}\,e \;\hra^*\ir{answer}\,v}
{\cev        e v }{\ir {cev}}
\end{array}$}\]
  \caption{Top level transition rules}
  \label{fig:cev}
\end{figure}

The formulation of the subject reduction property of this machine
follows the statement in~\cite{CervesatoP02}, although we consider
sequences of transitions by taking the reflexive-transitive closure
$\hra^*$ of the small-step relation, and a top level initialization
rule \textbf{cev} (Figure~\ref{fig:cev}).  Of course, we need to add
typing judgments for the new syntactic categories, namely
instructions, continuations and states. These can be found in
Figure~\ref{fig:tp}, whereas we refer the reader to
Figure~\ref{fig:dyn-st} as far as typing of expressions goes.

\begin{thm}   $K\:\diamond\: i \;\hra^*\; \ir{answer}\,v$ and $\G\vd_i
  i\hastype\tau_1$ and $\vd_K K\hastype \tau_1\Imp\tau_2$ implies
  $\cdot\vd_e v\hastype\tau_2$.
\end{thm}
\begin{proof}
  By induction on the length of the execution path using inversion
  properties of the typing judgments.
\qed\end{proof}
\begin{cor}[Subject Reduction]
  $\cev e v$ and $\cdot\vd_e e\hastype\tau$ entails $\cdot\vd_e
  v\hastype\tau$.
\end{cor}

As a matter of fact we could have obtained the same result by showing
the \emph{soundness} of the operational semantics of the continuation
machine \wrt big step evaluation, \viz that $\cev e v$ entails $\eval
e v$ (see Theorem 6.25 in \cite{Pfenning01book}) and then appealing to
type preservation of the latter. That would be another interesting
case study: the equivalence of the two operational semantics
(thoroughly investigated by Pfenning in Chapter 6 \opcit but in the
intuitionistic setting of LF), to gauge what the ``Olli'' approach
would buy us.

\begin{figure}[b]
\[\fboxsep=3mm\fbox{$\begin{array}{c}
  \hspace*{-6em}
   \ian{\G\vd_e e\hastype \tau}{\G\vd_i\ir{ev}\,e \hastype \tau}{ofI\_\ir{ev}}
 \qquad\qquad\qquad\qquad\qquad
   \ian{\G\vd_e v\hastype \tau}{\G\vd_i\ir{return}\,v \hastype \tau}{ofI\_\ir{return}}
  \\[1em]
  \hspace*{-6em}
 \ibn{\G\vd_e e_1\hastype \tau'\Imp\tau}
      {\G\vd_e e_2\hastype\tau'}
      {\G \vd _i\ir{app}_1\,e_1\,e_2 \hastype \tau }
      {ofI\_\ir{app}_1} 
\mbox{}\\[1em].\dotfill\\[1em]
\prooftree
\mbox{}\justifies {\vd_K\ir{init} \hastype \tau\Imp\tau}
\using{ofK\_\ir{init}}\endprooftree
 \qquad\qquad
\qquad
\prooftree
{x\oftp\tau_1\vd_i i\hastype \tau}\qquad
    {\vd_K K\hastype\tau\Imp\tau_2}
\justifies
    {\vd_K K;\ilam{x}{i} \hastype \tau_1\Imp\tau_2 }
\using    {ofK\_\ir{cont}}
\endprooftree    \\[2em]
.\dotfill\\[1em]
\ibnc{\vd_i i \hastype \tau_1} {\vd_K K \hastype \tau_1\ \fsp \ \tau_2}{\vd_s K\:\diamond\: i \hastype  \tau_2}{ofs\_\diamond}
 \qquad\qquad
\ianc{\vd_e v\hastype \tau }{\vd_s \ir{answer}\, v\hastype
  \tau}{ofs\_\ir{answer}}
\end{array}$}\]
\caption{Typing rules for the continuation machine\label{fig:tp}}
 \end{figure}

\subsection{Encoding the Object Logic}
\label{ssec:oprog}

We now show how to write the operational semantics of the continuation
machine as an Olli program, or more precisely as \ollim OL clauses.  
Rather than representing the continuation $K$ an explicit stack, we
will simply \emph{store instructions in the ordered context}. This is
particularly striking as we map machine states not into OL
data, but OL \emph{provability}.  In particular we will use the
following representation to encode machine states:
\[
K\:\diamond\:i \qquad \mbox{ $\leadsto$} \qquad 
\slvd{\elist\coln\rep{K}}{\slAt{\exec {\rep{i}}}}
\]
where $\rep{K}$ is the representation, described below, of the
continuation (stack) $K$ and $\rep{i}$ the obvious representation of
the instruction.\footnote{The reader may be relieved to learn that, at this
  late stage of the paper, we will be much more informal with the
  issue of the adequacy of this encoding, mainly trying to convey the
  general intuition. This is also notationally signaled by dropping
  the somewhat heavy notation $\encode \cdot \cdot$ for the lighter
  $\rep \cdot$. It is likely that the faithfulness of our
  representation could be obtained following the approach
  in~\cite{CervesatoP02}---see in particular Theorem $3.4$ \ibid.}  In fact,
if we retain the usual abbreviation
\begin{eqnarray*}
    \uexp &\idef &con\ \expr
\end{eqnarray*}
 the encoding of instructions can be simply realized with an \HOL
datatype, whose adequacy is standard: \[ \begin{array}{rrcl} \mbox{\textit{datatype}}& \mathit{instr} & = &
  \ikw{ev} \ \uexp\bnfalt\ikw{return} \ \mathit{uexp}\bnfalt
  \ikw{app_1} \ \mathit{uexp}\ \uexp
\end{array} \] 

\renewcommand{\Val}{\uexp} 

To describe the encoding of continuations,
we  use our datatype \textit{atm},  which describes the atomic formulas
of the OL\@. This time, it is more interesting and consists of:
\[ \begin{array}{rrcl}
\mbox{\textit{datatype}}& \atm & = & \ceval{\uexp}{\Val}\bnfalt
\exec{\instr}\bnfalt\init{\Val}\\
& & &\bnfalt\cont{(\Val\fs\instr)}\bnfalt
 \of{\uexp}{\htp}\\
& & &\bnfalt\ofI{\instr}{\htp}
\bnfalt\ofK{\htp}
\end{array} \] 
We have \emph{atoms} to describe the initial continuation ``$\init{}$'' of
type $\Val\fs\atm$, the continuation that simply returns its
value. Otherwise $K$ is an ordered context of atoms ``$\cont K$'' of
type  $(\Val\fs\instr)\fs\atm$.  The top level of evaluation
($\ceval{\rep e}{\rep v}$)  unfolds to 
the initial goal $\init{\rep v} \roimp \exec{(\EV {\rep e)}}$; our program will
evaluate the expression $\rep e$ and instantiate $\rep v$ with the resulting
value.  In other words, we  evaluate $e$ with the
initial continuation. The  other instructions are treated as follows:
the goal $\exec{(\ret {\rep v})} $ means: pass $v$ to
the top continuation on the stack (\ie, the rightmost element in the
ordered context): the instruction in the goal $\exec{(\appone{\rep
    {v_1}}{\rep {e_2}})}$ sequentializes the evaluation of application.

We have the following representations of machine states:
\[
\ir{init}\:\diamond\:\ir{return}\,v 
\quad \mbox{ $\leadsto$} \quad
\slvd{\elist\coln[ \init{W}]} { \slAt {\exec{(\ret {\rep{v}})}}}
\]
where the logic variable $W$ will be instantiated to the final answer;
\[
K;\ilam{x}{i}\:\diamond\:\ir{return}\,v
\mbox{ $\leadsto$} \quad
\slvd{\elist\coln(\rep{K}, \cont{( \ilam{x}{\rep{i}})})} {\slAt{\exec{(\ret {\rep{v}})}}}
\]
where the ordering constraints force the proof of $\exec{(\ret{\rep{v}})}$ to
focus on the rightmost ordered formula.

\begin{figure}[ht]
\begin{eqnarray*}
\hline \\[.5ex]
 \mathit{Inductive}\ \sbc{\_}{\_}{\_} &\Hoftype&  \fsprems{\atm, \ilist{\oo},
    \ilist{\oo}}\fs \bool\\[0.5em]
&\Implies & \sbc{\of{(E_1\ \oat\ E_2)}{T}}{\elist}{
 [{\slAt{\of{E_1}{(T'\fsp T)}}}, {\slAt{\of{E_2}{T'}}}]}\\
\prems {\abstr E } &\Implies & \sbc{\of{(\llFun x {E\ x}) }{( T_1 \fsp
 T_2)}}{\elist}{[\slAll {x}\slImp{(\of{x}{T_1})} {\slAt{\of{(E\
 x)}{T_2}}}]}\\
\prems {\abstr E } &\Implies & \sbc{\of{(\llrec {E \ x}) }{( T)}}{\elist}{[\slAll {x}\slImp{(\of{x}{T})}  {\slAt{\of{(E\
 (\llrec {E\ x})}{T}}}]}\\
&\Implies & \sbc{\ofI{(\EV{E})}{T}}{\elist}{[\slAt{\of{E}{T}}]}\\
&\Implies & \sbc{\ofI{(\ret{V})}{T}}{\elist}{[\slAt{\of{V}{T}}]}\\
&\Implies & \sbc{\ofI{(\appone{V}{E})}{T}}{\elist}{[\slAt{\of{V}{( T_2\fsp T)}},
\slAt{\of{E}{T_2}}]}\\
&\Implies & \sbc{\ofK (T \fsp T)}{[\slAt{\init {V}}]}{\elist}\\
 &\Implies & \sbc{\ofK (T_1 \fsp T_2)}{[{\slAt{\cont
    {K}}}, {\slAt{\ofK {T\fsp T_2}}}]} \\
&& \qquad\qquad\qquad\qquad\quad\ \
{[\slAll {x}\slImp{(\of{x}{T_1})} {\slAt{\ofI{(K\ x)}{T}}}]}\\[1em]
&\Implies & \sbc{\ceval {E}{V}}{[\init {V}\roimp\exec {(\EV {E})}]}{\elist}\\
&\Implies &  \sbc{\exec{(\ret{V})}}{[\slAt{\init {V} } ]} {\elist}\\
\prems {\abstr K } &\Implies & \sbc{\exec{(\ret{V) }}}
{[{\slAt{\cont {K}}}, {\slAt{\exec{(K\ V)}}}]} {\elist} \\
\prems {\abstr E } &\Implies & \sbc{\exec{(\EV{(\llFunc {E}))}}}
        {[\slAt{\exec{(\ret{(\llFunc {E})})} }]}{\elist}\\
&\Implies &  \sbc{\exec{(\EV (E_1\ \oat\ E_2)) }}
 {[ \cont {(\lambda v.\ \appone{v}{E_2})} \roimp \slAt{\exec{( \EV{E_1})}}]} {\elist}\\
\prems {\abstr E } &\Implies & \sbc{\exec{(\appone{(\llFunc E)}{E_2}) }}
{[\slAt{\exec{( \EV {(E\ E_2)})}}]} {\elist}
 \\[.5ex]\hline
\end{eqnarray*}
\caption{\hybrid's encoding of  the OL deductive systems of the
  continuation machine\label{fig:tph}}
\end{figure}

We can now give the clauses for the OL deductive systems in
Figure~\ref{fig:tph}, starting with typing.  These judgments are
intuitionistic, except typing of continuations.  The judgments for
expressions and instructions directly encode the corresponding
judgments and derivation rules.  The judgments for continuations
differ from their analogs in Figure~\ref{fig:tp} in that there is no
explicit continuation to type; instead, the continuation to be
typed is in the ordered context. Thus, these judgments must first get
a continuation from the ordered context and then proceed to type it.
  
The evaluation clauses of the program fully take advantage of ordered
contexts.  The first one corresponds to the \textbf{cev}
rule.  The rest directly mirror the machine transition
rules.

A sample derivation is probably in order and so it follows as
MC-Lemma~\ref{mcl:silly-ceval}.  Note that as far as examples of
evaluations go, this is not far away from total triviality, being the
evaluation of something which is already a value. However, our
intention here is not to illustrate the sequentialization of
evaluation steps typical of a continuation machine (for which we refer
again to \cite{Pfenning01book}); rather we aim to emphasize the role of the
ordered context, in particular the effect of non-deterministic
splitting on the complexity of proof search.

\begin{mclemma}
\label{mcl:silly-ceval}
 $\exists V.\ \slvde{\slAt{\ceval{(\llFun{x}{x})} \ V}} $
\end{mclemma}

\begin{proof}
  After introducing  the logic variable $?V$ (here we
  pay no attention to the height of the derivation) we apply
  rule \textbf{bc}, \ie, backchaining, obtaining the following 3
  goals:
\[
\begin{array}{l}
1.\ \  \sbc {\ceval{(\llFun{x}{x})} \ ?V} {  { [\init \ ?V\roimp\exec {(\EV
        {(\llFun{x}{x})})}]}}{ \elist}\\
2.\ \ \oslvd \elist \elist { { [\init \ ?V\roimp\exec {(\EV
        {(\llFun{x}{x})})}]}}\\
3.\ \ \islvd \elist  \elist
\end{array}
\]
Goals such as the third one (the base case of intuitionistic list
evaluation) will always arise when back-chaining on evaluation, as the
intuitionistic context plays no role, \ie, it is empty; since they are
trivially true, they will be resolved away without any further
mention.  So we have retrieved the body of the relevant clause and
passed it to ordered list evaluation:

 $$\oslvd {\elist} {\elist} { [\init \ ?V\roimp\exec {(\EV
        {(\llFun{x}{x})})}]}  $$
This leads to   splitting  the ordered context, \ie,
\[
\begin{array}{l}
1.\ \  \osplit \elist { Og} { Or}\\
2. \ \ \slvd {\elist\coln  Og} { \init \ ?V\roimp\slAt{\exec {(\EV
      {(\llFun{x}{x})})}}} \\
3. \ \ \oslvd {\elist}{Or}{\elist}
\end{array}
\]
In this case, ordered splitting is deterministic as it can only match
the base case and the two resulting
contexts $ Og$ and ${ Or}$ are both set to  empty:

 $$\slvd {\elist\coln\elist} { \init \ ?V\roimp\slAt{\exec {(\EV
       {(\llFun{x}{x})})}}} $$
The introduction rule for ordered implication (and simplification) puts
the goal in the form:

 $$\slvd {\elist\coln [\init \ ?V]} {\slAt{\exec {(\EV
        {(\llFun{x}{x})})}}} $$
which  corresponds to the execution of the identity function with the
initial continuation. Another backchain yields:

\[
\begin{array}{l}
1. \ \ \abstr (\lambda x.\ x)\\
2. \ \  \osplit {[\init ?V]} { Og_{1}} { Or_{1}}\\
3. \ \ \slvd{\elist\coln { Og_{1}}} {\slAt{\exec {(\ret
       {(\llFun{x}{x})})}}} 
\end{array}
\]
As usual, \ikw{abstr\_tac} \ takes care of the first goal, while now
we encounter the first interesting splitting case. To be able to solve
the goal by assumption in the SL, we need to pass the (singleton)
context to the left context $ Og_1$. One way to achieve this is to
gently push the system  by  proving  the simple
lemma $\exists A.\ \osplit {[A]}{[A]}{\elist}$. Using the latter as an
introduction  rule for subgoal $2$, we get:

 $$\slvd {\elist\coln[\init ?V]}{} {\slAt{\exec {(\ret
       {(\llFun{x}{x})})}}}$$
More backchaining yields:

  $$\oslvd {\elist}{[\init ?V]} [{\slAt{\init {(\llFun{x}{x})}}}]$$
and with another similar ordered split to the left we have

  $$\slvd {\elist\coln[\init ?V]} {\slAt{\init {(\llFun{x}{x})}}}$$
which is true by the $\mathbf{init}_\Omega$ rule. This concludes the
derivation, instantiating $?V$ with $\llFun{x}{x}$.
\qed\end{proof}

If we collect in \texttt{sig\_def} all the definitions pertaining to
the signature in question and bundle up in \texttt{olli\_intrs} all
the introduction rules for the sequent calculus, (ordered) splitting
and the program database:
\begin{verbatim}
fast_tac(claset()  addIs olli_intrs
        (simpset() addSolver (abstr_solver sig_defs)));
\end{verbatim}
the above tactic will automatically and very quickly prove the above lemma,
by backtracking on all the possible ordered splittings, which are, in
the present case, preciously few.  However, this will not be the case
for practically any other goal evaluation, since splitting is highly
non-deterministic in so far as all the possible partitions of the
contexts need to be considered. To remedy this, we could encode a
variant of the \emph{input-output} sequent calculus described
in~\cite{Polakow01phd} and further refined in \cite{Polakow06}, which
describes efficient resource management---and hence search---in
linear logic programming.  Then, it would be a matter of showing it
equivalent to the base calculus, which may be far from
trivial. In the end, our system will do fine for its aim, \ie,
investigation of the meta-theoretic properties of our case study.

The example may have shed some light about this peculiarity:
the operational semantics of the continuation machine is small-step; 
a sequence of transitions are connected (via rules for its reflexive
transitive closure) to compute a value, whereas our implementation
looks at first sight big-step, or, at least, shows no sign of
transitive closure.  In fact, informally, for every transition that a
machine makes from some state $s_i$ to $s_{i+1}$, there is a bijective
function that maps the \emph{derivation} of $\rep{s_i}$,
\ie, the sequent encoding $s_i$ to the derivation of $\rep
{s_{i+1}}$.  The \ollim interpreter essentially simulates the informal
trace of the machine obtained by transitive closure of each step $
K\:\diamond\: i \;\hra\;s'$ for some $s'$ with a tree
of attempts to establish $\slvd {\elist\coln\rep K}{\slAt{\rep i}}$ by
appropriate usage of the available ordered resources (the rest of
$\rep K$). In the above example, the paper and pencil proof is a tree
with \textbf{cev} at the root,  linked by the \textbf{step} rule to the
\textbf{st\_fun} and \textbf{st\_init} axioms. This corresponds to the
\ollim proof we have described, whose skeleton consists of the
statement of the lemma as root and ending with the axiom
$\mathbf{init}_\Omega$.
\[
\begin{array}{l}
 \slvd {\elist\coln [\init \ ?V]} {\slAt{\exec {(\EV
        {(\llFun{x}{x})})}}}\leadsto \\
\qquad \slvd {\elist\coln[\init ?V]}{} {\slAt{\exec {(\ret
       {(\llFun{x}{x})})}}}\leadsto \\
\qquad\qquad \slvd {\elist\coln[\init ?V]} {\slAt{\init {(\llFun{x}{x})}}}
\end{array}
\]

\bigskip
Now we can address the meta-theory, namely the subject reduction
theorem:
\begin{goal}[sub\_red\_aux]
\label{thm:sr}
\begin{eqnarray*}
\lefteqn{\slvdn{\elist\coln(\init{V},\Omega) \ }{i}{\slAt{\exec{I}}}
 \Implies}\\
& &  \forall
T_1T_2. \slvd{}{\slAt {\ofI{I}{T_1}}}\longrightarrow \\
& & 
(\slvd{\elist\coln(\init{V},\Omega) }{\slAt{\ofK{(T_1\fsp T_2)}}})
\longrightarrow\slvd{}{\slAt{\of{V}{T_2}}})
\end{eqnarray*}
\end{goal}
The proof of subject reduction again follows from first principles and
does not need any weakening or substitution lemmas. The proof and
proof scripts are considerably more manageable if we first establish
some simple facts about typing of various syntax categories and
instruct the system to aggressively apply every deterministic
splitting, \eg,
$$
 \bprems \osplit {\elist} {Og} {Or};\,\, \bprems Og = \elist;\,\,Or =
\elist\eprems \Implies P \eprems \Implies P
$$
as well as a number of elimination rules stating the impossibility of
some inversions such as
$$
\bprems\sbc{\cont K}{Ol}{Il} \eprems\Implies P
$$
The human intervention that is required is limited to providing the correct
splitting of
the ordered hypotheses and selecting the correct instantiations of the heights
of sub-derivations in order to apply the IH\@.
\begin{proof}
 The proof is by complete induction on the height of the derivation of
the premise.  The inductive hypothesis  is:
$$
\begin{array}{l}
  \forall m.\ m < n \longrightarrow\\
\qquad  (\forall I\ V\ \O. \\
\qquad\qquad\slvd{\elist\coln(\init V, \O)}{\slAt{\exec I}}
\longrightarrow\\
\qquad\qquad\forall T_1\ T_2.\\
\qquad\qquad\qquad \slvd {\elist}{\slAt{\ofI I T_1}}\land\mbox{}\\
\qquad\qquad\qquad\slvd {(\init V, \O)}{\slAt{\ofK {T_1\arrow
      T_2}}}\longrightarrow \slvd{}{\slAt{\ofV{V}{T_2}}})
\end{array}
$$
Not only we will omit the IH in the following, but we will also gloss
over the actual height of the derivations, hoping that the reader will
trust \HOL to apply the IH correctly. We remark that  in contexts
we overload the comma to denote adjoining an element to a list at both
ends.

We begin by inverting on
$\slvd{\elist\coln(\init V, \O)}{\slAt{\exec I}}$
and then on the \texttt{prog} clauses defining execution, yielding
several goals, one for each evaluation clause.  The statement for the
\ir{st\_return} case is as follows:
$$
\begin{array}{l}
  \bprems \dots;\ \abstr K; \\
\qquad\slvd{\elist\coln\elist}{\slAt{\ofI {(\ret V')} T_1}};\\
\qquad\slvd {\elist\coln(\init V, \O)}{\slAt{\ofK {T_1\arrow  T_2}}};\\
\qquad\oslvd{\elist}{(\init V, \O)}{[\slAt{\exec (K\ V')},\slAt{\cont
    K}]}\eprems\\ 
\qquad\qquad\qquad\Implies
\slvd{}{\slAt{\ofV{V}{T_2}}}
\end{array}
$$
We start by applying the typing lemma:
\begin{eqnarray*}
  \slvd{\elist\coln\elist}{\slAt{\ofI {(\ret V)} T}} & \Implies &
\slvd{}{\slAt{\ofI { V} T}}
\end{eqnarray*}
Inverting of the derivation of 
$\oslvd{\elist}{\init V,  \O}{[\slAt{\exec (K\ V')},\slAt{\cont K}]}$
  yields:
$$
\begin{array}{l}
  \bprems\dots;\ \osplit {(\init V,\O)} {Og} {[\cont K]};\\
\qquad\slvd{\elist\coln Og}{\slAt{\exec  {K\ V'} }};\\
\qquad\slvd{\elist\coln \elist}{\slAt{\ofV {V'} {T_1} }}\eprems\\
\qquad\qquad\Implies
\slvd{}{\slAt{\ofV V {T_2} }}
\end{array}
$$
Now, there is only one viable splitting of the first premise, where 
$\Forall L.\ \bprems\osplit \O{L} {[\cont K]} ;\,\,  
Og = (\init{V},  L) \eprems \Implies P$,
as the impossibility of the first one, entailing $\cont K = \init V$,
is ruled out by the freeness properties of the encoding of atomic
formulas. This  results in
$$
\begin{array}{l}
   \bprems\dots;\ \slvd{\elist\coln (\init{V},  L)}{\slAt{\exec  {K\
         V'} }};\\
\qquad\slvd {\elist\coln(\init V, \O)}{\slAt{\ofK {T_1\arrow  T_2}}};\\
\qquad\slvd{\elist\coln \elist}{\slAt{\ofV {V'} {T_1} }};\\
\qquad\osplit \O{L} {[\cont K]}\eprems\\
\qquad\qquad\Implies
\slvd{}{\slAt{\ofV V {T_2} }}
\end{array}
$$
We now use the reading of ordered split as  ``reversed'' append to force
$\O$ to be the concatenation of ${L}$ and $[\cont K]$, denoted here as
in the SL logic, \eg  $(L,\cont K)$:
$$
\begin{array}{l}
   \bprems\dots;\ \slvd{\elist\coln (\init{V},  L)}{\slAt{\exec  {K\
         V'} }};\\
\qquad\slvd {\elist\coln(\init V, L, \cont K)}{\slAt{\ofK
    {T_1\arrow  T_2}}};\\ 
\qquad\slvd{\elist\coln \elist}{\slAt{\ofV {V'} {T_1} }}\eprems\\
\qquad
\qquad\qquad\Implies
\slvd{}{\slAt{\ofV V {T_2} }}
\end{array}
$$
we now invert on the typing of continuation:
$$
\begin{array}{l}
   \bprems\dots;\ \oslvd{\elist}{\elist}{[\slAll v {\ofV v {T_1}}  \
     \mathsf{imp}\  \slAt{\ofI{(K' \ v)}{T}} ] };\\ 
\qquad \slvd{\elist\coln \elist}{\slAt{\ofV {V'} {T_1} }};\\
\qquad\oslvd {\elist}{(\init V, L,\cont K)}{[\slAt{\ofK
    {T\arrow  T_2}},\slAt{\cont K'}]};\\ 
\qquad\slvd{\elist\coln (\init{V},  L)}{\slAt{\exec  {(K\ V')}
  }}\eprems\\ 
\qquad\qquad\qquad\Implies
\slvd{}{\slAt{\ofV V {T_2} }}
\end{array}
$$
The informal proof would require an application of the substitution
lemma. Instead here we use cut to infer:
$$\slvd{\elist\coln\elist}{\slAt{ \ofI {(K'\ V')} T}}$$
We first have to invert on the hypothetical statement $\slAll v {\ofV
  v {T_1}} \ \mathsf{imp}\ \slAt{\ofI{(K' \ v)}{T}}$ and instantiate
${v}$ with $V'$:
$$
\begin{array}{l}
   \bprems\dots;\ \slvd{\elist\coln \elist}{\slAt{\ofV {V'} {T_1} }};\\
\qquad\slvd{\elist\coln (\init{V},  L)}{\slAt{\exec  {(K\
         V')} }};\\

\qquad\oslvd {\elist}{(\init V, L,\cont K)}{[\slAt{\ofK
    {T\arrow  T_2}},\slAt{\cont K'}]};\\ 
\qquad\slvd{[\ofV {V'}{T_1}]\coln\elist}{\slAt{ \ofI {(K'\ V')}
    T}}\eprems\\
\qquad\Implies
\slvd{}{\slAt{\ofV V {T_2} }} 
\end{array}
$$
Now one more inversion on $ \oslvd{\elist}{(\init V, L, \cont
  K)}{[\slAt{\ofK {T\arrow T_2}},\slAt{\cont K'}]} $ brings us to
split $\osplit {(\init V, L,\cont K)}{Og}{[\cont K']}$ so that
${(\init V, L) = Og}$ and ${K = K'}$:
$$
\begin{array}{l}
  \bprems\dots;\ \slvd{\elist\coln \elist}{\slAt{\ofV {V'} {T_1} }};\\
 \qquad\slvd{\elist\coln (\init{V},  L)}{\slAt{\exec  {(K'\ V')}
   }};\\  
\qquad\slvd{\elist\coln\elist}{\slAt{ \ofI {(K'\ V')} T}};\\
\qquad\slvd {\elist\coln(\init V, L)}{\slAt{\ofK
    {T\arrow  T_2}}}\eprems\\
 \qquad
\qquad\qquad\Implies
\slvd{}{\slAt{\ofV V {T_2} }} 
\end{array}
$$
This final sequent follows from complete induction for height $i$.
\qed\end{proof}

\begin{mccorollary}[subject\_reduction]
$    \prems{\slvde{\ceval{E}{V}} ;\;\slvde{\of{E}{T}}} \Implies
    \slvde{\of{V}{T}}$
\end{mccorollary}

\section{Related Work}
\label{sec:rel}

There is nowadays extensive literature on approaches to
representing and reasoning about what we have called ``object
logics,'' where the notion of variable bindings is
paramount. These approaches are supported by implementations in the
form of proof checkers, proof assistants and theorem provers. We will
compare our approach to others according to two categories: whether
the system uses different levels for different forms of reasoning and
whether it is relational (\ie, related to proof search) or functional
(based on evaluation).

\subsection{Two-level, Relational Approaches}
\label{ssec:2lr}

Our work started as a way of porting most of the ideas of
$\foldn$~\cite{McDowell01} into the mainstream of current proof
assistants, so that they can enjoy the facilities and support that
such assistants provide.  As mentioned in the introduction, \HOL or Coq plays
the role of $\foldn$, the introduction/elimination rules of inductive
definitions (types) simulate the \emph{defR} and \emph{defL} rules of
PIDs and the \hybrid meta-language provides $\foldn$'s
$\lambda$-calculus. In addition, our approach went beyond 
$\foldn$, featuring  meta-level induction and co-induction, which were later
proved consistent with the theory of (partial)
inductive definitions~\cite{MomiglianoT03}.  These features are now
standard in $\foldn$'s successor, \emph{Linc}~\cite{Tiu04phd}.

One of the more crucial advances given by \emph{Linc}-like logic lies
in the treatment of induction over \emph{open} terms, offered by the
proof-theory of
\cite{Tiu04phd,miller05tocl}. The latter has been recently modified~\cite{Tiu07} to simplify the
theory of $\nabla$-quantification by removing local contexts of
$\nabla$-bounded variables so as to enjoy properties closer to the
\emph{fresh} quantifier of nominal logic, such as strengthening and
permutation (see later in this section).  Finally the $\cal G$ logic
\cite{gacek08lics} brings fully together PIDs and $\nabla$-quantification by
allowing the latter to occur in the head of definitions. This gives
excellent new expressive power, allowing for example to \emph{define}
the notion of freshness. Furthermore it eases induction over open
terms and even gives a logical reading to the notion of ``regular
worlds'' that are crucial in the meta-theory of Twelf.

\smallskip

Recently, \emph{Linc}-like meta-logics and the two-level approach have
received a new implementation from first principles. Firstly,
\emph{Bedwyr}~\cite{Bedwyr} is a model-checker of higher-order
specifications, based on a logic programming interpretation of
$\nabla$-quantification and case analysis. Coinductive reasoning is
achieved via \emph{tabling}, although no formal justification of the
latter is given. Typical applications are in process calculi, such as
bisimilarity in $\pi$-calculus. The already cited \emph{Abella}
\cite{Abella} is emerging a real contender in this category: it
implements a large part of the $\cal G$ logic and sports a significant
library of theories, including an elegant proof of the
\textsc{PoplMark} challenge \cite{poplmark2005} as well as a proof of
strong normalization by logical relations \cite{AbellaSOS}, an issue
which has been contentious in the theorem proving world.  This proof
is based on a notion of arbitrarily cascading substitutions, which
shares with nominal logic encodings the problem that once nominal
constants have been introduced, the user often needs to spend some
effort controlling their spread.  In fact, there is currently some
need to control occurrences of names in terms and thus to rely on
``technical'' lemmas that have no counterpart in the informal
proof. This is not a problem of the prover itself, but it is induced
by the nominal flavor that logics such as Linc's successors
LG$^\Omega$ and $\cal G$ have introduced. More details can be found
in~\cite{FeltyM09}.

\smallskip

 The so far more established competitor in the two-level relational
 approach is \emph{Twelf}~\cite{TwelfSP}.
Here, the LF type theory is used to encode OLs as judgments
and to specify meta-theorems as relations (type families) among them;
a logic programming-like interpretation provides an operational
semantics to those relations, so that an external check for totality
(incorporating termination, well-modedness,
coverage~\cite{SchurmannP03,Pientka05}) verifies that the given
relation is indeed a realizer for that theorem.  In this sense the
Twelf totality checker can be seen to work at a different level than
the OL specifications.

Hickey \etal~\cite{HickeyNYK06} built a theory for two-level reasoning
within the MetaPRL system, based on reflection.  A HOAS representation
is used at the level of reflected terms. A computationally equivalent
de Bruijn representation is also defined.  Principles of induction are
automatically generated for a reflected theory, but it is stated that
they are difficult to use interactively because of their size.  In
fact, there is little experience using the system for reasoning about
OLs.

\subsection{Two-level, Functional Approaches}
\label{ssec:2lf}

There exists a second approach to reasoning in LF that is built on the
idea of devising an explicit (meta-)meta-logic for reasoning
(inductively) about the framework, in a fully automated
way~\cite{S00}.  $\Momega$ can be seen as a constructive first-order
inductive type theory, whose quantifiers range over possibly open LF
objects over a signature.  In this calculus it is possible to express
and inductively prove meta-logical properties of an OL\@.  By
the adequacy of the encoding, the proof of the existence of the
appropriate LF object(s) guarantees the proof of the corresponding
object-level property.  $\Momega$ can be also seen as a
dependently-typed functional programming language, and as such it has
been refined first into the \emph{Elphin} programming
language~\cite{SPS:TLCA2005} and finally in \emph{Delphin}
\cite{PosSch08}. 
 $\mathrm{ATS^{LF}}$~\cite{CuiDX05} is an instantiation of Xi's
\emph{applied type systems} combining programming with proofs and can
be used as a logical framework.  In a similar vein the contextual
modal logic of Pientka, Pfenning and Naneski~\cite{NanevskiTOCL}
provides a basis for a different foundation for programming with HOAS
based on hereditary substitutions. This has been explicitly formulated
as the programming language \emph{Beluga}~\cite{Pientka10}. Because
all of these systems are programming languages, we refrain from a
deeper discussion.  See~\cite{FeltyPientka:ITP10} for a
comparison of Twelf, Beluga, and \hybrid on some benchmark examples.

\subsection{One-level, Functional Approaches}
\label{ssec:1lf}

Modal $\lambda$-calculi were formulated in the early attempts by
Sch\"{u}rmann, Despeyroux, and Pfenning \cite{SchurmannDP01} to
develop a calculus that allows the combination of HOAS with a
primitive recursion principle in the \emph{same} framework, while preserving
the adequacy of representations.
For every type $A$ there is a type $\Box A$ of {\em closed\/} objects
of type $A$.  In addition to the regular function type $A\fs B$, there
is a more restricted type $A \to B \equiv \Box A \fs B$ of
``parametric'' functions.  Functions used as arguments for
higher-order constructors are of this kind and thus roughly correspond
to our notion of abstraction.
The dependently-typed case is considered
in~\cite{Despeyroux00jflp} but the approach seems to have been abandoned
in view of \cite{NanevskiTOCL}.
Washburn and Weirich~\cite{WashburnWeirichJFP08} show how standard
first-class polymorphism can be used instead of a special modal
operator to restrict the function space to ``parametric'' functions.
They encode and reason about higher-order iteration operators.

We have mentioned earlier the work by Gordon and
Melham~\cite{Gor93,Gordon96}, which we used as a starting point for
\hybrid.
Building on this work, Norrish improves the recursion
principles~\cite{Norrish:TPHOLs04}, allowing greater flexibility in
defining recursive functions on this syntax.

\subsection{Other One-Level Approaches}
\label{ssec:oth}

\emph{Weak} \footnote{For the record, the by now standard terminology
``weak'' HOAS was coined by the second author of the present paper
in~\cite{MomTP01}.} higher-order abstract syntax~\cite{DFHtlca95} is
an approach that strives to co-exist with an inductive setting, where
the positivity condition for datatypes and hypothetical judgments must
be obeyed.  In weak HOAS, the problem of negative occurrences in
datatypes is handled by replacing them with a new type.  For example,
the $\ikw{fun}$ constructor for Mini-ML introduced in
Section~\ref{using} has type $(\underline{\var} \fs uexp)\fs uexp$,
where $\var$ is a type of variables, isomorphic to natural
numbers. \emph{Validity} predicates are required to weed out exotic
terms, stemming from case analysis on the $\var$ type, which at times
is inconvenient. The approach is extended to hypothetical judgments by
introducing distinct predicates for the negative occurrences.  Some
axioms are needed to reason about hypothetical judgments, to mimic
what is inferred by the cut rule in our architecture.  Miculan \etal's
framework~\cite{HonsellMS01,CiaffaglioneLM07,Miculan:IC01} embraces an
\emph{axiomatic} approach to meta-reasoning with weak \hoas in an
inductive setting. It has been used within Coq, extended with a
``theory of contexts'' (ToC), which includes a set of axioms
parametric to an \hoas\ signature.  The theory includes the reification
of key properties of names akin to \emph{freshness}. Exotic terms are
avoided by taking the $\var$ to be a parameter and assuming
axiomatically the relevant properties. Furthermore, higher-order induction
and recursion schemata on expressions are also assumed. To date, the
consistency with respect to a categorical semantics has been
investigated for higher-order logic~\cite{BucaloHMSH06}, rather than
\wrt a (co)inductive dependent type theory such as the one underlying
Coq~\cite{Gim98tr}.

From our perspective, ToC can be seen as a stepping stone towards 
Gabbay and Pitts \emph{nominal logic}, which aims to be a foundation of
programming and reasoning with \emph{names}, in a one-level
architecture.  This framework started as a variant of the
Frankel-Mostowski set theory based on
permutations~\cite{PittsAM:newaas-jv}, but  it is now presented as a
first-order theory~\cite{pitts03ic}, which includes primitives for
variable renaming and variable freshness, and a (derived) new
``freshness'' quantifier.  Using this theory, it is possible to prove
properties by structural induction and also to define functions by
recursion over syntax~\cite{Pitts06}.  The proof-theory of nominal
logic has been thoroughly investigated in~\cite{GabbayC04,Cheney05},
and the latter  also investigates the proof-theoretical
relationships between the $\nabla$ and the ``freshness'' quantifier, by
providing a translation of the former to the latter.

Gabbay has tried
to implement nominal sets on top of Isabelle~\cite{GabbayMJ02tphol}.  A
better approach has turned out
to be Urban \etal's; namely to engineer a \emph{nominal datatype
  package} inside Isabelle/HOL~\cite{urban05cade,Nominal} analogous to the
standard datatype package but defining equivalence classes of term
constructors. In more recent versions, principles of primitive
recursion and strong induction have been added~\cite{UrbanB06} and
many case studies tackled successfully, such as proofs by logical
relations
(see~\cite{Nominal} for more
examples).  The approach has also been compared in detail with de Bruijn
syntax~\cite{BerghoferU07} and in hindsight owes to McKinna and
Pollack's ``nameless'' syntax~\cite{McKinna99}. Nominal logic is
beginning to make its way into Coq; see~\cite{Aydemir07}.

It is fair to say that while Urban's nominal package allows the
implementation of informal proofs obeying the Barendregt convention almost
literally, a certain number of lemmas that the convention
conveniently hides must still be proved \wrt the judgment involved;
for example to choose a \emph{fresh} atom for an object $x$, one has
to show that $x$ has \emph{finite support}, which may be tricky for
$x$ of functional type, notwithstanding the aid of general tactics
implemented in the package.  HOAS, instead, aims to make
$\alpha$-conversion disappear and tries to extract the
abstract higher-order nature of calculi and proofs thereof, rather
than follow line-by-line the informal development. On the other hand,
it would be interesting to look at versions of the freshness
quantifier at the SL level, especially for those
applications where the behavior of the OL binder is not
faithfully mirrored by HOAS, namely with the traditional universal
quantification at the SL-level; well known examples of this case
include (mis)match in the $\pi$-calculus and closure-conversion in
functional programming.

 Chlipala~\cite{ChlipalaICFP08} recently
introduced an alternate axiomatic approach to reasoning with weak
HOAS\@.  Object-level terms are identified as meta-terms belonging to an
inductive type family, where the type of terms is parameterized by the
type of variables.  Exotic terms are ruled out by parametricity
properties of these polymorphic types.  Clever encodings of OLs are
achieved by instantiating these type variables in different ways,
allowing data to be recorded inside object-level variables (a
technique borrowed from~\cite{WashburnWeirichJFP08}).  Example proofs
developed with this technique include type preservation and semantic
preservation of program transformations on functional programming
languages.

\subsection{\hybrid Variants}
\label{ssec:h-v}

Some of our own related work has involved alternative versions of
\hybrid as well as improvements to \hybrid, which we describe here.

\paragraph{Constructive \hybrid.}

A \emph{constructive} version of \hybrid implemented in
Coq~\cite{CapFel07} provides an alternative that could also serve as
the basis for a two-level architecture.  This version provides a new
approach to defining induction and non-dependent recursion principles
aimed at simplifying reasoning about OLs.  In contrast
to~\cite{SPS:TLCA2005}, where built-in primitives are provided for the
reduction equations for the higher-order case, the recursion principle
is defined on top of the base de Bruijn encoding, and the reduction
equations proved as lemmas.

In order to define induction and recursion principles
for particular OLs, terms of type $expr$ are paired with proofs
showing that they are in a form that can represent an object-level
term.  A dependent type is used to store such pairs; here we omit the
details and just call it $expr'$, and sometimes oversimplify and
equate $expr'$ with $expr$.  For terms of Mini-ML for example, in
addition to free variables and bound variables, terms of the forms
$(\Con{cAPP} \app E_1 \app E_2)$,
$(\Con{cABS} \app \LAM{x}{E\ x})$ and
$(\Con{cFIX} \app \LAM{x}{E\ x})$,
which correspond to the bodies of the definitions of $\lapp$,
$\ikw{fun}$, and $\ikw{fix}$, are the only ones that can be
paired with such a proof.  Analogues of the definitions for
constructing object-level terms of type $expr$ are defined for type
$expr'$.  For example, $(e_1 \overline{\lapp} e_2)$ is defined to be
the dependent term whose first component is an application (using
$\lapp$) formed from the first components of $e_1$ and $e_2$, and
whose second component is formed from the proof components of $e_1$
and $e_2$.

Instead of defining a general \ikw{lambda} operator, a version of
\ikw{lbind} that does not rely on  classical constructs  is defined
for each OL\@.  Roughly, $(\ikw{lbind}~e)$ is obtained by applying $e$
to a \emph{new} free variable and then replacing it with de Bruijn
index $0$.  A new variable for a term $e$ of type $\expr\fs\expr$ is
defined by adding $1$ to the maximum index in subterms of the form
$(\Var{x})$ in $(e(\Bnd{0}))$. Note that terms that do not satisfy
\ikw{abstr} may have a different set of free variables for every
argument, but for those which do satisfy \ikw{abstr}, choosing
$(\Bnd{0})$ as the argument to which $e$ is applied does give an
authentic free variable.  Replacing free variable $(\Var{n})$ in
$(e(\Var{n}))$ with $(\Bnd{0})$ involves defining a substitution operator that increases bound indices as appropriate as it descends
through \ikw{ABS} operators.  This description of \ikw{lbind} is
informal and hides the fact that these definitions are actually given
on dependent pairs, \ie, $e$ has type $expr'\to expr'$.
Thus, the definition of \ikw{lbind} depends on the OL because $expr'$
is defined for each OL\@. 
Induction and recursion are also defined directly on type $expr'$.  To
obtain a recursion principle, it is shown that for any type $t$, a
function $\ikw{f}$ of type $expr'\to t$ can be defined by
specifying its results on each ``constructor'' of the OL\@. For example,
for the $\lapp$ and $\ikw{fun}$ cases of Mini-ML, defining
$\ikw{f}$ involves defining $\mathsf{Happ}$ and $\mathsf{Hfun}$ of
the following types:
$$\begin{array}{l}
\mathsf{Happ}: expr' \to expr' \to B \to B
  \to B\\
\mathsf{Hfun}: (expr' \to expr') \to B \to B
\end{array}$$
and then the following reduction equations hold.
\begin{eqnarray*}
\ikw{f}\,(e_1 \lapp e_2)& =&
   \mathsf{Happ}\,e_1\,e_2\,(\ikw{f}\,e_1)\,(\ikw{f}\,e_2)\\
\ikw{f}\,(\ikw{fun}~\lambda x.\ fx)& = &
    \mathsf{Hfun}\,(\ikw{canon}(\lambda x.\ f\,x))\,
        (\ikw{f}\,(\ikw{lbind}~(\lambda x.\ f\,x)))
\end{eqnarray*}
In these equations we oversimplify, showing functions $\ikw{f}$,
$\mathsf{Happ}$, and $\mathsf{Hfun}$ applied to terms of type $expr$;
in the actual equations, proofs paired with terms on the left are used
to build proofs of terms appearing on the right.  The \ikw{canon}
function in the equation for $\ikw{fun}$ uses another substitution
operator to obtain a ``canonical form,'' computed by replacing de
Bruijn index $0$ in $(\ikw{lbind}~(\lambda x.\ f\,x))$ with $x$.  This
function is the identity function on terms that satisfy \ikw{abstr}.

Another version of constructive \hybrid~\cite{CapFel06LC} in Coq has
been proposed, in which theorems such as induction and recursion
principles are proved once at a general level, and then can be applied
directly to each OL\@.  An OL is specified by a \emph{signature}, which
can include sets of sorts, operation names, and even built-in typing
rules.  A signature specifies the binding structure of the operators,
and the recursion and induction principles are formulated directly on
the higher-order syntax.

\paragraph{\hybrid $0.2$.}
\label{pa:ho2}

During the
write-up of this report, the infrastructure of
\hybrid has developed significantly, thanks to the work by Alan
Martin (see \cite{MMF07}), so that we informally talk of \hybrid $0.2$.
Because those changes have been recent and only relatively
influence the two-level approach, we have decided not to update the
whole paper, but mention here the relevant differences.

The main improvement concerns an overall reorganization of the
infrastructure described in Section~\ref{sec:introh}, based on the
internalization as a type of the set of \emph{proper} terms.  Using
\HOL's \emph{typedef} mechanism, the type $\propert$ is defined
as a \emph{bijective image} of the set \( \{\, s \oftype \expr
\mid \level{0}{s} \,\} \), with inverse bijections \( \mathsf{expr}
\oftype \propert \fs \expr \) and \( \mathsf{prpr} \oftype
\expr \fs \propert \).  In effect, \emph{typedef} makes
$\propert$ a subtype of $\expr$, but since \HOL's type
system does not have subtyping, the conversion function must be
explicit.  Now that OL terms can only be well-formed de Bruijn terms,
we can replace the \emph{proper\_abst} property
(MC-Lemma~\ref{mclem:proper_abst}) with the new lemma
\begin{mclemma}[abstr\_const]
\begin{eqnarray*}
    \abstr{(\ilam{\ivar{}}{t :: \propert})}
  \end{eqnarray*}
\end{mclemma}

From the standpoint of two-level reasoning this lemma allows us to
dispose of all \ikw{proper} assumptions: in particular the SL
universal quantification has type $(\propert\fs\oo)\fs\oo$ and the
relative SL clause (Figure~\ref{fig:nseq}) becomes:
\begin{eqnarray*}
  \prems{\forall x.\ \slvdn{\Gamma}{n}{(G\ x)}  } & \Implies &
\slvdn{\Gamma}{n+1}{(\slAll{x}{G\ x})}\
\end{eqnarray*}
Therefore, in the proof of MC-Lemma~\ref{le:sl-ex} no \ikw{proper}
assumptions are generated. The proof of OL Subject Reduction
(MC-Theorem~\ref{mcthm:olsr}) does not need to appeal to
property~(\ref{eq:ol-prop}) or, more importantly, to part 1 of
MC-Lemma~\ref{mclem:evalproper}.
While this is helpful, it does not eliminate the need for adding
well-formedness annotations in OL judgments for the sake of
establishing  adequacy of the encoding.

Further, a structural definition of abstraction allows us to state the
crucial quasi-injectivity property of the \hybrid binder $\ikw{LAM}$, 
strengthening MC-Theorem~\ref{thm:inj} by requiring only one of $e$
and~$f$ to satisfy this condition (instead of both), thus simplifying the elimination rules for inductively defined OL
judgments:
\begin{goal}[strong\_lambda\_inject]
\[\abstr{e} \Implies (\LAM{x}{e\,x} = \LAM{y}{f\,y}) = (e = f) \]  
\end{goal}
The new definition allows us to drop \ikw{abstr\_tac} for plain
\HOL simplification, and the same applies, \afortiori to
\ikw{proper\_tac}.

A significant case study using this infrastructure has being tackled
by Alan Martin \cite{MartinCase,MartinPhD2010} and consists of an
investigation of the meta-theory of a functional programming language
with references using a variety of approaches, culminating with the
usage of a linearly ordered SL\@.  This study extends the work in
Section~\ref{sec:olli} and~\cite{MomiglianoP03}, as well as offering a
different encoding of Mini-ML with references than the one analyzed
with a linear logical framework~\cite{CervesatoP02}.\footnote{We
remark that this approach seems to be exempt from the problems
connected to verifying meta-theoretical \emph{sub-structural}
properties in LF-style, as pointed out in \cite{Reed08}.}

Martin's forthcoming doctoral thesis~\cite{MartinPhD2010} also
illustrates that it is possible to use alternate techniques for
induction at the SL level.  Instead of natural number induction, some
proofs of the case study are carried out by structural induction on
the definition of the SL\@.  In these proofs, it was necessary to
strengthen the desired properties to properties of arbitrary sequents,
and to define specialized weakening operators for contexts along with
lemmas supporting reasoning in such contexts.  It is not clear how
well this technique generalizes; this is the subject of future work.
In another technique, natural numbers are replaced by ordinals in the
definition of the SL, and natural number induction is replaced by
transfinite induction.  This technique is quite general and simplifies
proofs by induction that involve relating the proof height of one
derivation in the SL to one or more others.

\paragraph{Induction over Open Terms}

In this paper's examples, proofs by induction over derivations were
always on closed judgment such as evaluation, be it encoded as a
direct inductive definition at the meta-level or as \ikw{prog} clauses
used by the SL\@.  In both cases, this judgment was encoded without
the use of hypothetical and parametric judgments, and thus induction
was over \emph{closed} terms, although we essentially used case
analysis on open terms.  Inducting over open terms and hypothetical
judgments is a challenge that has required major theoretical
work~\cite{S00,gacek08lics}. Statements have to be generalized to
non-empty contexts, and these contexts have to be of a certain form,
which must enforce the property in question.  In~\cite{FeltyM09} we
showed how to accomplish this in \hybrid with only a surprisingly
minimal amount of additional infrastructure: we can use the
$\ikw{VAR}$ constructor to encode free variables of OLs, and simply
add a definition (\ikw{newvar}) that provides the capability of
creating a variable which is \emph{fresh}, in particular w.r.t.\ a
context.  We express the induction hypothesis as a ``context
invariant,'' which is a property that must be preserved when adding a
fresh variable to the context.  The general infrastructure we
build is designed so that it is straightforward to express context
invariants and prove that they are preserved when adding a fresh
variable.  Very little overhead is required, namely a small library of
simple lemmas, where no reasoning about substitution or
$\alpha$-conversion is needed as in first-order approaches.  Yet the
reasoning power of the system and the class of properties that can be
proved is significantly increased.

\section{Conclusions and Future Work}
\label{sec:future}

We have presented a multi-level architecture that allows reasoning
about objects encoded using HOAS in well-known systems such as \HOL
and Coq that implement well-understood logics.  The support for
reasoning includes induction and co-induction as well as various forms
of automation available in such systems such as tactical-style
reasoning and decision procedures.  We have presented several examples
of its use, including an arguably innovative case study.  As we have
demonstrated, there are a variety of advantages of this kind of
approach:
\begin{itemize}
\item It is possible to replicate in a well-understood and interactive
  setting the style of proof used in systems such as $\foldna$
  designed specially for reasoning using higher-order encodings.  The
  reasoning can be done in such a way that theorems such as subject
  reduction proofs are proven without ``technical'' lemmas foreign to
  the mathematics of the problem.
\item Results about the intermediate layer of specification logics,
  such as cut elimination, are proven once and for all; in fact it is
  possible to work with different specification logics without
  changing the infrastructure.
\item It is possible to use this architecture as a way of ``fast
  prototyping'' HOAS logical frameworks since we can quickly implement
  and experiment with a potentially interesting SL, rather than
  building a new system from scratch.
\end{itemize}

Since our architecture is based on a very small set of theories that
definitionally builds an HOAS meta-language on top of a standard
proof-assistant, this allows us to do without any axiomatic
assumptions, in particular freeness of HOAS constructors and
extensionality properties at higher-order types, which in our setting are
now theorems.  Furthermore, we have shown that mixing of meta-level
and OL specifications make proofs more easily mechanizable.  Finally,
by the simple reason that the \hybrid system sits on top of \HOL or
Coq, we benefit from the higher degree of automation of the latter.

\smallskip

Some of our current and future work will concentrate on the practical
side, such as continuing the development and the testing of the new
infrastructure to which we have referred as {\hybrid $0.2$} (see
Section~\ref{pa:ho2} and \cite{MMF07}), especially to exploit the new
features offered by \HOL 2010.
Further, we envisage developing a package similar in spirit to Urban's
nominal datatype package for \HOL \cite{Nominal}.  For \hybrid, such a
package would automatically supply a variety of support from a user
specification of an OL, such as validity predicates like \ikw{isterm},
a series of theorems expressing freeness of the constructors of such a
type including injectivity and clash theorems, and an induction
principle on the shape of expressions analogous to
MC-Theorem~\ref{thm:proper-induct}.  To work at two levels, such a
package would include a number of pre-compiled SLs (including
cut-elimination proofs and other properties) as well as some
lightweight tactics to help with two-level inference. Ideally, the
output of the package could be in itself generated by a tool such as
\emph{OTT} (\cite{OTTpaper}) so as to exploit the tool's capabilities
of supporting work on large programming language definitions, where
``the scale makes it hard to keep a definition internally consistent,
and hard to keep a tight correspondence between a definition and
implementations'', \opcit.

\smallskip

We clearly need to explore how general our techniques for induction
over open terms \cite{FeltyM09} are, both by attempting other typical
case studies such as the POPLMark challenge or the Church-Rosser
theorem, as well as analyzing the relationship with theoretical
counterpart such as the regular world assumptions and context
invariants in Abella.  This may also have the benefit of a better
understanding and ``popularization'' of proofs in those less known
frameworks.  In Twelf, in particular, much of the work in  constructing
proofs is currently handled by an external check for properties such as
termination and coverage~\cite{SchurmannP03,Pientka05}.  We are
investigating \hybrid as the target of a sort of ``compilation'' of
such proofs into the well-understood higher-order logic of \HOL.  
More in-depth comparisons with nominal logic ideas such as freshness
and the Gabbay-Pitts quantifier are also in order. In fact, any
concrete representation of bound variables does not fit well with
HOAS, where the former have no independent identities.  However, there
are relevant applications (\eg, mismatch in the $\pi$-calculus, see
\cite{CheneyNLF} for other examples) where names of bound variables do
matter.

\begin{acknowledgements}
  Most of the material in this paper is based on previous joint work
  with Simon Ambler and Roy
  Crole~\cite{MomTP01,Ambler02,Momigliano02lfm,ACM03prim,Momigliano03fos},
  Jeff Polakow~\cite{MomiglianoP03} and Venanzio
  Capretta~\cite{CapFel07}, whose contributions we gratefully
  acknowledge.  The paper has also benefited from discussions with
  Andy Gordon, Alan Martin, Marino Miculan, Dale Miller, Brigitte
  Pientka, Randy Pollack, Frank Pfenning and Carsten Sch\"{u}rman. We
  thank the anonymous reviewers for many useful suggestions.
\end{acknowledgements}

\bibliographystyle{plain}

\begin{thebibliography}{100}

\bibitem{Ong}
Samson Abramsky and C.-H.~Luke Ong.
\newblock Full abstraction in the lazy lambda calculus.
\newblock {\em Inf. Comput.}, 105(2):159--267, 1993.

\bibitem{ACM03prim}
S.~J. Ambler, R.~L. Crole, and Alberto Momigliano.
\newblock A definitional approach to primitive recursion over higher order
  abstract syntax.
\newblock In {\em {MER$\lambda$IN '03: Proceedings of the 2003 ACM SIGPLAN
  workshop on MEchanized Reasoning about Languages with varIable biNding}},
  pages 1--11, New York, NY, USA, 2003. ACM Press.

\bibitem{Ambler02}
Simon Ambler, Roy~L. Crole, and Alberto Momigliano.
\newblock Combining higher order abstract syntax with tactical theorem proving
  and (co)induction.
\newblock In Carre{\~n}o et~al. \cite{DBLP:conf/tphol/2002}, pages 13--30.

\bibitem{Aydemir07}
Brian Aydemir, Aaron Bohannon, and Stephanie Weirich.
\newblock Nominal reasoning techniques in {Coq}.
\newblock {\em Electron. Notes Theor. Comput. Sci.}, 174(5):69--77, 2007.

\bibitem{ACPPW07}
Brian Aydemir, Arthur Chargu\'{e}raud, Benjamin~C. Pierce, Randy Pollack, and
  Stephanie Weirich.
\newblock Engineering formal metatheory.
\newblock {\em SIGPLAN Not.}, 43(1):3--15, 2008.

\bibitem{poplmark2005}
Brian~E. Aydemir, Aaron Bohannon, Matthew Fairbairn, J.~Nathan Foster,
  Benjamin~C. Pierce, Peter Sewell, Dimitrios Vytiniotis, Geoffrey Washburn,
  Stephanie Weirich, and Steve Zdancewic.
\newblock Mechanized metatheory for the masses: the {\textsc{poplmark}}
  challenge.
\newblock In Joe Hurd and T.~Melham, editors, {\em Theorem Proving in Higher
  Order Logics, 18th International Conference}, Lecture Notes in Computer
  Science, pages 50--65. Springer, 2005.

\bibitem{Bedwyr}
David Baelde, Andrew Gacek, Dale Miller, Gopalan Nadathur, and Alwen Tiu.
\newblock The {B}edwyr system for model checking over syntactic expressions.
\newblock In Frank Pfenning, editor, {\em CADE}, volume 4603 of {\em Lecture
  Notes in Computer Science}, pages 391--397. Springer, 2007.

\bibitem{Ballarin03}
Clemens Ballarin.
\newblock Locales and locale expressions in {I}sabelle/{I}sar.
\newblock In Berardi et~al. \cite{DBLP:conf/types/2003}, pages 34--50.

\bibitem{MIL-LITE}
Nick Benton and Andrew Kennedy.
\newblock Monads, effects and transformations.
\newblock {\em Electr. Notes Theor. Comput. Sci.}, 26, 1999.

\bibitem{MLJ}
Nick Benton, Andrew Kennedy, and George Russell.
\newblock Compiling standard {ML} to {J}ava bytecodes.
\newblock In {\em ICFP 1998}, pages 129--140, 1998.

\bibitem{DBLP:conf/types/2003}
Stefano Berardi, Mario Coppo, and Ferruccio Damiani, editors.
\newblock {\em Types for Proofs and Programs, International Workshop, {TYPES}
  2003, Torino, Italy, April 30 - May 4, 2003, Revised Selected Papers}, volume
  3085 of {\em Lecture Notes in Computer Science}. Springer, 2004.

\bibitem{BerghoferN-TPHOLs00}
Stefan Berghofer and Tobias Nipkow.
\newblock Proof terms for simply typed higher order logic.
\newblock In J.~Harrison and M.~Aagaard, editors, {\em Theorem Proving in
  Higher Order Logics}, volume 1869 of {\em LNCS}, pages 38--52. Springer,
  2000.

\bibitem{BerghoferU07}
Stefan Berghofer and Christian Urban.
\newblock A head-to-head comparison of de {B}ruijn indices and names.
\newblock {\em Electr. Notes Theor. Comput. Sci.}, 174(5):53--67, 2007.

\bibitem{bertot/casteran:2004}
Yves Bertot and Pierre Cast\'{e}ran.
\newblock {\em Interactive Theorem Proving and Program Development. Coq'Art:
  The Calculus of Inductive Constructions}.
\newblock Springer, 2004.

\bibitem{metalp82}
K.~A. Bowen and R.~A. Kowalski.
\newblock Amalgamating language and metalanguage in logic programming.
\newblock In K.~L. Clark and S.~A. Tarnlund, editors, {\em Logic programming,
  vol 16 of APIC studies in data processing}, pages 153--172. Academic Press,
  1982.

\bibitem{BucaloHMSH06}
Anna Bucalo, Furio Honsell, Marino Miculan, Ivan Scagnetto, and Martin Hoffman.
\newblock Consistency of the theory of contexts.
\newblock {\em J. Funct. Program.}, 16(3):327--372, 2006.

\bibitem{CapFel06LC}
Venanzio Capretta and Amy Felty.
\newblock Higher order abstract syntax in type theory.
\newblock http://www.cs.ru.nl/~venanzio/publications/HOUA.pdf, 2006.

\bibitem{CapFel07}
Venanzio Capretta and Amy~P. Felty.
\newblock Combining de {B}ruijn indices and higher-order abstract syntax in
  {C}oq.
\newblock In Thorsten Altenkirch and Conor McBride, editors, {\em TYPES},
  volume 4502 of {\em Lecture Notes in Computer Science}, pages 63--77.
  Springer, 2006.

\bibitem{DBLP:conf/tphol/2002}
Victor Carre{\~n}o, C{\'e}sar Mu{\~n}oz, and Sofi{\`e}ne Tashar, editors.
\newblock {\em Theorem Proving in Higher Order Logics, 15th International
  Conference, {TPHOL}s 2002, Hampton, {VA}, {USA}, August 20-23, 2002,
  Proceedings}, volume 2410 of {\em Lecture Notes in Computer Science}.
  Springer, 2002.

\bibitem{CervesatoP02}
Iliano Cervesato and Frank Pfenning.
\newblock A linear logical framework.
\newblock {\em Inf. Comput.}, 179(1):19--75, 2002.

\bibitem{Cheney05}
James Cheney.
\newblock A simpler proof theory for nominal logic.
\newblock In Vladimiro Sassone, editor, {\em FoSSaCS}, volume 3441 of {\em
  Lecture Notes in Computer Science}, pages 379--394. Springer, 2005.

\bibitem{CheneyNLF}
James Cheney.
\newblock A simple nominal type theory.
\newblock {\em Electr. Notes Theor. Comput. Sci.}, 228:37--52, 2009.

\bibitem{ChlipalaICFP08}
Adam Chlipala.
\newblock Parametric higher-order abstract syntax for mechanized semantics.
\newblock In {\em 13th ACM SIGPLAN International Conference on Functional
  Programming}, September 2008.

\bibitem{Church40}
Alonzo Church.
\newblock A formulation of the simple theory of types.
\newblock {\em Journal of Symbolic Logic}, 5:56--68, 1940.

\bibitem{CiaffaglioneLM07}
Alberto Ciaffaglione, Luigi Liquori, and Marino Miculan.
\newblock Reasoning about object-based calculi in (co)inductive type theory and
  the theory of contexts.
\newblock {\em J. Autom. Reasoning}, 39(1):1--47, 2007.

\bibitem{Clement86}
D.~Clement, J.~Despeyroux, T.~Despeyroux, and G.~Kahn.
\newblock A simple applicative language: Mini-{ML}.
\newblock In {\em Proceedings of the 1986 {ACM} Conference on Lisp and
  Functional Programming}, pages 13--27. ACM, ACM, August 1986.

\bibitem{CraryS03}
Karl Crary and Susmit Sarkar.
\newblock Foundational certified code in a metalogical framework.
\newblock In Franz Baader, editor, {\em CADE}, volume 2741 of {\em Lecture
  Notes in Computer Science}, pages 106--120. Springer, 2003.

\bibitem{CroleHA}
Roy Crole.
\newblock Hybrid adequacy.
\newblock Technical Report CS-06-011, School of Mathematics and Computer
  Sience, University of Leicester, UK, November 2006.

\bibitem{CuiDX05}
Sa~Cui, Kevin Donnelly, and Hongwei Xi.
\newblock {ATS}: {A} language that combines programming with theorem proving.
\newblock In Bernhard Gramlich, editor, {\em FroCos}, volume 3717 of {\em
  Lecture Notes in Computer Science}, pages 310--320. Springer, 2005.

\bibitem{Danvy99hoots}
Olivier Danvy, Belmina Dzafic, and Frank Pfenning.
\newblock On proving syntactic properties of {CPS} programs.
\newblock In Andrew Gordon and Andrew Pitts, editors, {\em Proceedings of
  HOOTS'99}, Paris, September 1999.
\newblock Electronic Notes in Theoretical Computer Science, Volume 26.

\bibitem{DeBruijn91lf}
N.~G. de~Bruijn.
\newblock A plea for weaker frameworks.
\newblock In G.~Huet and G.~Plotkin, editors, {\em Logical Frameworks}, pages
  40--67. Cambridge University Press, 1991.

\bibitem{DFHtlca95}
Jo{\"e}lle Despeyroux, Amy Felty, and Andr{\'e} Hirschowitz.
\newblock Higher-order abstract syntax in {Coq}.
\newblock In {\em Second International Conference on Typed Lambda Calculi and
  Applications}, pages 124--138. Springer, {\em Lecture Notes in Computer
  Science}, April 1995.

\bibitem{Despeyroux00jflp}
Joelle Despeyroux and Pierre Leleu.
\newblock Metatheoretic results for a modal $\lambda$-calculus.
\newblock {\em Journal of Functional and Logic Programming}, 2000(1), 2000.

\bibitem{Eriksson94}
Lars-Henrik Eriksson.
\newblock Pi: an interactive derivation editor for the calculus of partial
  inductive definitions.
\newblock In Alan Bundy, editor, {\em CADE}, volume 814 of {\em Lecture Notes
  in Computer Science}, pages 821--825. Springer, 1994.

\bibitem{FeltyPientka:ITP10}
Amy Felty and Brigitte Pientka.
\newblock Reasoning with higher-order abstract syntax and contexts: {A}
  comparison.
\newblock In M.~Kaufmann and L.~Paulson, editors, {\em International Conference
  on Interactive Theorem Proving}, volume 6172 of {\em Lecture Notes in
  Computer Science}, pages 228--243. Springer, 2010.

\bibitem{Felty02}
Amy~P. Felty.
\newblock Two-level meta-reasoning in {C}oq.
\newblock In Carre{\~n}o et~al. \cite{DBLP:conf/tphol/2002}, pages 198--213.

\bibitem{FeltyM09}
Amy~P. Felty and Alberto Momigliano.
\newblock Reasoning with hypothetical judgments and open terms in hybrid.
\newblock In Ant{\'o}nio Porto and Francisco~Javier L{\'o}pez-Fraguas, editors,
  {\em PPDP}, pages 83--92. ACM, 2009.

\bibitem{FordM03}
Jonathan Ford and Ian~A. Mason.
\newblock Formal foundations of operational semantics.
\newblock {\em Higher-Order and Symbolic Computation}, 16(3):161--202, 2003.

\bibitem{PittsAM:newaas-jv}
M.~J. Gabbay and A.~M. Pitts.
\newblock A new approach to abstract syntax with variable binding.
\newblock {\em Formal Aspects of Computing}, 13:341--363, 2001.

\bibitem{GabbayC04}
Murdoch Gabbay and James Cheney.
\newblock A sequent calculus for nominal logic.
\newblock In {\em LICS}, pages 139--148. IEEE Computer Society, 2004.

\bibitem{GabbayMJ02tphol}
Murdoch~J. Gabbay.
\newblock Automating {F}raenkel-{M}ostowski syntax.
\newblock Technical Report CP-2002-211736, NASA, 2002.
\newblock Track B Proceedings of TPHOLs'02.

\bibitem{Abella}
Andrew Gacek.
\newblock The {Abella} interactive theorem prover (system description).
\newblock In Alessandro Armando, Peter Baumgartner, and Gilles Dowek, editors,
  {\em IJCAR}, volume 5195 of {\em Lecture Notes in Computer Science}, pages
  154--161. Springer, 2008.

\bibitem{gacek08lics}
Andrew Gacek, Dale Miller, and Gopalan Nadathur.
\newblock Combining generic judgments with recursive definitions.
\newblock In {\em LICS}, pages 33--44. IEEE Computer Society, 2008.

\bibitem{AbellaSOS}
Andrew Gacek, Dale Miller, and Gopalan Nadathur.
\newblock Reasoning in {Abella} about structural operational semantics
  specifications.
\newblock {\em Electr. Notes Theor. Comput. Sci.}, 228:85--100, 2009.

\bibitem{Gillard00}
Guillaume Gillard.
\newblock A formalization of a concurrent object calculus up to
  $\alpha$-conversion.
\newblock In David~A. McAllester, editor, {\em CADE}, volume 1831 of {\em
  Lecture Notes in Computer Science}, pages 417--432. Springer, 2000.

\bibitem{Gim98tr}
Eduardo Gimenez.
\newblock {A Tutorial on Recursive Types in Coq}.
\newblock Technical Report RT-0221, Inria, 1998.

\bibitem{Gor93}
Andrew Gordon.
\newblock A mechanisation of name-carrying syntax up to $\alpha$-conversion.
\newblock In {J.J. Joyce} and {C.-J.H. Seger}, editors, {\em International
  Workshop on Higher Order Logic Theorem Proving and its Applications}, volume
  780 of {\em Lecture Notes in Computer Science}, pages 414--427, Vancouver,
  Canada, August 1994. University of British Columbia, Springer.

\bibitem{Gordon96}
Andrew~D. Gordon and Tom Melham.
\newblock Five axioms of $\alpha$-conversion.
\newblock In J.~von Wright, J.~Grundy, and J.~Harrison, editors, {\em
  Proceedings of the 9th International Conference on Theorem Proving in Higher
  Order Logics (TPHOLs'96)}, pages 173--191, Turku, Finland, August 1996.
  Springer-Verlag LNCS 1125.

\bibitem{GunterwhyMLnot}
Elsa~L. Gunter.
\newblock {Why we can't have SML-style datatype declarations in HOL.}
\newblock In Luc J.~M. Claesen and Michael J.~C. Gordon, editors, {\em TPHOLs},
  volume A-20 of {\em IFIP Transactions}, pages 561--568.
  North-Holland/Elsevier, 1992.

\bibitem{Halnass91}
L.~Hallnas.
\newblock Partial inductive definitions.
\newblock {\em Theoretical Computer Science}, 87(1):115--147, July 1991.

\bibitem{Harper93jacm}
Robert Harper, Furio Honsell, and Gordon Plotkin.
\newblock A framework for defining logics.
\newblock {\em Journal of the Association for Computing Machinery},
  40(1):143--184, January 1993.

\bibitem{HickeyNYK06}
Jason Hickey, Aleksey Nogin, Xin Yu, and Alexei Kopylov.
\newblock {Mechanized meta-reasoning using a hybrid HOAS/de Bruijn
  representation and reflection.}
\newblock In John~H. Reppy and Julia~L. Lawall, editors, {\em ICFP 2006}, pages
  172--183. ACM Press, 2006.

\bibitem{Hill&Gallagher1998}
P.~M. Hill and J.~Gallagher.
\newblock Meta-programming in logic programming.
\newblock In Dov Gabbay, Christopher~J. Hogger, and J.~A. Robinson, editors,
  {\em Handbook of Logic in Artificial Intelligence and Logic Programming,
  Volume 5: Logic programming}, pages 421--498. Oxford University Press,
  Oxford, 1998.

\bibitem{Lolli}
Joshua~S. Hodas and Dale Miller.
\newblock Logic programming in a fragment of intuitionistic linear logic.
\newblock {\em Inf. Comput.}, 110(2):327--365, 1994.

\bibitem{HonsellMS01}
Furio Honsell, Marino Miculan, and Ivan Scagnetto.
\newblock An axiomatic approach to metareasoning on nominal algebras in {HOAS}.
\newblock In Fernando Orejas, Paul~G. Spirakis, and Jan van Leeuwen, editors,
  {\em ICALP}, volume 2076 of {\em Lecture Notes in Computer Science}, pages
  963--978. Springer, 2001.

\bibitem{Howe96ic}
Douglas~J. Howe.
\newblock Proving congruence of bisimulation in functional programming
  languages.
\newblock {\em Information and Computation}, 124(2):103--112, 1996.

\bibitem{Hybrid}
{Hybrid Group}.
\newblock {Hybrid}: {A} package for higher-order syntax in {Isabelle} and
  {Coq}.
\newblock \url{hybrid.dsi.unimi.it}, 2008.

\bibitem{ISAR}
{Isar Group}.
\newblock {Isar - Intelligible semi-automated reasoning}.
\newblock \url{http://isabelle.in.tum.de/Isar}, 2000. Accessed 13 May 2010.

\bibitem{JoMinLog}
I.~Johansson.
\newblock Der {Minimalkalkül}, ein reduzierter intuitionistischer
  {Formalismus}.
\newblock {\em Compositio Math}, 4:119--136, 1937.

\bibitem{Lassen06}
S{\o}ren~B. Lassen.
\newblock Head normal form bisimulation for pairs and the
  $\lambda\mu$-calculus.
\newblock In {\em LICS}, pages 297--306. IEEE Computer Society, 2006.

\bibitem{lee+:towards}
Daniel~K. Lee, Karl Crary, and Robert Harper.
\newblock Towards a mechanized metatheory of standard {ML}.
\newblock In {\em POPL '07: Proceedings of the 34th annual ACM SIGPLAN-SIGACT
  symposium on Principles of programming languages}, pages 173--184, New York,
  NY, USA, 2007. ACM Press.

\bibitem{LEGO}
{{LEGO} Group}.
\newblock {The {LEGO} Proof Assistant}.
\newblock \url{www.dcs.ed.ac.uk/home/lego/}, 2001. Accessed 18 May 2010.

\bibitem{acl2java}
Hanbing Liu and J.~Strother Moore.
\newblock Executable {JVM} model for analytical reasoning: {A} study.
\newblock {\em Sci. Comput. Program.}, 57(3):253--274, 2005.

\bibitem{MartinPhD2010}
Alan Martin.
\newblock {\em Higher-Order Abstract Syntax in Isabelle/{HOL}}.
\newblock PhD thesis, University of Ottawa, 2010.
\newblock forthcoming.

\bibitem{MartinCase}
Alan~J. Martin.
\newblock Case study: Subject reduction for {M}ini-{M}{L} with references, in
  {Isabelle/HOL + Hybrid}.
\newblock Workshop on Mechanizing Metatheory, \url{
  www.cis.upenn.edu/~sweirich/wmm/wmm08/martin.pdf}, Retrieved 1/7/2010,
  September 2008.

\bibitem{MartinLof85}
Per Martin-L{\"o}f.
\newblock On the meanings of the logical constants and the justifications of
  the logical laws.
\newblock {\em Nordic Journal of Philosophical Logic}, 1(1):11--60, 1996.

\bibitem{McCreight03}
Andrew McCreight and Carsten Sch{\"u}rmann.
\newblock A meta linear logical framework.
\newblock Informal Proceedings of LFM'04.

\bibitem{mcdowell00tcs}
Raymond McDowell and Dale Miller.
\newblock Cut-elimination for a logic with definitions and induction.
\newblock {\em Theoretical Computer Science}, 232:91--119, 2000.

\bibitem{McDowell01}
Raymond McDowell and Dale Miller.
\newblock Reasoning with higher-order abstract syntax in a logical framework.
\newblock {\em ACM Transactions on Computational Logic}, 3(1):80--136, January
  2002.

\bibitem{McKinna99}
James McKinna and Robert Pollack.
\newblock Some lambda calculus and type theory formalized.
\newblock {\em J. Autom. Reasoning}, 23(3-4):373--409, 1999.

\bibitem{Melham94}
Thomas~F. Melham.
\newblock A mechanized theory of the $\pi$-calculus in {HOL}.
\newblock {\em Nord. J. Comput.}, 1(1):50--76, 1994.

\bibitem{Miculan:IC01}
Marino Miculan.
\newblock On the formalization of the modal $\mu$-calculus in the calculus of
  inductive constructions.
\newblock {\em Information and Computation}, 164(1):199--231, 2001.

\bibitem{Forum}
Dale Miller.
\newblock Forum: {A} multiple-conclusion specification logic.
\newblock {\em Theor. Comput. Sci.}, 165(1):201--232, 1996.

\bibitem{miller04llcs}
Dale Miller.
\newblock Overview of linear logic programming.
\newblock In Thomas Ehrhard, Jean-Yves Girard, Paul Ruet, and Phil Scott,
  editors, {\em Linear Logic in Computer Science}, volume 316 of {\em London
  Mathematical Society Lecture Note}, pages 119--150. Cambridge University
  Press, 2004.

\bibitem{Miller91apal}
Dale Miller, Gopalan Nadathur, Frank Pfenning, and Andre Scedrov.
\newblock Uniform proofs as a foundation for logic programming.
\newblock {\em Annals of Pure and Applied Logic}, 51:125--157, 1991.

\bibitem{miller05tocl}
Dale Miller and Alwen Tiu.
\newblock A proof theory for generic judgments.
\newblock {\em ACM Trans. Comput. Logic}, 6(4):749--783, 2005.

\bibitem{Momigliano03fos}
Alberto Momigliano and Simon Ambler.
\newblock Multi-level meta-reasoning with higher order abstract syntax.
\newblock In A.~Gordon, editor, {\em FOSSACS'03}, volume 2620 of {\em LNCS},
  pages 375--392. Springer Verlag, 2003.

\bibitem{MomTP01}
Alberto Momigliano, Simon Ambler, and Roy Crole.
\newblock A comparison of formalisations of the meta-theory of a language with
  variable binding in {I}sabelle.
\newblock In R.~J. Boulton and P.~Jackson, editors, {\em 14th International
  Conference on Theorem Proving in Higher Order Logics (TPHOLs01), Supplemental
  Proceedings}, pages 267--282. Informatics Research Report EDI-INF-RR-01-23,
  2001.

\bibitem{Momigliano02lfm}
Alberto Momigliano, Simon Ambler, and Roy~L. Crole.
\newblock A {H}ybrid encoding of {H}owe's method for establishing congruence of
  bisimilarity.
\newblock {\em Electr. Notes Theor. Comput. Sci.}, 70(2), 2002.

\bibitem{MMF07}
Alberto Momigliano, Alan~J. Martin, and Amy~P. Felty.
\newblock Two-level {H}ybrid: {A} system for reasoning using higher-order
  abstract syntax.
\newblock {\em Electr. Notes Theor. Comput. Sci.}, 196:85--93, 2008.

\bibitem{MomiglianoP03}
Alberto Momigliano and Jeff Polakow.
\newblock A formalization of an ordered logical framework in {H}ybrid with
  applications to continuation machines.
\newblock In {\em MERLIN}. ACM, 2003.

\bibitem{MomiglianoT03}
Alberto Momigliano and Alwen~Fernanto Tiu.
\newblock Induction and co-induction in sequent calculus.
\newblock In Berardi et~al. \cite{DBLP:conf/types/2003}, pages 293--308.

\bibitem{NanevskiTOCL}
Aleksandar Nanevski, Frank Pfenning, and Brigitte Pientka.
\newblock Contextual modal type theory.
\newblock {\em ACM Trans. Comput. Log.}, 9(3), 2008.

\bibitem{Nipkow-Paulson-Wenzel:2002}
Tobias Nipkow, Lawrence~C. Paulson, and Markus Wenzel.
\newblock {\em Isabelle/{HOL}: {A} Proof Assistant for Higher-Order Logic},
  volume 2283 of {\em Lecture Notes in Computer Science}.
\newblock Springer Verlag, 2002.

\bibitem{Nominal}
{Nominal Methods Group}.
\newblock Nominal {Isabelle}.
\newblock \url{isabelle.in.tum.de/nominal}, 2008. Accessed 15 May 2010.

\bibitem{Norrish:TPHOLs04}
Michael Norrish.
\newblock Recursive function definition for types with binders.
\newblock In {\em Seventeenth International Conference on Theorem Proving in
  Higher Order Logics}, pages 241--256. Springer-Verlag Lecture Notes in
  Computer Science, 2004.

\bibitem{cade92-pvs}
{S.} Owre, {J.}~{M.} Rushby, and {N.} Shankar.
\newblock {PVS:} {A} prototype verification system.
\newblock In Deepak Kapur, editor, {\em Proceedings of the 11th International
  Conference on Automated Deduction}, pages 748--752. Springer-Verlag LNAI 607,
  1992.

\bibitem{PaulinMohring93}
Christine Paulin-Mohring.
\newblock Inductive definitions in the system {Coq}: Rules and properties.
\newblock In M.~Bezem and J.~F. Groote, editors, {\em Proceedings of the
  International Conference on Typed Lambda Calculi and Applications}, pages
  328--345, Utrecht, The Netherlands, March 1993. Springer-Verlag LNCS 664.

\bibitem{Paulson94cade}
Lawrence~C. Paulson.
\newblock A fixedpoint approach to implementing (co)inductive definitions.
\newblock In Alan Bundy, editor, {\em Proceedings of the 12th International
  Conference on Automated Deduction}, pages 148--161, Nancy, France, June 1994.
  Springer-Verlag LNAI 814.

\bibitem{Pfenning99handbook}
Frank Pfenning.
\newblock Logical frameworks.
\newblock In Alan Robinson and Andrei Voronkov, editors, {\em Handbook of
  Automated Reasoning}. Elsevier Science Publishers, 1999.

\bibitem{Pfenning01book}
Frank Pfenning.
\newblock {\em Computation and Deduction}.
\newblock Cambridge University Press, Draft from March 2001 available at
  \url{www.cs.cmu.edu/~fp/courses/comp-ded/handouts/cd.pdf}.
\newblock Accessed 30 April 2010.

\bibitem{Pientka05}
Brigitte Pientka.
\newblock Verifying termination and reduction properties about higher-order
  logic programs.
\newblock {\em J. Autom. Reasoning}, 34(2):179--207, 2005.

\bibitem{Pientka10}
Brigitte Pientka.
\newblock Beluga: Programming with dependent types, contextual data, and
  contexts.
\newblock In Matthias Blume, Naoki Kobayashi, and Germ{\'a}n Vidal, editors,
  {\em FLOPS}, volume 6009 of {\em Lecture Notes in Computer Science}, pages
  1--12. Springer, 2010.

\bibitem{pitts03ic}
Andrew~M. Pitts.
\newblock Nominal logic, {A} first order theory of names and binding.
\newblock {\em Information and Computation}, 186(2):165--193, 2003.

\bibitem{Pitts06}
Andrew~M. Pitts.
\newblock Alpha-structural recursion and induction.
\newblock {\em J. ACM}, 53(3):459--506, 2006.

\bibitem{Polakow01phd}
Jeff Polakow.
\newblock {\em Ordered Linear Logic and Applications}.
\newblock PhD thesis, CMU, 2001.

\bibitem{Polakow06}
Jeff Polakow.
\newblock Linearity constraints as bounded intervals in linear logic
  programming.
\newblock {\em J. Log. Comput.}, 16(1):135--155, 2006.

\bibitem{Polakow99mfps}
Jeff Polakow and Frank Pfenning.
\newblock Relating natural deduction and sequent calculus for intuitionistic
  non-commutative linear logic.
\newblock In Andre Scedrov and Achim Jung, editors, {\em Proceedings of the
  15th Conference on Mathematical Foundations of Programming Semantics}, New
  Orleans, Louisiana, April 1999.
\newblock {\em Electronic Notes in Theoretical Computer Science, Volume 20}.

\bibitem{Polakow00lfm}
Jeff Polakow and Frank Pfenning.
\newblock Properties of terms in continuation-passing style in an ordered
  logical framework.
\newblock In Jo{\"e}lle Despeyroux, editor, {\em 2nd Workshop on Logical
  Frameworks and Meta-languages (LFM'00)}, Santa Barbara, California, June
  2000.
\newblock Proceedings available as INRIA Technical Report.

\bibitem{PolakowYi01flops}
Jeff Polakow and Kwangkeun Yi.
\newblock Proving syntactic properties of exceptions in an ordered logical
  framework.
\newblock In Herbert Kuchen and Kazunori Ueda, editors, {\em Proceedings of the
  5th International Symposium on Functional and Logic Programming (FLOPS'01)},
  pages 61--77, Tokyo, Japan, March 2001. Springer-Verlag LNCS 2024.

\bibitem{PosSch08}
Adam Poswolsky and Carsten Sch{\"u}rmann.
\newblock Practical programming with higher-order encodings and dependent
  types.
\newblock In Sophia Drossopoulou, editor, {\em ESOP}, volume 4960 of {\em
  Lecture Notes in Computer Science}, pages 93--107. Springer, 2008.

\bibitem{Reed08}
Jason Reed.
\newblock Hybridizing a logical framework.
\newblock {\em Electron. Notes Theor. Comput. Sci.}, 174(6):135--148, 2007.

\bibitem{S00}
Carsten Sch{\"u}rmann.
\newblock {\em Automating the Meta-Theory of Deductive Systems}.
\newblock PhD thesis, Carnegie-Mellon University, 2000.
\newblock CMU-CS-00-146.

\bibitem{TwelfSP}
Carsten Sch{\"u}rmann.
\newblock The {Twelf} proof assistant.
\newblock In Stefan Berghofer, Tobias Nipkow, Christian Urban, and Makarius
  Wenzel, editors, {\em TPHOLs}, volume 5674 of {\em Lecture Notes in Computer
  Science}, pages 79--83. Springer, 2009.

\bibitem{SchurmannDP01}
Carsten Sch{\"u}rmann, Jo{\"e}lle Despeyroux, and Frank Pfenning.
\newblock Primitive recursion for higher-order abstract syntax.
\newblock {\em Theor. Comput. Sci.}, 266(1-2):1--57, 2001.

\bibitem{SchurmannP03}
Carsten Sch{\"u}rmann and Frank Pfenning.
\newblock A coverage checking algorithm for {LF}.
\newblock In David~A. Basin and Burkhart Wolff, editors, {\em TPHOLs}, volume
  2758 of {\em Lecture Notes in Computer Science}, pages 120--135. Springer,
  2003.

\bibitem{SPS:TLCA2005}
Carsten Sch{\"u}rmann, Adam Poswolsky, and Jeffrey Sarnat.
\newblock The $\bigtriangledown$-calculus. {F}unctional programming with
  higher-order encodings.
\newblock In {\em Seventh International Conference on Typed Lambda Calculi and
  Applications}, pages 339--353. Springer, {\em Lecture Notes in Computer
  Science}, April 2005.

\bibitem{OTTpaper}
Peter Sewell, Francesco~Zappa Nardelli, Scott Owens, Gilles Peskine, Tom Ridge,
  Susmit Sarkar, and Rok Strnisa.
\newblock Ott: effective tool support for the working semanticist.
\newblock In Ralf Hinze and Norman Ramsey, editors, {\em ICFP 2007}, pages
  1--12. ACM, 2007.

\bibitem{Tiu04phd}
Alwen Tiu.
\newblock {\em A Logical Framework for Reasoning about Logical Specifications}.
\newblock PhD thesis, Pennsylvania State University, May 2004.

\bibitem{Tiu07}
Alwen Tiu.
\newblock A logic for reasoning about generic judgments.
\newblock {\em Electr. Notes Theor. Comput. Sci.}, 174(5):3--18, 2007.

\bibitem{UrbanB06}
Christian Urban and Stefan Berghofer.
\newblock A recursion combinator for nominal datatypes implemented in
  {Isabelle/HOL}.
\newblock In Ulrich Furbach and Natarajan Shankar, editors, {\em IJCAR}, volume
  4130 of {\em Lecture Notes in Computer Science}, pages 498--512. Springer,
  2006.

\bibitem{urban05cade}
Christian Urban and Christine Tasson.
\newblock Nominal techniques in {Isabelle/HOL}.
\newblock In R.~Nieuwenhuis, editor, {\em Proceedings of the 20th International
  Conference on Automated Deduction (CADE)}, volume 3632 of {\em LNCS}, pages
  38--53. Springer, 2005.

\bibitem{Vestergaard01rta}
Ren{\'e} Vestergaard and James Brotherston.
\newblock A formalised first-order confluence proof for the $\lambda$-calculus
  using one-sorted variable names.
\newblock {\em Inf. Comput.}, 183(2):212--244, 2003.

\bibitem{WashburnWeirichJFP08}
Geoffrey Washburn and Stephanie Weirich.
\newblock Boxes go bananas: Encoding higher-order abstract syntax with
  parametric polymorphism.
\newblock {\em Journal of Functional Programming}, 18(1):87--140, January 2008.

\bibitem{clf}
Kevin Watkins, Iliano Cervesato, Frank Pfenning, and David Walker.
\newblock A concurrent logical framework: The propositional fragment.
\newblock In Berardi et~al. \cite{DBLP:conf/types/2003}, pages 355--377.

\end{thebibliography}

\end{document}